\documentclass[aos]{imsart}
\usepackage[utf8]{inputenc}
\usepackage{amsmath,amsthm,mathtools,amssymb,lscape,mathabx}
\usepackage{thmtools,thm-restate}
\usepackage{mathrsfs}
\usepackage{proba}
\usepackage{booktabs}
\usepackage{longtable}
\usepackage{threeparttable,subfigure}
\usepackage{enumerate}
\usepackage[shortlabels]{enumitem}
\usepackage{graphicx,multirow}
\usepackage{dsfont}
\usepackage{caption}
\usepackage{xcolor}
\usepackage{chngcntr}
\usepackage{apptools}
\AtAppendix{\counterwithin{lemma}{section}}
\usepackage[english]{babel}
\usepackage[pagewise]{lineno}
\usepackage{titletoc}

\usepackage{natbib}
\usepackage{hyperref}
\hypersetup{
    unicode=false,          % non-Latin characters in Acrobat’s bookmarks
    pdftoolbar=true,        % show Acrobat’s toolbar?
    pdfmenubar=true,        % show Acrobat’s menu?
    pdffitwindow=false,     % window fit to page when opened
    pdftitle={},    % title
    pdfauthor={},     % author
    pdfsubject={},   % subject of the document
    pdfcreator={},   % creator of the document
    pdfproducer={},  % producer of the document
    pdfkeywords={}, % list of keywords
    pdfnewwindow=true,      % links in new window
    colorlinks=true,       % false: boxed links; true: colored links
    linkcolor={red!50!black},          % color of internal links
    citecolor={blue!50!black},        % color of links to bibliography
    filecolor=magenta,      % color of file links
    urlcolor={blue!80!black}           % color of external links
}

\newtheorem{Def}{Definition}
\newtheorem{assumption}{Assumption}

\newtheorem{corollary}{Corollary}
\newtheorem{theorem}{Theorem}
\newtheorem{lemma}{Lemma}

\newtheorem{Rem}{Remark}
\newtheorem{Ex}{Example}

\numberwithin{equation}{section}

\def\m{\mathcal}
\def\b{\boldsymbol}
\def\vec{\mathsf{vec}\,}

\def\diag{\mathsf{diag}\,}

\def\C{\mathbb{C}}
\def\H{\mathcal{H}}
\def\E{\mathbb{E}}
\def\R{\mathbb{R}}

\def\1{\mathds{1}}
\def\S{\mathcal{S}}

\def\d{\mathsf{d}}
\def\co{\mathsf{co}}

\newcommand{\lsim}{\lesssim}
\newcommand{\viii}[1]{{\left\vert\kern-0.25ex\left\vert\kern-0.25ex\left\vert #1 
    \right\vert\kern-0.25ex\right\vert\kern-0.25ex\right\vert}}

\DeclareSymbolFont{sfoperators}{OT1}{cmss}{m}{n}
\DeclareSymbolFontAlphabet{\mathsf}{sfoperators}

\makeatletter
\def\operator@font{\mathgroup\symsfoperators}
\makeatother

\begin{document}

\begin{frontmatter}
\title{Bridging factor and sparse models}
%\title{A sample article title with some additional note\thanksref{t1}}
\runtitle{Bridging factor and sparse models}
%\thankstext{T1}{A sample additional note to the title.}

\begin{aug}
%%%%%%%%%%%%%%%%%%%%%%%%%%%%%%%%%%%%%%%%%%%%%%%
%% Only one address is permitted per author. %%
%% Only division, organization and e-mail is %%
%% included in the address.                  %%
%% Additional information can be included in %%
%% the Acknowledgments section if necessary. %%
%%%%%%%%%%%%%%%%%%%%%%%%%%%%%%%%%%%%%%%%%%%%%%%
\author[A]{\fnms{Jianqing} \snm{Fan}\thanks{Supported by ONR grant N00014-22-1-2340 and NSF grants DMS-2210833, DMS-2053832, and DMS-2052926.}\ead[label=e1]{jqfan@princeton.edu}},
\author[B]{\fnms{Ricardo P.} \snm{Masini}\ead[label=e2]{rmasini@princeton.edu}}
\and
\author[C]{\fnms{Marcelo C.} \snm{Medeiros}\ead[label=e3]{mcm@econ.puc-rio.br}}
%%%%%%%%%%%%%%%%%%%%%%%%%%%%%%%%%%%%%%%%%%%%%%
%% Addresses                                %%
%%%%%%%%%%%%%%%%%%%%%%%%%%%%%%%%%%%%%%%%%%%%%%
\address[A]{Department of Operations Research and Financial Engineering, Princeton University
\printead{e1}}

\address[B]{Center for Statistics and Machine Learning, Princeton University
\printead{e2}}

\address[C]{Department of Economics, Pontifical Catholic University of Rio de Janeiro
\printead{e3}}
\end{aug}

\begin{abstract}
Factor and sparse models are two widely used methods to impose a low-dimensional structure in high-dimensions. However, they are seemingly mutually exclusive. We propose a lifting method that combines the merits of these two models in a supervised learning methodology that allows for efficiently exploring all the information in high-dimensional datasets. The method is based on a flexible model for high-dimensional panel data, called factor-augmented regression model with observable and/or latent common factors, as well as idiosyncratic components. This model not only includes both principal component regression and sparse regression as specific models but also significantly weakens the cross-sectional dependence and facilitates model selection and interpretability. The method consists of several steps and a novel test for (partial) covariance structure in high dimensions to infer the remaining cross-section dependence at each step. We develop the theory for the model and demonstrate the validity of the multiplier bootstrap for testing a high-dimensional (partial) covariance structure. The theory is supported by a simulation study and applications.
\end{abstract}



\end{frontmatter}

\section{Introduction}
With the emergence of new and large datasets in almost all disciplines, the correct characterization of the dependence among variables is of substantial importance. Usually, to achieve this goal, the literature has followed two seemingly orthogonal tracks over the last two decades. On the one hand, factor models have become an essential tool to summarize information in large datasets under the assumption that the remaining dependence structure is negligible. For instance, panel factor models are now applied to various problems, ranging from forecasting to causal inference and network analysis. However, on the other hand, there have been significant advances in parameter estimation in ultra high-dimensions under the assumption of sparsity or weak sparsity. That is, a variable depends only on a (very) small subset of the other variables. For an overview of these two topics and their exciting developments, see \cite{FLZZ20}.

In this paper, we take an alternative route and combine the best of the two worlds described above to better characterize the dependence structure in high-dimensional datasets. More specifically, we consider that the covariance structure of a large set of variables, organized in a panel data format, is characterized as a combination of a factor structure, where factors can be either observed, unobserved, or both, and a weakly-sparse idiosyncratic component. This formulation is general enough to accommodate many data-generating processes of interest in economics, finance, epidemiology, and energy/engineering, for example. The proposed methodology has two ingredients: a multi-step estimation procedure and a new test for structure in high dimensional (partial) covariance matrices. The steps of the estimation procedure are as follows. In the first one, we take the original data and remove the effects of any observed factors. These factors can be deterministic terms such as seasonal dummies and/or trends or other observed covariates. The first step can be parametric or nonparametric, low or high dimensional. A latent factor model is then estimated using the residuals from the first stage. In a third step, we model the dependence among idiosyncratic terms as a weakly sparse regression estimated by the Least Absolute Shrinkage and Selection Operator (LASSO). A final additional fourth step can be used to build models for out-of-sample forecasting, taking into account each of the components in the previous steps. At each step, the null hypothesis of no remaining cross-section dependence can be tested by the proposed test for the (partial) covariance structure in high dimensions.

Our approach has many statistical applications. It can enhance high-dimensional prediction, select more interpretable variables, construct counterfactuals for policy evaluations, and depict (partial) correlation networks, among many others. 

\subsection{Motivation}\label{S:motivation}
Let $\b Y_t:=(Y_{1,t},\ldots, Y_{n,t})'$ be a random vector generated as $Y_{i,t}=\b\lambda_i'\b F_t+ U_{i,t}$, for $i\in[n]$, $t\in[T]$, where $\b\Sigma:=\b\E(\b U_t\b U_t')$, with $\b U_t:=(U_{1,t},\ldots, U_{n,t})'$, is not necessarily diagonal. Fix one component of interest $i\in[n]$, which serves as a response variable. Consider the following predictive models:
\begin{equation}\label{E:conditional_expectation_models}
\mathcal{M}_1:\,\E(Y_{i,t}|\b Y_{-i,t}),\quad
\mathcal{M}_2:\,\E(Y_{i,t}|\b F_t),\quad\textnormal{and}\quad
\mathcal{M}_3:\,\E(Y_{i,t}| \b F_t,\b U_{-i,t}),
\end{equation}
where $\b Y_{-i,t}$ and $\b U_{-i,t}$ are, respectively, vectors with the elements of $\b Y_t$ and $\b U_t$ without the $i$-th entry. $\mathcal{M}_3$ is a factor augmented regression since it is the same as $\E(Y_{i,t}| \b F_t,\b Y_{-i,t})$.

Suppose that we observe both $\b F_t$ and $\b U_{-i,t}$. Which one of the three models above is best in terms of mean square error ($\mathsf{MSE}$) for prediction? Comparison between $\mathcal{M}_1$ and $\mathcal{M}_2$ is not clear since it depends, among other things, on the magnitude of $\b \Sigma$ relative to $\b\Lambda'\b\Lambda$, where $\b\Lambda :=(\b\lambda_1,\ldots,\b\lambda_n)'$.
However, since the $\sigma$-algebras generated by $\b Y_{-i,t}$ and $\b F_t$ are both included in the $\sigma$-algebra generated by $(\b F_t,\b U_{-i,t})$, it is not surprising that $\mathsf{MSE}(\mathcal{M}_3)\leq\min[\mathsf{MSE}(\mathcal{M}_1),\mathsf{MSE}(\mathcal{M}_2)]$.
The same will hold true if we replace the models in \eqref{E:conditional_expectation_models} by their best linear projections, which we denote by $\mathcal{\widetilde{M}}_j$ for $j\in\{1,2,3\}$. Therefore, we can write the ``gains'' of  $\widetilde{\mathcal{M}}_3$ when compared to $\widetilde{\mathcal{M}}_1$ and $\widetilde{\mathcal{M}}_2$:
\begin{align*}
\mathsf{MSE}(\widetilde{\mathcal{M}}_3) &- \mathsf{MSE}(\widetilde{\mathcal{M}}_1) =- \b\vartheta_i'\b\Sigma_{-i,-i}\b\vartheta_i \\
\mathsf{MSE}(\widetilde{\mathcal{M}}_3) &- \mathsf{MSE}(\widetilde{\mathcal{M}}_2)=-\b\Delta_{1,i}\b\Delta_{1,i}' - \b\Delta_{2,i}'\b\Sigma_{-i,-i}\b\Delta_{2,i},
\end{align*}
where $\b\vartheta_i$ is coefficients of the projection of $U_{i,t}$ onto $\b U_{-i,t}$; $\b\Sigma_{-i,-i}$ is $\b\Sigma$ excluding the $i$-th row and column; $\b\Delta_{1,i} := \b\Lambda_i -\b\beta_i'\b\Lambda_{-i}$; $\b\Delta_{2,i} := \b\beta_i-\b\vartheta_i$; and $\b\beta_i$ is the coefficient of the projection of $Y_{i,t}$ onto $\b Y_{-i,t}$. From the previous expressions, it becomes evident that both $\widetilde{\mathcal{M}}_1$ and $\widetilde{\mathcal{M}}_2$ are restrictions on $\widetilde{\mathcal{M}}_3$. Broadly speaking, whenever one does not expect to have an \emph{exact} factor model, there are potential gains of taking into account the contribution of the idiosyncratic components $\b U_{-i,t}$. Therefore, we use $\widetilde{\mathcal{M}}_3$ as the base model for the estimation methodology described in Section \ref{S:methodology}.

\subsection{Main Contributions and Comparison with the Literature}
The contributions of this paper are multi-fold. First, our methodology bridges the gap between two apparently competing methods for high-dimensional modeling; see, for example, the discussion in \cite{illusion2017} and \cite{jFyKkW2020}. This yields a vast number of potential applications and spin-offs. For instance, in \cite{jFrMmM2020a}, we apply the methods developed here to evaluate the effects of interventions and contribute to the literature on synthetic controls and related methods by combining the approaches of \cite{lGtM2016}, and \cite{cCrMmM2018}. Therefore, in our setup, both a common strong factor structure and weak sparsity can coexist.

Second, the methodology proposed here contributes to the forecasting literature. For instance, in the second application considered in this paper, we build forecasting models based on a large cross-section of macroeconomic variables. We call this method the \texttt{FarmPredict}. We show that combining factors and a sparse regression strongly outperforms the traditional principal component regression as in \cite{jSmW2002b}. Therefore, \texttt{FarmPredict} can be an additional contribution to the forecasting and machine learning toolkit. The method can be easily extended to a multivariate setting combining factor-augmented vector autoregressions (FAVAR) as in \cite{bBjBpE2005} and sparse vector models as in \cite{aKlC2015} and \cite{rMmMeM2019}. Our methodology can also be applied in areas beyond economics. For example, it can be useful to construct forecasting models for the spread of infectious diseases, such as COVID-19, by taking into account co-movements and spillovers among different locations.

Third, we show the consistency of factor estimation based on the residuals of a first-step regression. Our results hold for both parametric (linear or nonlinear) and nonparametric first stage. A high-dimensional first stage is also allowed. Note that current results in the literature consider that factors are estimated based on observed data, and our derivations favor a much more flexible and general setup \citep{Bai2003,jBsN2002,jBsN2006}. More specifically, our methodology allows settings where there are both observed and latent factors and trend-stationary data. In the latter, the trend can be first removed by (nonparametric) first-stage regression.  Whenever the unobserved factors and the observed covariates are correlated, the method proposed in \cite{mhP2006} can be used, and all  results follow directly.

Fourth, we contribute to the LASSO literature. LASSO can not be model selection consistent for highly correlated variables. By decomposing covariates into factors and idiosyncratic components, namely the idea of lifting, we decorrelate the variables and make the model selection condition much easier to hold; see, for example, \cite{jFyKkW2020}. We show the consistency of the estimates based on the residuals of the previous steps. Our results are derived under restrictions on the population covariance matrix of the data and not on the estimated one, as is usual in many papers. See, for example, \cite{vandegeer2009}. Furthermore, we derive our results under much mild conditions than the ones considered in \cite{jFyKkW2020}.

Fifth, we extend the results in \cite{vCdCkK2013,CCK2018} to strong-mixing data in order to construct hypothesis tests for covariance and partial covariance structure in high dimensions.\footnote{\cite{alex2020} also extended \cite{vCdCkK2013}. However, their setup differs from ours as they only consider the case of independent and identically distributed data.} We also establish the consistency of a new estimator of the partial covariance matrix in high-dimensions and strong-mixing data. Our proposed tests can be used to infer if the (partial) covariance matrix of a high-dimensional random vector is diagonal or block-diagonal. More generally, we can test any pre-defined structure. Furthermore, we show that the test remains valid when we use the residuals from a previous step estimation to compute the covariance matrix. This result allows us to apply the test to the multi-stage estimation procedure proposed here. These are important results for many applications in economics, finance, epidemiology, and many other areas. For instance, our inference procedure can serve as a diagnostic and misspecification tool. For panel data models with interactive fixed effects as in \cite{mhP2006}, \cite{jB2009}, \cite{rMmW2015} and \cite{jByL2017}, our test can be directly applied to uncover the dependence structure among cross-sectional units before and after accounting for common factor components. If the factor structure is informative enough, we expect the idiosyncratic covariance matrix to be almost sparse. If this is not the case, we may have possibly underestimated the number of factors. One popular application is in asset pricing, as discussed in \cite{pGeOoS2019} and the empirical section of this paper. There are a huge number of proposed factors as described in \cite{gFsGdX2020}. We can apply our methodology to test for omitted factors and estimate network connections among firms, as in \cite{fxDkY2014} and \cite{cBgGgL2020}. Finally, as a diagnostic tool, our paper tackles the same problem as \cite{pGeOoS2019}. However, we take an alternative solution strategy that relies on a much different set of hypotheses; see also \cite{pGeOoS2020}.

Although our results are derived under the assumption that the number of factors is known, simulation results presented in the Supplementary Material provide evidence that the test has good finite-sample properties even when the number of factors is determined by data-driven methods commonly found in the literature.  In addition, due to factor augmentation, our method is robust to overestimating the number of factors. Over the past years, a vast number of papers proposed different methods to test for covariance structure in high dimensions. See, for example, \cite{cai2017} and the references therein. To the best of our knowledge, we complement all the previous papers by simultaneously considering high-dimensions, strong-mixing data with mild distributional assumptions, and pre-estimation when constructing tests for both covariance and partial covariance structure.

Finally, it is essential to highlight the theoretical challenges that we tackle in this paper. First, to derive the properties of our proposed test for (partial) covariance structure, we prove a new high dimensional Central Limit Theorem (CLT) for strong mixing sequences. To our knowledge, this CLT is new in the literature and allows us to apply Gaussian approximation results in a much more general framework. Second, as a side result, to develop the test, we first show the consistency of kernel-based estimation of a high-dimensional long-run covariance matrix of strong-mixing processes. This is also a new result with significant consequences for the theory of high-dimensional regression with dependent errors. Second, all the non-asymptotic bounds for our multi-step estimators are derived under the assumption that the distributions of the random variables in the model have polynomial tails, and the estimation errors of previous steps are also taken into account. This introduces several difficulties in proving the results but makes our method well-suited for applications with fat-tailed data, such as those observed in financial applications. Finally, in the derived bounds, the strong mixing coefficient appears explicit and can be allowed to grow with the sample size. This not only introduces technical challenges but makes our results very general.     

%\cite{oLmW2002}, \cite{chen2010}, \cite{aOmMmH2013}, \cite{cai2013}, \cite{wLyQ2014}, \cite{sZzChZrL2019}, \cite{cai2016}, \cite{zheng2019}, and \cite{xGcyT2020}, among many others.\footnote{For a nice recent review, see \cite{cai2017}.}

Summarizing, our approach provides:
\begin{enumerate}
\item
A systematic way to unify factor and sparse models to construct statistical specifications which use all available information and that can be applied, for example, to
\begin{enumerate}
\item
Forecasting in a high-dimensional setting;
\item
Construction of counterfactuals to aggregate data; or
\item
Estimation of partial correlation networks;
\end{enumerate}
\item
An inferential procedure to test for general structures in covariance and partial covariance matrices which can be useful, among other things, for:
\begin{enumerate}
\item
Test for misspecification in factor models; or
\item
Test for nontrivial links among idiosyncratic units.
\end{enumerate}
\end{enumerate}

\subsection{Organization of the Paper}
In addition to this Introduction, the paper is organized as follows. We present the model setup and assumptions in Section \ref{S:Setup}. The theoretical results are presented in Section \ref{S:Theory}. We discuss the empirical application in Section \ref{S:Applications}. Section \ref{S:Conclusions} concludes. All proofs, additional discussion, simulation, and empirical results are deferred to the Supplementary Material due to space constraints. Tables and figures in the Supplementary Material are referenced with an ``S'' before the number.

\subsection{Notation} All random quantities (real-valued, vectors and matrices) are defined in a common probability space $(\Omega,\mathscr{F},\P)$. We denote random variables by an upper case letter, $X$, and its realization by a lower case letter, $X=x$. The expected value operator is with respect to the $\P$ law, such that $\E(X):=\int_\Omega X(\omega) \mathrm{d}\P(\omega)$. Matrices and vectors are written in bold letters $\b X$.  Except for the number of factors,  $r$,  and the number of covariates,  $k$,  defined below,  all other dimensions are allowed to depend on the sample size ($T$). However, we omit this dependency throughout the paper to avoid clustering the notation prematurely. Also, we write $[n]:=\{1,\dots,n\}$ for $n\in\N$ and denote the cardinality of a set $\m{S}$ by $|\m{S}|$.

We use $\|\,\cdot\,\|_p$ to denote the $\ell^p$ norm for $p\in[1,\infty]$, such that for a $d$-dimensional (possibly random) vector $\b{X}=(X_1,\ldots, X_d)'$, we have $\|\b{X}\|_p:=(\sum_{i=1}^d|X_i|^p)^{1/p}$ for $p\in[1,\infty)$ and $\|\b{X}\|_\infty:=\sup_{i\leq d} |X_i|$. If $\b X$ is a $(m\times n)$ possibly random matrix, then $\|\b{X}\|_p$ denotes the matrix $\ell^p$-induced norm and $\|\b X\|_{\max}$ denotes the maximum entry in absolute terms of the matrix $\b X$.  Note that whenever $\b X$ is random,  then $\|\b{X}\|_p$ for $p\in[1,\infty]$ and $\|\b{X}\|_{\max}$ are random variables.  We also reserve the symbol $\|\,\cdot\,\|$ without subscript for the Euclidean norm  $\|\cdot\| :=\|\cdot\|_2$ for both vectors and matrices. We denote the $\ell^p(\P)$ norm of $X$ by $\viii{X}_p$ for $p\in[0,\infty]$, i.e., $\viii{X}_p:=(\E |X|^p)^{1/p}$ for $p\in[1,\infty)$ and $\viii{X}_\infty$ is the essential supremum of $X$ (in respect to $\P$).  Also,  when $\b X$ is a random vector, we define $\viii{\b X}_p:=\sup_{\|\b{u}\|\leq 1}\viii{\b{u}'\b X}_p$. Note that $\|\b X\|_p$ is random variable while $\viii{\b X}_p$ is non-random scalar for $p\in[0,\infty]$.

For any vector $\b{X}$, $\diag(\b{X})$ is a diagonal matrix whose diagonal is the elements of $\b{X}$. $\1(A)$ is an indicator function ion the event $A$, i.e, $\1(A) = 1$ if $A$ is true and $0$ otherwise. For any matrix (possibly random)  $\b M\in \R^{n\times T}$, $M_{i,t}\in\R$ represents the entry for row $i$ and column $t$, $\b M_{t}\in\R^n$ is the column-vector with all rows of column $t$ and $\b M_{i,\cdot}\in \R^T$ is a column-vector with the transpose of all elements of row $i$. We decide not to write $\b M_{t}\in\R^n$ as $\b M_{\cdot,t}$ to avoid a cumbersome notation. Finally, for non-negative sequences $x_m$ and $y_m$, we write $x\lsim y$ if there is a constant $C$ independent of $m$ such that $x_m\leq C y_m$ for all $m$. Also, we write $x\asymp y$ if both $x\lsim y$ and $y \lsim x$. Similarly, for non-negative random sequences $X_m$ and $Y_m$,  we write $X_m\lsim_\P Y_m$ if for every $\epsilon>0$ there is a constant $C$ independent of $m$ such that $\P(X_m\leq C Y_m)\leq \epsilon$ for all $m$. Also, $X_m\asymp_\P Y_m$, if both $X_m\lsim_\P Y_m$ and $Y_m \lsim_\P X_m$.

\section{Setup and Method}\label{S:Setup}

\subsection{Data Generating Process}\label{S:Example}

We consider a very general panel data model rich enough to nest several important cases in economics, finance, and related areas. We define the following Data Generating Process (DGP).
\begin{assumption}[DGP]\label{Ass:DGP} For $T\geq 4$ and $n\geq 2$, the process $\{Y_{i,t}:i\in[n], t\in [T]\}$  is generated by covariate-adjusted factor model
\begin{equation}\label{E:DGP}
Y_{i,t} =\b\gamma_{i}' \b X_{i,t}+\underbrace{\b\lambda_i'\b F_t, +U_{i,t}}_{=:R_{i,t}}
\end{equation}
where $\b X_{i,t}$ is a $k$-dimensional observable (random) vector which may also include a constant term and is typically used for adjustments of heterogeneity, seasonality, and covariate dependence, $\b F_t$ is a $r$-dimensional vector of common latent factors, and $U_{i,t}$ is a zero mean idiosyncratic component.\footnote{For simplicity, we assume that all the units $i$ have the same number of covariates ($k$). The framework can certainly accommodate situations where $k_i$ depends on $i$.  It also includes cross-sectional regression as a specific example.} The unknown parameters are $\b\gamma_i\in\R^{k}$, the factor loadings $\b\lambda_i$, and the covariance matrix of the idiosyncratic components. Finally, we assume that $\b X_{i,t}$, $\b F_t$ and $U_{i,t}$ are mutually uncorrelated, but can be serially autocorrelated. 
\end{assumption}

\begin{Rem}
In Assumption \ref{Ass:DGP} we consider that $k$, the dimension of $\b X_{i,t}$, is finite and fixed. Furthermore, the relation between $Y_{i,t}$ and $\b X_{i,t}$ is linear. This is for the sake of exposition. As our theoretical results are written in terms of the consistency rate of the first-step estimation, the DGP can be made much more general by just changing the rates.
\end{Rem}

\begin{Rem}\label{R:endo}
The assumption that $\b X_{i,t}$, $\b F_t$ and $U_{i,t}$ are mutually uncorrelated can be relaxed. Whenever $\b X_{i,t}$ is correlated with $\b F_t$ and $U_{i,t}$, and the interest lies of the estimation of the parameters $\b\gamma_i$, $i\in[n]$, the method proposed by \cite{mhP2006} can applied in the first-stage of the procedure considered in this paper and our theoretical results will follow. Nevertheless, we provide several examples where the assumption that $\b X_{i,t}$, $\b F_t$ and $U_{i,t}$ are mutually uncorrelated is reasonable.
\end{Rem}

Our modeling strategy does not stop at \eqref{E:DGP}.  Instead, we attempt to further explain $U_{i,t}$ and impose dynamics on $\b F_t$.  The former allows us to use other idiosyncratic components to further explain $Y_{i,t}$ and hence, increase the information set. The latter builds a dynamic model $\b F_t$ to facilitate out-of-sample prediction, for instance.

For each $i\in[n]$, let $\b W_{i,t}$ be a vector whose elements  form a (non-empty) subset of $\left(\b U_{-i,t}', \b U_{t-1}',\ldots,\b  U_{t-l_U}'\right)'$, where $l_U$ is a non-negative integer (much) smaller than $T$. This subset of variables attempts to explain further $U_{i,t}$ by the following population regression model:
\begin{equation}\label{E:U_reg}
U_{i,t} =\b \theta_i'\b W_{i,t} + V_{i,t},
\end{equation}
where $\E\left(\b W_{i,t}'V_{i,t}\right) = \b 0$, for  $i\in[n]$ and $t\in[T]$.
For simplicity of exposition, we assume that the dimension of $\b W_{i,t}$ is the same for all $i\in[n]$, which we denote by $d_W$. Clearly, $d_W\leq n (l_U+1)$.  

Similarly, for a non-negative integer $l_F$ (much) smaller than $T$, let $\b G_{j,t}$ be a vector whose elements form a (non-empty) subset of $\left(\b F_{t}', \b F_{t-1}',\ldots,\b  F_{t-l_F}'\right)'$, for  $j\in[r]$, which attempts to explain further the $j^{th}$ latent factor. Consider the following population regression model:
\begin{equation*}
F_{j,t} =\b \rho_j'\b G_{j,t} + V_{j,t}^F,
\end{equation*}
where $\E\left(\b G_{j,t}'{V_{j,t}^F}\right) = \b 0$, for $j\in[r]$ and $t\in[T]$.  Once again, we assume that the dimension of $\b G_{j,t}$ is the same for all $j\in[r]$, which we denote by $d_G$, where $d_G\leq r(l_F+1)$. Note that we do not necessarily exclude $\b F_t$ from $\b G_{t}$ to contemplate cases when the contemporaneous factor is available in a prediction exercise. For future reference we write \eqref{E:3rd_stage} as
\begin{equation}\label{E:3rd_stage}
\b F_{t} =\b P \b G_{t} + \b V_{t}^F,
\end{equation}
where $\b G_t :=(\b G_{1,t}',\dots, \b G_{r,t}')'$, $\b V_{t}^F:=( V_{1,t}^F,\dots, V_{r,t}^F)'$, and $\b P:=(\b\rho_1:\dots:\b\rho_r)'$.

\begin{Ex}[Asset Pricing Models]\label{Ex:APM}
Suppose $Y_{i,t}$ is the return of an asset $i$ at time $t$ and let $\b X_{i,t}:=\b X_t$ be a set of $k$ observable common risk factors, such as the market returns or Fama-French factors \citep{eFkF1993,eFkF2015}. $\b F_t$ can be a set of additional, non-observable risk factors. Several asset pricing models, such as the Capital Asset Pricing Model (CAPM) or the Arbitrage Pricing Theory (APT) model, are nested into this general framework. The idiosyncratic terms, $U_{i,t}$, can be non-trivially correlated across assets, representing links among firms that are not captured by the common factor structure. In this case, we may be interested in estimating a regression where $\b W_{i,t}=\b U_{-i,t}$ such that:
\[
U_{i,t} =\b \theta_i'\b U_{-i,t} + V_{i,t}.
\]
The coefficients $\b \theta_i$ represent the links among firms after controlling for the factors. The structures of the covariance and partial covariance matrix of $\b U_t=\left(U_{1,t},\ldots,U_{n,t}\right)'$ also shed light on these potential links. 
\end{Ex}

\begin{Ex}[Networks]
Model \eqref{E:DGP} also complements the network specifications discussed in Barigozzi and Hallin (2016,2017b)\nocite{mBmH2016,mBmH2017b} and \cite{mBcB2019}. Furthermore, as discussed in the previous example, the test proposed here can be used to detect network links as in \cite{fxDkY2014}, and \cite{cBgGgL2020}. As another example, $Y_{i,t}$ can be the (realized) volatility of financial assets and $\b X_{i,t}:=\b X_t$ can be volatility factors as in \cite{eAeG2021}.
\end{Ex}

\begin{Ex}[FAVAR]
In the case where the index $i$ represents a different dependent (endogenous) variable and $U_{i,t}$ is a dependent process, model \eqref{E:DGP} turns out to be equivalent to the Factor Augmented Vector Autoregressive (FAVAR) model of \cite{bBjBpE2005}. In this case, $\b{X}_{i,t}$ may include a constant and seasonal dummies,  and $\b W_{i,t}=\left(\b U_{t-1}',\ldots,\b  U_{t-l_U}'\right)'$. Furthermore, in order to construct $h$-step-ahead out-of-sample forecasts, we can set $\b G_{j,t}=(\b F_{t-h}',\ldots,\b F_{t-l_F}')'$.
\end{Ex}

\begin{Ex}[Panel Data Models]
Model \eqref{E:DGP} is the panel model with iterative fixed-effects considered in \cite{lGtM2016}, where the authors propose an alternative to the Synthetic Control method of \cite{aAjG2003} to evaluate the effects of regional policies. Model \eqref{E:DGP} is also in the heart of the \textnormal{\texttt{FarmTreat}} method of \cite{jFrMmM2020a}, where the authors set $\b W_{i,t}=\b U_{-i,t}$.
\end{Ex}

\subsection{The Method}\label{S:methodology}
The method proposed here for estimation, inference, and prediction consists of multiple stages, where the residuals' covariance structure can be tested at the end of each stage.
\begin{enumerate}
\item
For each $i\in[n]$ run the regression:
\[
Y_{i,t}=\b\gamma_i'\b{X}_{i,t}+R_{i,t},\quad \,t\in[T],
\]
and compute $\widehat{R}_{i,t}:= Y_{i,t} - \widehat{\b\gamma}_i'\b{X}_{i,t}$. The first stage may consist of a regression on a constant, a deterministic time trend, and seasonal dummies, for instance, or, as in Example \ref{Ex:APM}, a regression on observed factors. After removing the contribution from the observables, we can use the test for the null hypothesis of no remaining (partial) covariance structure to check if the (partial) covariance of $R_{i,t}$ is dense or sparse. If it is dense, we move to Step 2. Otherwise, we jump directly to Step 3. This first parametric, low-dimensional step can be replaced by a nonlinear/nonparametric regression or by a high-dimensional model when, for example, the number of observed factors is large. As pointed out in Remark \ref{R:endo}, the Pesaran's (2006) estimator can be also used whenever correlation between $\b{X}_{i,t}$ and $R_{i,t}$ is allowed. This will be discussed more in the subsequent sections.
\item
Write $\b{R}_t:=(R_{1,t},\ldots,R_{n,t})'$ and $\b{R}_t=\b\Lambda\b{F}_t+\b{U}_t$.
This step consists of estimating $\b\Lambda$ and $\b{F}_t$, for $t\in[T]$, through principal component analysis (PCA) \footnote{Another approach is to use the joint estimation as in \cite{agarwal2012noisy}.  
The key difference between the two approaches is the optimization-based approach (fully-iterated) and the one-step approach as well as the different assumptions behind the two approaches.  The PCA approach is based on the strong factor assumption with a large eigengap but does not impose the sparse structure on the idiosyncratic component covariance matrix. See \cite{FLM2013} for additional discussion.  }
of $\widehat{\b R}_t$ and compute
\[
\widehat{\b U}_t=\widehat{\b R}_t - \widehat{\b\Lambda}\widehat{\b F}_t.
\]
After estimating the factors and loadings, we apply our testing procedure to check for the remaining covariance structure in $\b U_t$. The second-step estimation can be carried out also by dynamic factor models. In Section \ref{S:Guide} we discuss the selection of the number of factors.
\item
Now, define $\widehat{\b W}_{i,t}$ be a non-empty subset of $(\widehat{\b U}_{-i,t}',  \widehat{\b U}_{t-1}',\ldots,\widehat{\b  U}_{t-l_U}')'$ where $l_U$ is a non-negative integer (lag). This includes contemporary regression for association studies and lagged regression for prediction of $U_{i,t}$ as two specific examples. The third step consists of a sparse regression to estimate the following model for each $i\in[n]$:
\[
\widehat{U}_{i,t}=\b\theta_i'\widehat{\b{W}}_{i,t} +V_{i,t};\qquad t\in[T].
\]
The regression in Step 3 provides useful augmentation for reducing the error in explaining  $Y_{i,t}$ in \eqref{E:DGP} further from $U_{i,t}$ to $V_{i,t}$  and hence the prediction error for  $Y_{i,t}$; see \eqref{eq-2.3}.

\item 
(\emph{For partial covariance analysis only}) Estimate the following sparse regression model for each $i,j\in[n]$:
\[
\widehat{U}_{i,t}=\b\theta_{i,j}'\widehat{\b{U}}_{-ij,t} +V_{i,j,t};\qquad t\in[T],
\]
where $\widehat{\b{U}}_{-ij,t}$ is the vector $\widehat{\b{U}}_{t}$ without $i^{th}$ and $j^{th}$ elements.  Let $\{\hat V_{i,j,t}\}$ be the residuals. Then compute the partial correlation as the sample correlation of $\{(\hat V_{i,j,t}, \hat V_{j,i,t})\}_{t=1}^T$. 
\item 
(\emph{For forecasting only}) Define $\widehat{\b G}_{t}:=(\widehat{\b G}_{1,t}',\dots, \widehat{\b G}_{r,t}')'$ where $\widehat{\b G}_{j,t}$ is a non-empty subset of $(\widehat{\b F}_{t}',  \widehat{\b F}_{t-1}',\ldots,\widehat{\b  F}_{t-l_F}')'$ for each $j\in[r]$; 
and $l_F$ is a non-negative integer (lag) \footnote{We write at this generality to accommodate the cross-sectional applications in which no latent factor needs to be predicted as in the principal component regression.  In this case, $\widehat{\b{G}}_{j,t} = \widehat{F}_{j,t}$ and this step is not needed.}. This step is multiple linear regression to estimate the following model for each $j\in[r]$:
\[
\widehat{F}_{j,t}=\b\rho_j'\widehat{\b{G}}_{j,t} +V_{j,t}^F,\qquad t\in[T].
\]
This step aims at establishing a predictive model for latent factors.
The estimator $\widehat{\b P}:=(\widehat{\b\rho}_1:\dots:\widehat{\b\rho}_r)'$  will be used in the predictive model defined in \eqref{eq-2.3}.
\end{enumerate}

\subsection{Predictive Models}\label{S:prediction}

In a pure prediction exercise, one is usually interested in the linear projection of $Y_{i,t}$ onto $(\b X_{i,t}',\b G_t', \b W_{i,t}')'$ motivated by the discussion in Section \ref{S:motivation}. This results in the factor-augmented regression model (\texttt{FARM})
\begin{equation} \label{eq2.2}
 {Y}_{i,t}= {\b\gamma_i}'\b X_{i,t}+ {\b\lambda_i}'\b P {\b G}_t +{\b\theta_i}'{\b W}_{i,t} + V_{i,t}^Y; \qquad i\in[n],\quad t\in[T],
\end{equation}
in which $\b P {\b G}_t$ predicts $\b F_t$; see \eqref{E:DGP}.  This
can used for prediction, following the steps described in Section  \ref{S:methodology},  by
\begin{equation} \label{eq-2.3}
   \widehat{Y}_{i,t}:=\widehat{\b\gamma_i}'\b X_{i,t}+\widehat{\b\lambda_i}'\widehat{\b P}\widehat{\b G}_t +\widehat{\b\theta_i}'\widehat{\b W}_{i,t};\qquad i\in[n],\quad t> T.
\end{equation}
We call the prediction model {\texttt{FarmPredict}}.

Note that model \eqref{eq2.2} is equivalent to using the predictors $ \b X_{i,t}, {\b Y}_{-i,t}$ and ${\b F}_t$, which augment predictors $ \b X_{i,t}, {\b Y}_{-i,t}$ by using the common factors ${\b F}_t$. The form in \eqref{eq2.2} mitigates the collinearity issues in high dimensions. Model \eqref{eq2.2} also bridges factor regression ($\b \theta_i = \b 0$) on one end and (sparse) regression on the other end with $\b \lambda_i = \b \Lambda_{-i}' \b \theta_i$, where $\b \Lambda_{-i}$ is the loading matrix without the $i$th row. In the latter, if we set $\b W_{i,t}=\b U_{-i,t}$ and $\b G_t = \b F_t$ (hence, $\b P = \b I_r$), model \eqref{eq2.2} becomes a (sparse) regression model:
\begin{equation*} 
    {Y}_{i,t}= {\b\gamma_i}'\b X_{i,t}+ {\b\theta_i}'{\b R}_{-i,t} + V_{i,t}^Y; \qquad i\in[n],\quad t\in[T].
\end{equation*}
In this case, \eqref{eq2.2} decorrelates the predictor $\b R_{-i,t}$,  which makes the model selection consistency much easier to satisfy and forms the basis of \texttt{FarmSelect} in \cite{jFyKkW2020}. Our contribution in this specific task is to allow heterogeneity adjustments, resulting in the estimated data $\b R_t$. In general, for model \eqref{eq2.2} with sparsity, \texttt{FarmPredict} chooses additional idiosyncratic components to enhance the prediction of the factor regression.

\begin{Ex}[First-Order FAVAR]
Consider the case of model \eqref{eq2.2} where $\b X_{i,t}=1$ and $l_U=l_F=1$, and set $\b W_{i,t}=\b U_{t-1}$ and $\b G_t=\b F_{t-1}$. Therefore, we can write
\begin{equation}\label{eq:FAVAR}
\begin{split}
\b Y_t & = (\b I - \b \Theta)\b \gamma + (\b\Lambda\b P - \b\Theta\b\Lambda)\b F_{t-1} + \b\Theta\b Y_{t-1} + \b V_{t}^Y,\\
& = \b \Theta_{0} + \b \Theta_{F}\b F_{t-1} + \b\Theta \b Y_{t-1} + \b V_{t}^Y,  \end{split}
\end{equation}
where $\b Y_{t}=(Y_{1t},\ldots,Y_{nt})'$, $\b\gamma=(\gamma_1,\ldots,\gamma_n)$, and $\b\Theta=(\b\theta_1:\cdots:\b\theta_n)'$. Model \eqref{eq:FAVAR} is a first-order FAVAR model.
\end{Ex}

\subsection{Covariance Structure and Inference}\label{S:cov_inference}

In several applications, the structure of the idiosyncratic components $\b U_t=(U_{1,t},\ldots, U_{n,t})'$ is the objective of interest. An
estimator for $\b\Sigma := \E(\b U_t \b U_t')$ could be simply given by
\begin{equation} \label{eq2.5}
   \widehat{\b\Sigma}:=\frac{1}{T}\sum_{t=1}^T \widehat{\b U}_t\widehat{\b U}_t'.
\end{equation}
The task of estimating $\b\Sigma$ is well documented in literature even in high-dimensions; see, for example, \cite{FLZZ20}, \cite{oLmW2021b}, and the references therein.  Nevertheless, we show that \eqref{eq2.5} can be used within our testing framework.

In order to proper understand the (linear) relation between a pair $(U_{i,t},U_{j,t})$ of $\b U_t$, a simple covariance estimate sometimes is not enough. In applications, it is often desirable to directly measure how $U_{i,t}$ and $U_{j,t}$ are connected. By direct connection, we mean the relation between those units removing the contribution of other variables of $\b U_t$. For this purpose, we use the partial covariance between $U_{i,t}$ and $U_{j,t}$, defined for any pair $i,j\in[n]$ as $\pi_{i,j} := \E(V_{i,j,t}V_{j,i,t})$, where $V_{i,j,t}:=U_{i,t} - \mathsf{Proj}(U_{i,t}|\b U_{-ij,t})$ and $\mathsf{Proj}(U_{i,t}|\b U_{-ij,t})$ denotes the linear projection of $U_{i,t}$ onto the space spanned by all the units except $i$ and $j$, which we denote by $\b U_{-ij,t}$. As in \citet{jPpWnZjZ2009}, we suggest to estimate the partial covariance matrix $\b\Pi := (\pi_{i,j})$ by
\begin{equation} \label{eq2.6}
    \widehat{\b\Pi}:=(\widehat{\pi}_{i,j})\quad\textnormal{and}\quad \widehat{\pi}_{i,j}:=\frac{1}{T}\sum_{t=1}^T\widehat{V}_{i,j,t}\widehat{V}_{j,i,t},
\end{equation}
where $\widehat{V}_{i,j,t}:=\widehat{U}_{i,t}-\widehat{\b\theta}_{i,j}'\widehat{\b U}_{-ij,t}$ is the residual of the LASSO regression of $\widehat{U}_{i,t}$ onto $\widehat{\b U}_{-ij,t}$ obtained in step 4 for $i,j\in[n]$.

We also would like to conduct formal tests on the population structure of $\b U_t$. Specifically, we propose a test for the following null hypothesis:
\begin{equation}\label{E:null}
\H_\m{D}^\Sigma:\b\Sigma_\m{D} = \b\Sigma_\m{D}^0;\qquad \m{D}\subseteq[n]^2,
\end{equation}
for a given subset $\m{D}$, where, for a given $(n\times n)$ matrix $\b M$, the notation $\b M_\m{D}$ denotes the $d:=|\m{D}|$-dimensional vector of  $\vec(\b M)$ indexed by the elements in $\m{D}\subseteq[n]^2$. Note we allow $d$ to diverge as $n,T\to\infty$. For testing the structure on the partial covariance matrix, consider:
\begin{equation}\label{E:null_partial_cov}
\H_{\m{D}}^\Pi:\b\Pi_{\m{D}}= \b\Pi^0_\m{D};\qquad \m{D}\subseteq[n]^2.
\end{equation}

The null hypotheses \eqref{E:null} and \eqref{E:null_partial_cov} nest several cases of interest. The most common would be to test for a diagonal or a block diagonal structure in $\b\Sigma$ and/or $\b\Pi$. But it also accommodates other structures.\footnote{With minor changes, the proposed test can also be used to test the null $\b M\vec(\b\Sigma)=\b m$ for some  $d\times n^2$ matrix $\b M$ and $d$-dimensional vector $\b m$ where $d:=d_T$ is also a function of $T$.} The challenges for testing \eqref{E:null} and \eqref{E:null_partial_cov} can be summarized as follows:
\begin{enumerate}
\item As we allow for both $n$ and $d$ to diverge as $T$ grows, sometimes at a faster rate, we have a high-dimensional test where some sort of Gaussian approximation result for dependent data must be deployed as we also allow the number covariances to be tested $d$ to diverge. In this case, a high-dimensional long-run covariance matrix must be estimated if one expects to get the (asymptotic) correct test size.
\item We do not observe $\{\b U_t\}$ or $\{V_{i,j,t}\}$. Instead, we have an estimate of both from a postulated model on observable random variables. Therefore, the estimation error must be considered to claim some sort of asymptotic properties of the test. In fact, it is not uncommon to obtain estimates of both $\{\b U_t\}$  and $\{V_{i,j,t}\}$ from a multi-stage estimation procedure as we illustrate later.
\end{enumerate}

We propose to test \eqref{E:null} using the statistic
\begin{equation}\label{E:test_statistic_cov}
  S_\m{D}^\Sigma:= \|\sqrt{T}(\widehat{\b\Sigma}_\m{D} - \b\Sigma_\m{D}^0)\|_{\max}.
\end{equation}
The quantiles of $S_\m{D}^\Sigma$ are approximated by a Gaussian bootstrap. To describe the procedure, let $\b\Upsilon_{\Sigma}$ denote the $(d\times d)$ covariance matrix for the vectorized submatrix $(\widetilde{\sigma}_{i,j})_{(i,j)\in\m{D}}$, where $\widetilde{\sigma}_{i,j}:=\frac{1}{T}\sum_{t=1}^T U_{i,t}U_{j,t}$. Since the process $\{\b U_{t}\}$ might present some form of temporal dependence (refer to Assumption \ref{Ass:Moments}(c)), we estimate $\b\Upsilon_{\Sigma}$ using a Newey-West-type estimator. For a given $\m{K}\in\K$,  where $\K$ is a class of kernel functions described below in \eqref{E:kernel_class} and bandwidth $h>0$ ,
 $\b\Upsilon_\Sigma$ is estimated by
\begin{equation}\label{E:NW_cov}
    \widehat{\b\Upsilon}_\Sigma:=\sum_{|b|<T}\m{K}(b/h) \widehat{\b M}_{\Sigma,b}\qquad \textnormal{and}\qquad\widehat{\b M}_{\Sigma,b} :=\frac{1}{T}\sum_{t = b+1}^T \widehat{\b D}_{\Sigma,t}\widehat{\b D}_{\Sigma,t-b}',
\end{equation}
where $\widehat{\b D}_{\Sigma,t}$ is a $d$-dimensional vector with entries given by $\widehat{U}_{i,t}\widehat{U}_{j,t}-\widehat{\sigma}_{i,j}$ for $(i,j)\in\m{D}$, where $\widehat{\sigma}_{i,j}$ is the $(i,j)$ element of $\widehat {\b \Sigma}$ defined in \eqref{eq2.5}. Finally, let $c^*_\Sigma(\tau)$ be the  conditional $\tau$-quantile of the Gaussian bootstrap $S^*_\m{D}:=\|\b Z^*_\Sigma \|_\infty$ where $\b Z^*_{\Sigma}|\b X,\b Y \sim \m{N}(\b 0,\widehat{\b\Upsilon}_\Sigma)$, i.e.
\[c^*_\Sigma(\tau):=\inf\{q\in \R:\P(S^*_\m{D}\leq q|\b X,\b Y)\geq \tau\}\]
Theorem \ref{T:inference_cov} demonstrates the validity of  Gaussian bootstrap procedure described above, i.e., it states conditions under which the $\tau$-quantile of the test statistic \eqref{E:test_statistic_cov} can be approximated by $c^*_\Sigma(\tau)$ in the appropriate sense.

Similarly, the test statistic for \eqref{E:null_partial_cov} is given by
\begin{equation}\label{E:test_statistic_partial_cov}
  S_\m{D}^\Pi:= \|\sqrt{T}(\widehat{\b\Pi}_\m{D} - \b\Pi_\m{D}^0)\|_{\max}.
\end{equation}
Let $\b\Upsilon_{\Pi}$ denote the $(d\times d)$ covariance matrix of $(\widetilde{\pi}_{i,j})_{(i,j)\in\m{D}}$ where $\widetilde{\pi}_{i,j}:=\frac{1}{T}\sum_{t=1}^TV_{i,j,t}V_{j,i,t}$.  $\b\Upsilon_\Pi$ is estimated by
\begin{equation}\label{E:NW_partial_cov}
 \widehat{\b\Upsilon}_\Pi:=\sum_{|b|<T}\m{K}(b/h) \widehat{\b M}_{\Pi,b};\qquad \widehat{\b M}_{\Pi,b} :=\frac{1}{T}\sum_{t = b+1}^T \widehat{\b D}_{\Pi,t}\widehat{\b D}_{\Pi,t-b}',
\end{equation}
where $\widehat{\b D}_{\Pi,t}$ is a $d$-dimensional vector with entries given by $\widehat{V}_{i,j,t}\widehat{V}_{j,i,t}-\widehat{\pi}_{i,j}$ for $(i,j)\in\m{D}$. Also, let $c^*_\Pi(\tau)$ be the conditional $\tau$-quantile of the Gaussian bootstrap $S^*_\m{D}:=\|\b Z^*_\Pi \|_\infty$ where $\b Z^*_\Pi|\b X,\b Y \sim \m{N}(\b 0,\widehat{\b\Upsilon}_\Pi)$. Theorem \ref{T:inference_partial_cov} establish conditions for the validity of  Gaussian bootstrap to approximate the quantiles of \eqref{E:test_statistic_partial_cov}.

\section{Theoretical Results}\label{S:Theory}

In this section, we collect all the theoretical guarantees for estimating the model \eqref{E:DGP} by using the proposed multi-stage method described above. Specifically,  Subsection \ref{SS:estimation} presents non-asymptotic bounds for the (parametric) estimation, Subsection \ref{SS:forecast} deals similar results concerning forecasting, and Subsection \ref{SS:inference} deals with inference on the (partial) covariance structure of $\b U_t$.

To present the results, it is convenient to use a compact notation. For each $i\in[n]$, we define the $T$-dimensional vectors $\b Y_{i,\cdot}:=(Y_{i,1},\ldots, Y_{i,T})'$ and $\b U_{i,\cdot}:=(U_{i,1},\ldots, U_{i,T})'$. We also define the $(T\times k)$ matrix of covariates $\b X_{i,\cdot}:=(\b X_{i,1},\ldots, \b X_{i,T})'$, for each $i\in[n]$ and the $(T\times r)$ matrix of factors $\b F:=(\b F_1,\dots, \b F_T)'$, such that \eqref{E:DGP} can be represented as
\begin{equation}\label{E:DGP_vec_i_form}
\begin{split}
\b Y_{i,\cdot} &= \b X_{i,\cdot}\b\gamma_i +\b F\b\lambda_i+ \b U_{i,\cdot},\qquad i=1,2,\ldots,n,\\
 &= \b X_{i,\cdot}\b\gamma_i + \b R_{i,\cdot},
\end{split}
\end{equation}
for each cross-sectional unit $i$, where $\b R_{i,\cdot}:=\b F\b\lambda_i+ \b U_{i,\cdot}$.

We define for each $t\in[T]$, the $n$-dimensional vectors $\b Y_t:=(Y_{1,t},\ldots, Y_{n,t})'$ and $\b U_t:=(U_{1,t},\ldots, U_{n,t})'$; and the $nk$-dimensional vector $\b X_t:=(\b X_{1,t}',\ldots, \b X_{n,t}')'$. Also, set the $(n\times nk)$ block diagonal matrix $\b \Gamma$ whose block diagonal is given by $(\b\gamma_1',\dots, \b\gamma_n')$ and the $(n\times r)$ loading matrix $\b\Lambda:=(\b\lambda_1,\dots,\b\lambda_n)'$. Then, \eqref{E:DGP} can also be represented as panel time series
\begin{equation}\label{E:DGP_vec_t_from}
\begin{split}
\b Y_t &= \b \Gamma\b X_t+\b \Lambda \b F_t + \b U_t,\qquad t=1,2,\ldots,T\\
 &= \b \Gamma\b X_t+\b R_t,
\end{split}
\end{equation}
where $\b R_t:=\b\Lambda \b F_t+ \b U_t$.

\subsection{Estimation}\label{SS:estimation}

We start by stating the following assumption.

\begin{assumption}[\textbf{Moments and Dependency}]\label{Ass:Moments} Consider the following:
\begin{enumerate}[(a)]
    \item The stochastic process $\{\b Z_t:= (\b X_{S,t}',\b F_t',\b U_{t}')':t\in [T]\}$ is weakly stationary for each $T\in\N$, where $\b X_{S,t}$ denotes the vector $\b X_{t}$ after excluding all deterministic (non-random) components. Furthermore, the strong mixing coefficient of $\b Z_t$ is denoted by $\alpha_m$.
    
    \item Define $\b{\m{U}}_t:=(\b U_{t}',\b U_{t-1}',\dots, \b U_{t-l}')'$ for some  integer $l\geq 0$ and let $b>0$ be a finite constant such that $\lambda_{\min}\left[\E \left(\b{\m{U}}_t\b{\m{U}}_t'\right)\right]\geq b^2$, where $\lambda_{\min}(\cdot)$ is the minimum eigenvalue of $(\cdot)$.
\item
Assume there exists an universal constant $C>0$ such for all $t,s\in [T]$, $T\geq 2$ and $i\in[n]$:
\begin{enumerate}
    \item[(c.1)]  $\viii{\b Z_{t}}_{p+\epsilon}\leq C$ for some constants  $p\geq 8$ and $\epsilon>0$
    \item[(c.2)] $\viii{n^{-1/2}\left[\b U_s'\b U_t - \E(\b U_s'\b U_t)\right]}_{p}\leq C$
    \item[(c.3)] $\viii{n^{-1/2}\sum_{i=1}^n\lambda_{j,i}U_{i,t}}_{p}\leq C$
    \item[(c.4)] $\viii{\|(\b X_i'\b X_i/T)^{-1}\|}_{p}\leq C$.
\end{enumerate}
\end{enumerate}
\end{assumption}

Assumption (\ref{Ass:Moments}.a) excludes the deterministic components of $\b X_t$ to accommodate possibly non-random non-stationary, but uniformly bounded covariates as in Assumption (\ref{Ass:Moments}.c). Assumption (\ref{Ass:Moments}.b) ensures that the parameter in \eqref{E:U_reg} is well defined since, by the Cauchy interlacing theorem, we have $\inf_i \lambda_{\min}\left[\E\left(\b W_{i,t} \b W_{i,t}\right)\right]\geq \lambda_{\min}\left[\E\left( \b{\m{U}}_t\b{\m{U}}_t'\right)\right]$. Finally, Assumptions (\ref{Ass:Moments}.a) and (\ref{Ass:Moments}.c) allow us to apply a Marcinkiewicz-Zygmund type inequality for partial sums to deal with polynomial tails (see Lemma \ref{L:tail_bound_poly} in the Supplementary Material).

For each $i\in[n]$, let $\b R_i :=\b F\b\lambda_i+ \b U_i$ denote the unobservable error term in \eqref{E:DGP_vec_i_form}, $\widehat{\b\gamma}_i$ the least-squares estimator of $\b \gamma_i$ and $\widehat{\b R}_{i}:=\b Y_t-\b X_t \widehat{\b\gamma}_i$ the vector of residuals. Also, set $\widehat{\b R}:=(\widehat{\b R}_1,\dots, \widehat{\b R}_n)'$ and $\b R:=(\b R_1,\ldots, \b R_n)'$. We must control for  estimation error in the first step of the proposed methodology. The next result gives a bound for the maximum entry of the $(n\times T)$ matrix $\widehat{\b R}-\b R$ when the first stage is conducted by OLS in a linear setup. Note that in this case we assume that $\b X_{i,t}$, $\b F_t$ and $U_{i,t}$ are mutually uncorrelated.

We state the results below in terms of the strong mixing coefficient sequence, whose definition is presented here for convenience. For $m\in\{0,\ldots,T-1\}$, define
\begin{equation}\label{E:strong_mixing_coef}
    \alpha_m:=\sup\{|\P(A\cap B)-\P(A)\P(B)|:A\in\m{Z}_{1}^t,\, B\in\m{Z}_{t+m}^T, t\in[T]\},
\end{equation}
where $\m{Z}_s^t$ is the $\sigma$-algebra generated by $(\b Z_{s},\ldots, \b Z_t)$ for $1\leq s\leq t\leq T$.
Note that $\alpha_m$ might depend on both $T$ and $n$. 

\begin{theorem}\label{T:LS}  Under Assumptions \ref{Ass:DGP} and \ref{Ass:Moments}:
\[
\|\widehat{\b R}-\b R\|_{\max}\lesssim_\P \frac{\mathscr{R}_\alpha\sqrt{r}k^{1 + \tfrac{3}{p}}  n^{4/p}}{T^{1/2-1/p}},
\]
where $\mathscr{R}_\alpha:=\mathscr{R}_\alpha(T,n):=\left[\sum_{m=0}^{T-1} (m+1)^{(p/2)-2}\alpha_m^{1-p/(p+\epsilon)}\right]^{\tfrac{2}{p}}$ and $\alpha_m$ is defined by \eqref{E:strong_mixing_coef}. 
\end{theorem}

\begin{Rem}
Even though we will treat the number of factors, $r$, and the number of covariates, $k$, fixed (not depending on $T$ or $n$), we kept them explicit in the result above. Furthermore, if we assume that $\alpha_{m}\leq K (m+1)^{-c}$ for all $m$, where $K$ is a constant that might depend on $n$ and $c\geq 0$ is an universal constant, then

\[\mathscr{R}_{\alpha}(T,n)\lsim K^{\frac{2\epsilon}{p(p+\epsilon)}}\times
\begin{cases}
   1 &; c>\frac{(p-2)(p+\epsilon)}{2\epsilon}\\
   (\log T)^{\frac{2}{p}} &; c=\frac{(p-2)(p+\epsilon)}{2\epsilon}\\
   T^{\frac{2}{p}\left[\frac{(p-2)(p+\epsilon)}{2\epsilon}+1\right]} &; c<\frac{(p-2)(p+\epsilon)}{2\epsilon}.
   \end{cases}
\]
\end{Rem}

\begin{Rem}
When the first step of the method involves more complicated estimation, such as Pesaran's (2006) method, instrumental variables, or LASSO,  we write $\|\widehat{\b R}-\b R\|_{\max} \lsim_\P \varrho_R$, where $\varrho_R$ is a non-negative sequence of $n$ and $T$. This approach is adopted systematically in the following theorems.
\end{Rem}

Define the $(n\times T)$ matrices $\b Y := (\b Y_1,\ldots,\b Y_T)$ and $\b U := (\b U_1,\ldots,\b U_T)$;  and the $(nk\times T)$ matrix $\b X:=(\b X_{1},\ldots, \b X_T)$. We can write \eqref{E:DGP} in the matrix form as
\begin{equation}\label{E:DGP_matrix_form}
\b Y =\b \Gamma\b X+\b\Lambda\b F' + \b U.
\end{equation}

Notice that  $\widehat{\b R}= \b\Lambda\b F' + \widetilde{\b U}$ where $\widetilde{\b U}:=\b U + \widehat{\b R}-\b R$ and  $(\b\Lambda,\b F)$ can be estimated by Principal Component Analysis (PCA), which minimizes $\|\widehat{\b R} -\b\Lambda\b F'\|_F^2$ with respect to $\b\Lambda$ and $\b F$, subject to $\b F'\b F/T=\b I_r$.  The solution $\widehat{\b F}$ is the matrix whose columns are $\sqrt{T}$ times $r$ eigenvectors of the top $r$  eigenvalues of $\widehat{\b R}'\widehat{\b R}$ and $\widehat{\b\Lambda} = \widehat{\b R}\widehat{\b F}/T$.

Since we do not observe $\b U$, in the third step of the method we use  $\widehat{\b U}:=\widehat{\b R}-\widehat{\b\Lambda}\widehat{\b  F}'$ instead. Therefore, we must control for the estimation error in the previous steps: $\widehat{\b U} - \b U$. Also, the loading matrix $\b\Lambda$ and the factors $\b F$ are not separably identified since $\b\Lambda\b F_t = \b\Lambda\b H'\b H\b F_t$ for any matrix $\b H$ such that $\b H'\b H=\b I_r$. If we let $\b H:=T^{-1}\b V^{-1}\widehat{\b F}'\b F\b\Lambda'\b\Lambda$, where $\b V$ is the $(r\times r)$ diagonal matrix containing the $r$ largest eigenvalues of $\widehat{\b R}'\widehat{\b R}/T$ in decreasing order, we have that $\b H\b F_t$ is identified as $\b\Lambda \b F_t$ is identified.

The result below first appeared in \cite{Bai2003} for the case of  $(n,T)$ diverging and was further extended to hold uniformly in $(i\leq n,t\leq T)$ by \cite{FLM2013}. %\cite{FLZZ20} makes the conditions modular. 
However, both consider the case when the factor model is estimated using the actual data instead of the ``estimated'' ones (residuals) as in our case. Therefore, the next result is a generalization that takes into account the pre-estimation error term and quantifies how the error impacts the precision of factor analysis.

\begin{assumption}[\textbf{Factor Model}]\label{Ass:Factor_Model} Assume:
\begin{enumerate}[(a)]
    \item $\E(\b F_t) = \b 0$, $\E(\b F_t \b F_t') = \b I_r$, and $\b\Lambda'\b\Lambda$ is a diagonal matrix;
    \item All eigenvalues of $\b\Lambda'\b\Lambda/n$ are bounded away from zero and infinity as $n\to\infty$;
    \item $\|\b\Sigma - \b\Lambda\b\Lambda'\|\lsim 1$; and
    \item $\|\b\Lambda\|_{\max}\lsim 1$.
\end{enumerate}
\end{assumption}

\begin{Rem}
Assumption \ref{Ass:Factor_Model} is standard in the literature. Assumption $\E(\b F_t)=\b 0$ is not restrictive as our approach considers a first-step estimation which may include a constant in the set of regressors.  Assumption (\ref{Ass:Factor_Model}.a) is also needed for identifiability of the factor structure. Assumption (\ref{Ass:Factor_Model}.c) imposes a strong factor structure.
\end{Rem}

\begin{theorem}\label{T:FM} Under Assumptions \ref{Ass:DGP} --\ref{Ass:Factor_Model} , let $\varrho_R$ be a non-negative sequence of $n$ and $T$ such that $\|\widehat{\b R}-\b R\|_{\max} \lsim_\P \varrho_R$. Then,
\begin{enumerate}
\item[(a)] $\max_{t\leq T}\|\widehat{\b F}_t - \b H \b F_t\|_2\lsim_\P \frac{1}{\sqrt{T}} + \frac{T^{1/p}}{\sqrt{n}} +\varrho_R(nT)^{1/p}$,
\item[(b)] $\max_{i\leq n}\|\widehat{\b\lambda}_i - \b H \b\lambda_i\|_2\lsim_\P \frac{\mathscr{R}_\alpha n^{2/p}}{\sqrt{T}} + \frac{1}{\sqrt{n}} +\varrho_R$, and
\item[(c)] $
\|\widehat{\b U} - \b U\|_{\max}\lsim_\P \frac{\mathscr{R}_\alpha n^{2/p}}{T^{1/2-1/p}}+\frac{T^{1/p}}{\sqrt{n}} + \varrho_R(nT)^{1/p}$,
\end{enumerate}
provided that $\mathscr{R}_\alpha n^{4/p}/\sqrt{T} + (nT)^{1/p}\varrho_R \lsim 1$, where $\mathscr{R}_\alpha$ is defined in Theorem \ref{T:LS}.
\end{theorem}

By setting $\varrho_R=0$ we have the case of no estimation error in the first step. Note that in order to
have the error $\|\widehat{\b U} - \b U\|_{\max}$ vanishing in probability we must have the pre-estimation error $\|\widehat{\b R}-\b R\|_{\max}$ of order (in probability) smaller than  $(nT)^{-1/p}$. %in the polynomial case or $[\log(nT)]^{-1/\gamma_2}$ in the exponential case.

We have decided not to replace $\varrho_R$ in Theorem \ref{T:FM} with the rate obtained in Theorem \ref{T:LS} as the latter only applies to the least-squares estimator. In some applications, however, the first step of the procedure could be done using a different type of estimator. For instance, a penalized adaptive Huber regression as in \cite{fan2017estimation} if the number of features $k$ is comparable or even larger than $T$ and the tail of the distribution of $\b{X}_t$ is heavy. By stating Theorem \ref{T:FM} in terms of a generic rate, it is easier to account for the effect of different estimators. By combining Theorems \ref{T:LS} and \ref{T:FM} we have the following corollary.

\begin{corollary}\label{C:LS+FM} Under the same assumptions of Theorems \ref{T:LS} and \ref{T:FM}, when OLS is used in the first-stage to obtain $\widehat{\b R}$, we have 
\[
\|\widehat{\b U} - \b U\|_{\max}\lsim_\P \frac{\mathscr{R}_\alpha n^{5/p}}{T^{1/2-2/p}}+\frac{T^{1/p}}{\sqrt{n}}=:\varpi_U.
\]

\end{corollary}

We propose to estimate \eqref{E:3rd_stage} by LASSO using $\widehat{\b U}$ in place of $\b U$. Specifically, for a regularization parameter $\xi> 0$, we denote by $\widehat{\b\theta}_i$ a minimizer of $\mathcal{Q}_i$ given by
\begin{equation}\label{E:LASSO_obj_fun}
\mathcal{Q}_i(\b a):=\frac{1}{T}\sum_{t\in[T]} (\widehat{U}_{i,t} - \b a'\widehat{\b W}_{i,t} )^2+ \xi\|\b a\|_1;\qquad i\in[n].
\end{equation}

The next result presents non-asymptotic bounds for the (in sample) prediction error and the $\ell_1$-estimation error for the LASSO estimator \eqref{E:LASSO_obj_fun} based on the ``estimated data'' and quantifies how the estimation errors impact on the choice of regularization parameter and rates of convergence.

\begin{theorem}\label{T:Oracle_bounds} Let $\varrho_U$ be a non-negative sequence of $n$ and $T$ such that $\|\widehat{\b U}- \b U\|_{\max} \lsim_\P \varrho_U$ and  assume that Assumption \ref{Ass:Moments} holds. For every $\epsilon>0$ there is a constant $0<C_\epsilon<\infty$ such that if the penalty parameter is set $\xi\geq C_\epsilon\xi_0$, then for any minimizer $\widehat{\b\theta_i}$ of \eqref{E:LASSO_obj_fun},  with probability at least $1-\epsilon$:
\[
\max_{i\in[n]} \left[(\widehat{\b\theta}_i - \b\theta_i)'\E\left(\b W_{i,t} \b W_{i,t}'\right) (\widehat{\b\theta}_{i} - \b\theta_i)
+ \xi\|\widehat{\b\theta}_i - \b\theta_i\|_1\right]\leq 8 \frac{\xi^2s_0}{b^2},
\]
provided that $\frac{\xi_1 s_0}{b}\leq K_\epsilon$ where $s_k:=\max_{i\in[n]}\|\b\theta_i\|_k$ for $k\in\{0,1,2\}$, $K_\epsilon$ is a positive constant only depending on $\epsilon$, and 
\begin{align*}
\xi_0 &:=\left(1+s_2\right)\hspace{-0.06cm} \mathscr{L}_\alpha\left[n(l+1)\right]^{2/p}T^{-1/2}+\left(1+s_1\right)\left[(nT)^{1/p}\varrho_U + \varrho_U^2\right], \\
\xi_1 &:= \mathscr{L}_\alpha\left[n(l+1)\right]^{4/p}T^{-1/2}+\left[(nT)^{1/p}\varrho_U + \varrho_U^2\right]  
\end{align*}
with $\mathscr{L}_\alpha:= \mathscr{L}_\alpha(T,n,l):=\left[\sum_{t=0}^{T-1}(m+1)^{(p/2)-2}\alpha_{(m-l)\lor 0}^{1-p/(p+\epsilon)}\right]^{2/p}$.  
\end{theorem}

%\begin{Rem}
%Notice that we apply the compatibility condition on the non-random covariance matrix $\E (\b U_{-i}' \b U_{-i})/T$ instead of  the estimated random covariance matrix $\widehat{\b U}_{-i}'\widehat{\b U}_{-i}/T$ or the ``unobservable'' random matrix $ \b U_{-i}'\b U_{-i}/T$. Careful review of the proofs reveals that the same is true for the gradient of the objective function that defines our parameter via a first order condition. \textcolor{red}{Shouldn't it be $\b W$ instead of $\b U$?}
%\end{Rem}

Once again, we purposely avoided replacing $\varrho_U$ in Theorem \ref{T:Oracle_bounds} with the rate of Corollary \ref{C:LS+FM}  to make it readily applicable to the case when a different type of factor model was used or, as a matter of fact, any other pre-estimation procedure. By setting $\varrho_U$ equal to $\varpi_U$, the rate of Corollary \ref{C:LS+FM}, we have the next result.

\begin{corollary}\label{C:FM+LASSO} If $\varrho_U$ defined in Theorem \ref{T:Oracle_bounds} is taken to be rate given by Corollary \ref{C:LS+FM} and $l< T$ then under the conditions of the Theorem \ref{T:Oracle_bounds}: 
\[
\max_{i\in[n]}\|\widehat{\b\theta}_i - \b\theta_i\|_1\lsim_\P  \frac{s_0\left(1+s_1\right)}{b^2}\left[\frac{\mathscr{L}_\alpha n^{6/p}}{T^{(1/2)-(3/p)}}  +\frac{T^{2/p}}{n^{(1/2)-(1/p)}}\right]=:\varpi_\theta .
\]
%for the polynomial case; and 
%\[ \max_{i\in[n]}\|\widehat{\b\theta}_i - \b\theta_i\|_1\lsim_\P \left[ \frac{s_0(1+\max_{i}\|\b\theta_i\|_1)}{b^2}\left(\frac{ [\log (nT)]^{3/\gamma_2}[\log n]^{1/2+1/\gamma_2}}{\sqrt{T}}  +\frac{[\log T]^{2/\gamma_2}}{\sqrt{n}}\right) \right]\]
%for the exponential case.
\end{corollary}

\subsection{Forecasting}\label{SS:forecast}

Recall that in the context of out-of-sample forecasting, our object of interest is $\widetilde{Y}_{i,t}:={\b\gamma_i}'\b X_{i,t}+ {\b\lambda_i}'\b P {\b G}_t +{\b\theta_i}'{\b W}_{i,t}$ for $i\in[n]$ and $t\geq T$, which we estimate using $\widehat{Y}_{i,t}:=\widehat{\b\gamma_i}'\b X_{i,t}+\widehat{\b\lambda_i}'\widehat{\b P}\widehat{\b G}_t +\widehat{\b\theta_i}'\widehat{\b W}_{i,t}$. The next result bounds (in probability) the prediction error bound in terms of all estimation errors from previous steps.

\begin{theorem}\label{T:prediction_bounds} Under Assumptions \ref{Ass:DGP} and \ref{Ass:Moments} , let $\varrho_\gamma, \varrho_U, \varrho_\theta, \varrho_P,\varrho_\lambda$ and $\varrho_F$ be non-negative sequence on $n$ and $T$ such that, uniformly in $i\in[n]$, $\|\widehat{\b\gamma}_i-\b\gamma_i\|_{1} \lsim_\P \varrho_\gamma$, $\|\widehat{\b U}-\b U\|_{\max} \lsim_\P \varrho_U$, $\|\widehat{\b \theta}_i-\b\theta_i\|_{1} \lsim_\P \varrho_\theta$, $\|\widehat{\b P}-\b P\|_{2} \lsim_\P \varrho_P$, $\|\widehat{\b \lambda}_i-\b\lambda_i\|_{2} \lsim_\P \varrho_\lambda$, and $\|\widehat{\b F}_t-\b F_t\|_{2} \lsim_\P \varrho_F$, respectively. Then, for every $t\geq 1$, 
\[
\max_{i\in[n]}\left|\widehat{Y}_{i,t}-\widetilde{Y}_{i,t}\right|\lsim_\P (\varrho_\gamma+\varrho_\theta) n^{1/p} + \varrho_U s_1+ \varrho_P + \varrho_\lambda +\varrho_F,\]
where $s_1$ is defined in Theorem \ref{T:Oracle_bounds}.
\end{theorem}

\begin{Rem} If $\varrho_\gamma, \varrho_U, \varrho_\theta,\varrho_\lambda$ and $\varrho_F$ are taken to be rates appearing in Theorems~\ref{T:LS}--\ref{T:Oracle_bounds}, we have $\varrho_\gamma\lor \varrho_\lambda\lor \varrho_F\lsim \varrho_U$ and $\varrho_U\lsim \varrho_\theta$. If further, $\b P$ is estimated via OLS (for each $j\in[r])$ using $\{\widehat{\b F}_t,t\in[T]\}$,  we have $\varrho_P \lsim T^{-1/2}\lor \varrho_F\lsim \varrho_F$. Therefore, Theorem~\ref{T:prediction_bounds} reduces to
\[
\max_{i\in[n]}\left|\widehat{Y}_{i,t}-\widetilde{Y}_{i,t}\right|\lsim_\P \varpi_\theta\left(n^{1/p} +s_1\right),
\]
where $\varpi_\theta$ is the rate in Corollary~\ref{C:FM+LASSO}.
\end{Rem}

\subsection{Inference on Covariance and Partial Covariance Matrices}\label{SS:inference}

We now obtain the null distributions of our test statistics for the structures of the covariance and the partial covariance. Recall the setup  and notation of Section \ref{S:cov_inference}. Further define $\widetilde{\b \Sigma}:=T^{-1}\sum_{t=1}^T\b U_t \b U_t'$ and $\widetilde{\b \Sigma}_\m{D}$ as the element of $\widetilde{\b \Sigma}$ indexed by $\m{D}\subseteq[n]^2$.
Also, for any pair of random vectors $\b X,\b Y$ of the same dimensional, $d$, say; define the distance in distribution
\[
\rho(\b X, \b Y):=\sup_{A\in\m{R}}|\P( \b X\in A) - \P(\b Y \in A)|,
\]
where $\m{R}$ is the class of all rectangles in the from $\bigtimes_{j=1}^d (a_j,b_j]$ for some $-\infty\leq a_j\leq b_j\leq \infty$ and $j\in[d]$.

We assume that the kernel $\m{K}(\cdot)$ appearing in the covariance estimator defined by \eqref{E:NW_cov} belongs to the class defined in \cite{andrews91} which we reproduce below for convenience
\begin{equation}\label{E:kernel_class}
\begin{split}
\K:=\Bigg\{f:\R\to[-1,1]:f(0)=1,&f(x)=f(-x),\forall x\in\R,\\ &\int f^2(x)dx<\infty, f\text { is continuous}\Bigg\}.
\end{split}
\end{equation}
This includes most of the well-known kernels used in the literature. To avoid confusion, it is worth pointing out that our tuning parameter $h$, also called bandwidth parameter by \cite{andrews91}, is supposed to diverge, as opposed to the bandwidth in the density kernel estimation setup, which is expected to shrink to zero.

The following result shows how accurately the covariance matrix elements are estimated and validate the bootstrap method. 

\begin{theorem}\label{T:inference_cov} 
For $\m{D}\subseteq [n]^2$, let $\widetilde{\b J}:=\sqrt{T}(\widetilde{\b\Sigma}_\m{D} - \b\Sigma_\m{D})$ and $\b{\m{G}}$ be a zero-mean Gaussian vector with the same covariance matrix of $\widetilde{\b J}$, i.e., $\b{\m{G}}\sim N(\b{0},\b\Upsilon_\Sigma)$.
Under Assumptions \ref{Ass:DGP}--\ref{Ass:Factor_Model}, if further
\begin{enumerate}[(a)]
\item
$\{\b U_t:t\in[T]\}$ is fourth-order stationary process for each $T$;
\item
The strong mixing coefficient of $\{\b U_t:t\in[T]\}$,  $\alpha_m$, obeys $\alpha_m\leq  K m^{-r}$ for $r>\left(\frac{p+\epsilon}{\epsilon}\right)\left(\frac{p}{4}-1\right)\lor \frac{2p}{p-4}$ where $p\geq 8$ and $\epsilon>0$ are defined in Assumption \ref{Ass:Moments}; and a constant $K$ that might depend on $d$;
\item
the minimum eigenvalue of $\b\Upsilon_\Sigma$ is greater or equal to $\underline{c}$, for some $\underline{c}>0$, then,
\end{enumerate}
 \begin{align*}
    \rho(\widetilde{\b J},\b{\m{G}})&\lsim \frac{\log T}{\underline{c}^2}\left[ \frac{\log d}{T^{(1/2)-\kappa}}  +\frac{K^{1-4/p}\log d}{T^{(r/2)(1-4/p) -1}} + \frac{(\log d)^{3/2}}{T^{1/4}} + \frac{d^{4/p}(\log d)^2}{T^{1/2-2/p}}\right] \\
    &\qquad  +\frac{\left[  d(\log d)^{(3/4)p-4}\log T \log (dT)  \right]}{T^{1/4}\underline{c}^{p/(p-4)}}^{\frac{2}{p-4}}+ \frac{d^{2/(p+2)}\sqrt{\log (dT)}+ K}{T^{r\kappa -1/2}},
\end{align*}
where $d:=|\m{D}|$, $\kappa:=\kappa(p,r):=\frac{1 + D_p/2}{2(r+D_p/2)}\land 1/2$ and $D_p := \frac{p}{p+2}$.

\noindent Let $\widehat{\b J}:=\sqrt{T}(\widehat{\b\Sigma}_\m{D} - \b\Sigma_\m{D})$, then
\[
\rho(\widehat{\b J},\b{\m{G}})\lsim\rho(\widetilde{\b J},\b{\m{G}})  + \inf_{\delta>0}\left[\delta_1\sqrt{1\lor \log (d/\delta)} + \P(\|\widehat{\b J}-\widetilde{\b J}\|_\infty>\delta)\right]
.\]
Let $\widetilde{\b \Upsilon}_\Sigma$ be any positive semi definite estimator of $\b\Upsilon_\Sigma$ and $\b{\m{G}}^*|\b X,\b Y\sim N(\b 0,\widetilde{\b \Upsilon}_\Sigma)$, then
\[
\rho(\widehat{\b J},\b{\m{G}}^*)\lsim \rho(\widehat{\b J},\b{\m{G}}) + \inf_{\delta>0}\left[\delta \log d(1\lor |\log d|) + \P(\|\widetilde{\b\Upsilon}_\Sigma - \b\Upsilon_\Sigma\|_{\max}>\delta)\right].
\]
\end{theorem}

\begin{Rem}\label{R: cov_2} The first result in Theorem \ref{T:inference_cov} bounds the Komolgorov distance between the (unobservable) process $\left\{\frac{1}{\sqrt{T}}\sum_{t=1}^T \b U_t \b U_t '-\E\left(\b U_t \b U_t'\right)\right\}_{T\geq 1}$ and a Gaussian process with the same covariance structure. It is a direct consequence of a  more general Central Limit Theorem result for high-dimensional alpha mixing sequences (see Theorem \ref{T:CLT_HD} in the  Supplementary Material). The second one is similar but controls for the difference between $\b U_t-\widehat{\b U}_t$ and, therefore,  takes into account the estimation error.  Finally, the last result ensures  a bootstrap validity provided we can estimate the covariance matrix in an appropriate sense. 
\end{Rem}

\begin{Rem}\label{R:inference_cov}
Theorem \ref{T:inference_cov} seem complicated. However, they only depend on $d$, $T$, $p$, $r$, and the ``quality'' of the estimators  $\widehat{\b U}$ and $\widetilde{\b \Upsilon}$. The latter allows a different selection of estimators for any of the stages and the bootstrap procedure. If we were to specialized Theorem \ref{T:inference_cov}  to incorporate the rates obtained in Theorem \ref{T:LS} and \ref{T:FM}, and set $\widetilde{\b\Upsilon}=\widehat{\b \Upsilon}$ defined by \eqref{E:NW_cov}, we obtain a sufficient condition to ensure the bootstrap validity depending only on $n,T,r$ and $p$.
\end{Rem}

\begin{corollary}\label{C:inference_cov} Under the same conditions of Theorem \ref{T:inference_cov}, if $\widetilde{\b\Upsilon}_\Sigma=\widehat{\b \Upsilon}_\Sigma$ defined by \eqref{E:NW_cov} with $\m{K}\in\K$, where $\m{K}$ is defined by \eqref{E:kernel_class}, then, uniformly in $\m{D}\subseteq[n]^2$,
\[
\|\widehat{\b\Upsilon}_\Sigma - \b\Upsilon_\Sigma\|_{\max}\lsim_\P h\left[\varpi_U(nT)^{3/p} +n^{8/p}/\sqrt{T}\right],
\]
where $h>0$ is the bandwidth parameter of the covariance estimator and $\varpi_U$ is the rate appearing in Corollary \ref{C:LS+FM}. If further, as $h, n,T\to\infty$:
\begin{enumerate}[(a)]
\item
$\varrho_\Sigma=o(1)$, where $\varrho_\Sigma$ is the rate appearing in the first result of Theorem \ref{T:inference_cov} with $d$ replaced by $n^2$;
\item
$(\log n)^{3/2}\left(\sqrt{T}\varpi_U^2+ \frac{ n^{9/p}}{\sqrt{T}}+\frac{n^{6/p}}{T^{1/2-1/p}} +\frac{1}{n^{1/2-9/p}}\right) = o(1);$
\item
$(\log n)^3h\left[\varpi_U(nT)^{3/p} +d^{8/p}/\sqrt{T}\right] = o(1)$, then,
\end{enumerate}
\[
\sup_{\m{D}}\sup_{\tau\in(0,1)}\left|\P\left[S_\m{D}^\Sigma\leq c^*_\Sigma(\tau)\right] -\tau\right|=o(1),\]
where the first supremum is over all null hypotheses of the form \eqref{E:null} indexed by $\m{D}\subseteq[n]^2$.
\end{corollary}

The next theorem shows how well the elements of the partial autocovariance matrix are estimated and gives the conditions under which the bootstrap test is valid.  In the  calculation the partial covariance \eqref{eq2.6}, we use the residual $\widehat{V}_{i,j,t}$  of the LASSO regression of $\widehat{U}_{i,t}$ onto $\widehat{\b U}_{-ij,t}$. Define $\widetilde{\b\Pi}_\m{D}:=(\widetilde{\pi}_{i,j})_{(i,j)\in\m{D}}$ for $\m{D}\subseteq[n]^2$ and recall from Section \ref{S:cov_inference} that $\widetilde{\pi}_{i,j}:=\frac{1}{T}\sum_{t=1}^TV_{i,j,t}V_{j,i,t}$ for $i,j\in[n]$, $\b\Pi_\m{D}:=\E\left(\widetilde{\b \Pi}_\m{D}\right)$, and $\b\Upsilon_\Pi$ denotes the covariance of $\b\Pi_\m{D}$.

\begin{theorem}\label{T:inference_partial_cov} 
For $\m{D}\subseteq [n]^2$, let $\widetilde{\b Q}:=\sqrt{T}(\widetilde{\b\Pi}_\m{D} - \b\Pi_\m{D})$ and $\b{\m{H}}$ be a zero-mean Gaussian vector with the same covariance matrix of $\widetilde{\b Q}$, i.e., $\b{\m{H}}\sim N(\b 0,\b\Upsilon_\Pi)$.
Under the same assumptions and notation of Theorem \ref{T:inference_cov}, with $\b\Upsilon_\Sigma$ and $\underline{c}$ replaced by $\b\Upsilon_\Pi$ and $\underline{b}$, respectively in condition (c), we have
\begin{align*}
    \rho(\widetilde{\b Q},\b{\m{H}})&\lsim \frac{\log T
    }{\underline{b}^2}\left[ \frac{\log d}{T^{(1/2)-\kappa}}  +\frac{K^{1-4/p}\log d}{T^{(r/2)(1-4/p) -1}} +\frac{(\log d)^{3/2}}{T^{1/4}} + \frac{ d^{4/p}(\log d)^2}{T^{1/2-2/p}}\right]\\
    &\qquad +\frac{\left[  d(\log d)^{(3/4)p-4}\log T \log (dT)  \right]}{T^{1/4}\underline{b}^{p/(p-4)}}^{\frac{2}{p-4}}+ \frac{d^{1/(p/2+1)}\sqrt{\log d}+ K}{T^{r\kappa -1/2}},
\end{align*}
where $d,\kappa:=\kappa(p,r)$ are defined in Theorem \ref{T:inference_cov}.

\noindent Let $\widehat{\b Q}:=\sqrt{T}(\widehat{\b\Pi}_\m{D} - \b\Pi_\m{D})$, then
\[
\rho(\widehat{\b Q},\b{\m{H}})\lsim\rho(\widetilde{\b Q},\b{\m{H}})  + \inf_{\delta>0}\left[\delta_1\sqrt{1\lor \log (d/\delta)} + \P(\|\widehat{\b Q}-\widetilde{\b Q}\|_\infty>\delta)\right]
.\]
Let $\widetilde{\b \Upsilon}_\Pi$ be any positive semi definite estimator of $\b\Upsilon_\Pi$ and $\b{\m{H}}^*|\b X,\b Y\sim N(\b 0,\widetilde{\b \Upsilon}_\Pi)$, then
\[
\rho(\widehat{\b Q},\b{\m{H}}^*)\lsim \rho(\widehat{\b Q},\b{\m{H}}) + \inf_{\delta>0}\left[\delta \log d(1\lor |\log d|) + \P(\|\widetilde{\b\Upsilon}_\Pi - \b\Upsilon_\Pi\|_{\max}>\delta)\right].
\]
\end{theorem}

Similar comments as those appearing in Remarks \ref{R:inference_cov} apply to Theorem \ref{T:inference_partial_cov} as well, which results in the following corollary.

\begin{corollary}\label{C:inference_partial_cov} Under the same conditions of Theorem \ref{T:inference_partial_cov}, if $\widetilde{\b\Upsilon}_\Pi=\widehat{\b \Upsilon}_\Pi$ defined by \eqref{E:NW_partial_cov} with $\m{K}\in\K$, where $\m{K}$ is defined by \eqref{E:kernel_class}, then, uniformly in $\m{D}\subset[n]^2$,
\[
\|\widehat{\b\Upsilon}_\Pi - \b\Upsilon_\Pi\|_{\max}\lsim_\P h\left[((1 + \widetilde{s}_1) \varpi_U + \varrho_\chi n^{1/p})(1+\widetilde{s}_1)^3(nT)^{3/p} + \frac{(1+\widetilde{s}_2)^4n^{8/p}}{\sqrt{T}}\right] ,
\]
where $h>0$ is the bandwidth of the covariance estimator, $\widetilde{s}_k:=\max_{(i,j)\in\m{D}}\|\b\chi_{i,j}\|_k$ for $k\in\{0,1,2\}$, $\varpi_U$ is the rate appearing in Corollary \ref{C:LS+FM}, and $\varrho_\chi$ is the rate appearing in Corollary \ref{C:FM+LASSO} with $s_0$ and $s_1$ replaced by $\widetilde{s}_0$ and $\widetilde{s}_1$, respectively, and $l=0$. If further, as $h, n,T\to\infty$:
\begin{enumerate}[(a)]
\item
$\varrho_\Pi=o(1)$, where $\varrho_\Pi$ is the rate in the first result of Theorem \ref{T:inference_partial_cov} with $d$ replaced by $n^2$;
\item
$(\log n)^{3/2}\left[(1+\widetilde{s}_1+\varrho_\chi)^2 \left(\sqrt{T}\varpi_U^2+ \frac{ n^{9/p}}{\sqrt{T}}+\frac{n^{6/p}}{T^{1/2-1/p}} +\frac{1}{n^{1/2-9/p}}\right) + \varrho^2_\chi n^{4/p}\sqrt{T}\right] = o(1);$
\item
$(\log n)^3h\left\{[(1 + \widetilde{s}_1) \varpi_U + \varrho_\chi n^{1/p}](1+\widetilde{s}_1)^3(nT)^{3/p} + \frac{(1+M_2)^4n^{8/p}}{\sqrt{T}}\right\} = o(1)$, then,
\end{enumerate}
\[
\sup_{\m{D}}\sup_{\tau\in(0,1)}|\P\left[S_\m{D}^\Pi\leq c^*_\Pi(\tau)\right] -\tau|=o(1),\]
where the first supremum is over all null hypotheses of the form \eqref{E:null} indexed by $\m{D}\subseteq[n]^2$.
\end{corollary}

\begin{Rem} As opposed to the case of testing covariance, when testing partial covariance in high dimensions, the sparse structure plays a role in terms of $\widetilde{s}_0$ appearing in the conditions (b) and (c). Therefore,  these assumptions restrict the cases where the proposed partial covariance test has the correct asymptotic size.  For instance, in the case of a complete dense partial covariance structure, we are likely to have $\widetilde{s}_0$ of the order of $n$ and, therefore, conditions $(b)$ and $(c)$ are not expected to hold.
\end{Rem}

\section{Applications}\label{S:Applications}

\subsection{Factor Models and Network Structure in Asset Returns}

We illustrate our methodology by studying the factor structure of asset returns. We consider monthly close-to-close excess returns from a cross-section of 9,456 firms traded on the New York Stock Exchange. The data starts on November 1991 and runs until December 2018. There are 326 monthly observations in total. In addition to the returns, we also consider 16 monthly factors: Market ($\mathsf{MKT}$), Small-minus-Big ($\mathsf{SMB}$), High-minus-Low ($\mathsf{HML}$), Conservative-minus-Aggressive ($\mathsf{CMA}$), Robust-minus-Weak ($\mathsf{RMW}$), earning/price ratio, cash-flow/price ratio, dividend/price ratio, accruals, market beta, net share issues, daily variance, daily idiosyncratic variance, 1-month momentum, and 36-month momentum. The firms are grouped according to 20 industry sectors as in \cite{tMmG1999}. The following sectors are considered:\footnote{The number between parenthesis indicates the number of firms in our sample that belong to each sector.} Mining (602), Food (208), Apparel (161), Paper (81), Chemical (513), Petroleum (48), Construction (68), Primary Metals (133), Fabricated Metals (186), Machinery (710), Electrical Equipment (782), Transportation Equipment (166), Manufacturing (690), Railroads (25), Other transportation (157), Utilities (411), Department Stores (67), Retail (1018), Financial (3419), and Other (11).

We start the analysis by looking at the correlation matrix for  monthly returns of a sample of nine sectors: Mining, Food, Petroleum, Construction, Manufacturing, Utilities, Department Stores, Retail, and Financial. Figure \ref{F:corrRet} plots the correlations that are larger than 0.15 in absolute value. We test for the null of a diagonal covariance matrix. The null hypothesis is strongly rejected with $p$-value much lower than 1\% for all sectors. To conduct the test of the covariance matrix we use the simple sample estimator as described in the paper. However, the correlations plotted in Figure \ref{F:corrRet} and in the subsequent ones are based on the nonlinear shrinkage estimator proposed by \cite{oLmW2020}.

\begin{figure}[ht]
\centering
\includegraphics[width=\linewidth]{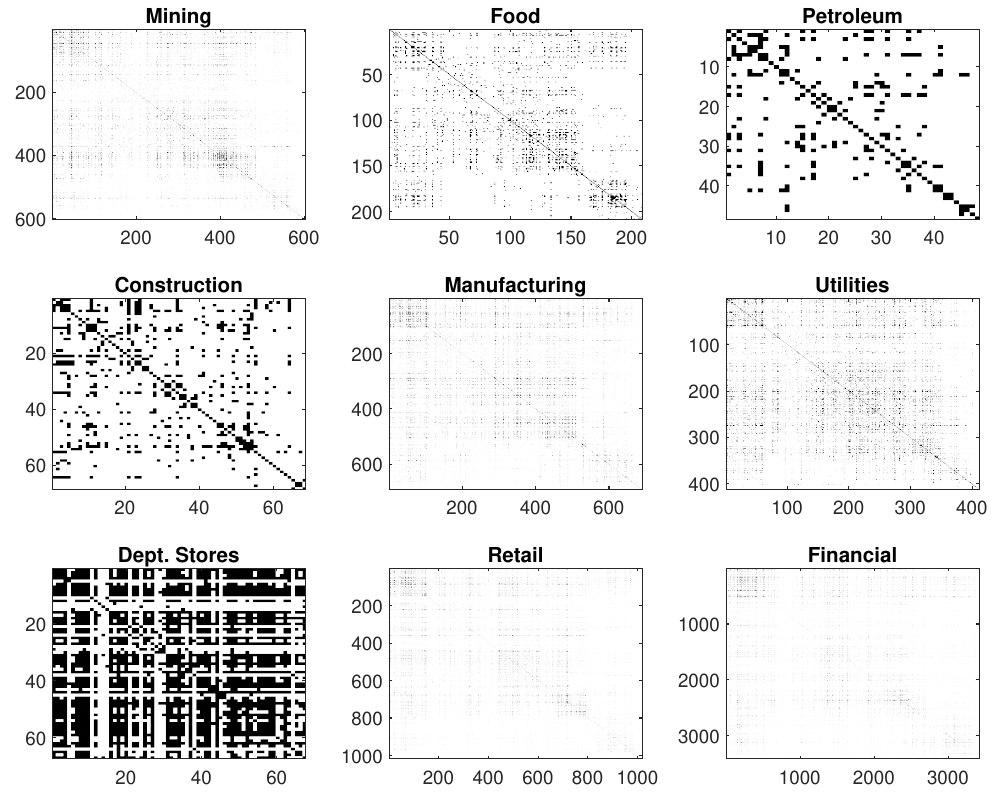}
\caption{Correlations of monthly returns.}
\begin{minipage}{1\linewidth}
\begin{footnotesize}
We estimate the correlations between all pairs of returns from specific sectors. The correlations that are higher than 0.15 in absolute value are shown as black dots.
\end{footnotesize}
\end{minipage}
\label{F:corrRet}
\end{figure}

We proceed by regressing the daily returns on the observed 16 factors. Figure \ref{F:corr1st} presents the estimated correlations for the first-stage residuals, namely factor-adjusted returns. We focus on the nine sectors as before. The first-stage regression is efficient in removing the correlation within specific sectors in some cases. The most notable ones are Financial and Retail, followed by Construction, Petroleum, and Manufacturing. On the other hand, Utilities, Department Stores, Mining, and Food still display a dense covariance matrix.

\begin{figure}[ht]
\centering
\includegraphics[width=\linewidth]{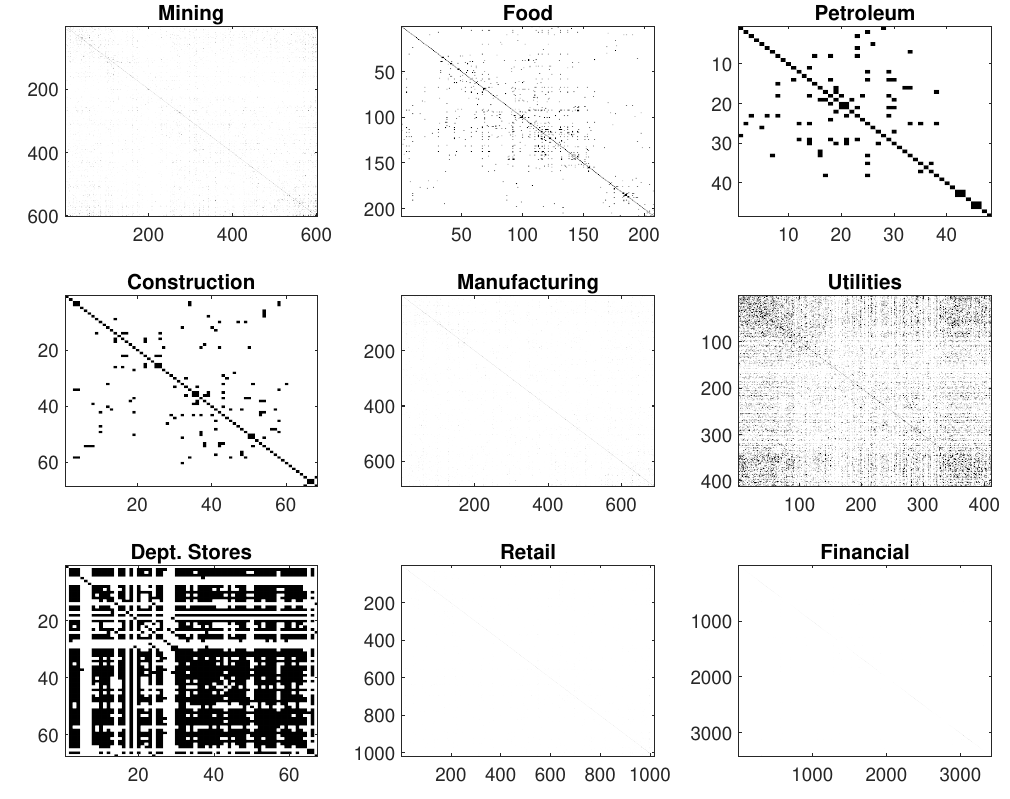}
\caption{Correlations of first-stage residuals.}
\begin{minipage}{1\linewidth}
\begin{footnotesize}
We estimate the correlations between all pairs of residuals from the first-stage OLS regression on 16 observed factors from specific sectors. The correlations that are higher than 0.15 in absolute value are shown as black dots.
\end{footnotesize}
\end{minipage}
\label{F:corr1st}
\end{figure}

The second step is to conduct a principal component analysis on the residuals  of the first stage. The eigenvalue ratio procedure selects two factors, while all four information criteria points to a single factor. We proceed with two factors. Note that, by construction, the principal component factors are orthogonal to all the 16 risk factors considered in the first stage. Figure \ref{F:corr2nd} shows the estimated correlations for the residuals (idiosyncratic component) of the second-stage. The latent factor is not able to reduce the correlations within each sector. However, when we consider the partial correlations the conclusions are much different. As can be seen from Figure \ref{F:parcorr} that the partial correlation matrices are (almost) diagonal. In addition, we are not able to reject the null of a diagonal covariance matrix at a 5\% significance level.

\begin{figure}[ht]
\centering
\includegraphics[width=\linewidth]{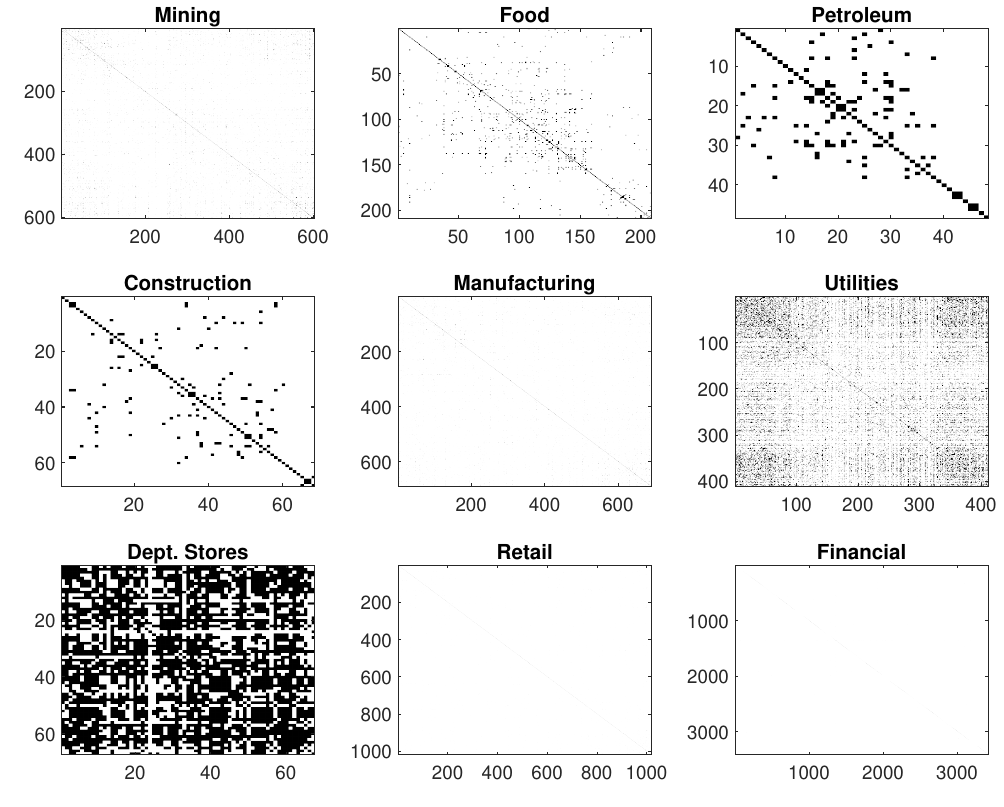}
\caption{Correlations of second-stage residuals.}
\begin{minipage}{1\linewidth}
\begin{footnotesize}
We estimate the correlations between all pairs of residuals from the second-stage principal component analysis from specific sectors. The correlations that are higher than 0.15 in absolute value are shown as black dots.
\end{footnotesize}
\end{minipage}
\label{F:corr2nd}
\end{figure}

\begin{figure}[ht]
\centering
\includegraphics[width=\linewidth]{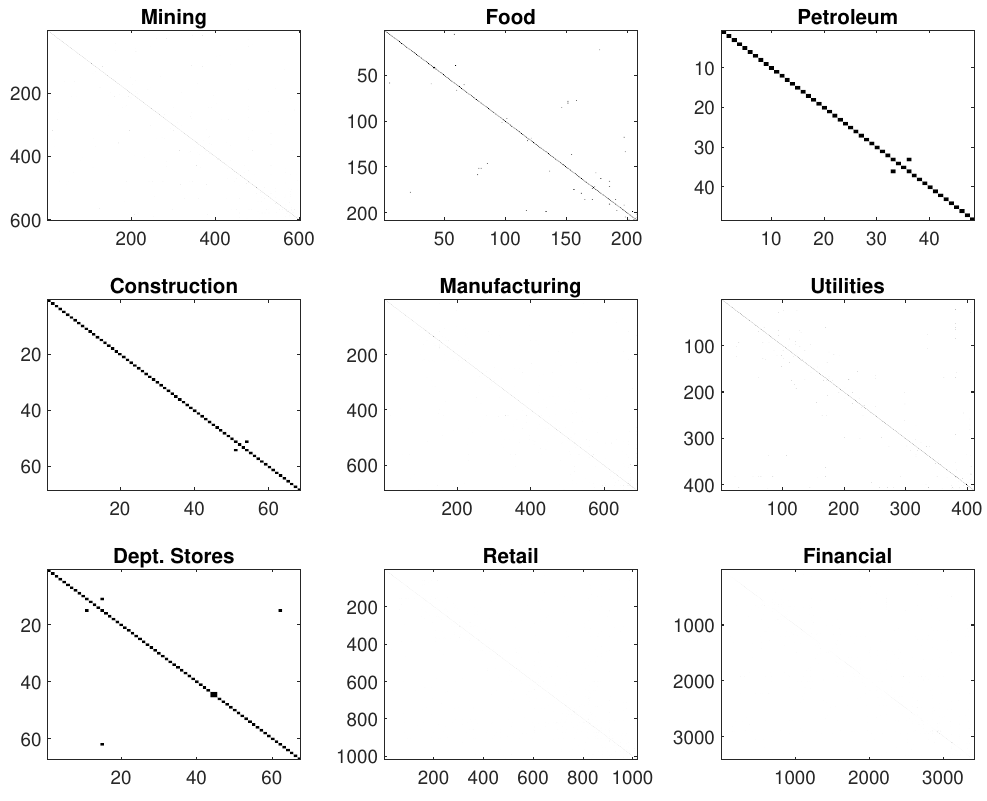}
\caption{Partial correlations of second-stage residuals.}
\begin{minipage}{1\linewidth}
\begin{footnotesize}
We estimate the partial correlations between all pairs of residuals from the second-stage LASSO regression from specific sectors. The correlations that are higher than 0.15 in absolute value are shown as black dots.
\end{footnotesize}
\end{minipage}
\label{F:parcorr}
\end{figure}

To shed some light on the links among different sectors, we report how often variables from sector $i$ are selected in the third-stage LASSO regression for firms in sector $j$. The numbers are normalized by the total number of firms in each sector and are presented in Figure \ref{F:crossparcorr}. The most interesting fact is that covariates from the financial sector are the ones most frequently selected for all the other sectors. Other sectors, such as Mining, Chemical, Machinery, Electrical Equipment, Manufacturing, and Retail are also frequently selected.  This may indicate that there are industry factors, specifically a ``financial factor'', that is unmodeled in the first two stages. However, if we augment the set of regressors in the first stage by the value-weighted portfolio from the financial sector, although the remaining dependence among firms is attenuated, particularly for Department Stores, we do not get close to an exact factor model.  This finding suggests that there are hidden links among firms. 

\begin{figure}[ht]
\centering
\includegraphics[width=0.9\linewidth,height=0.4\linewidth]{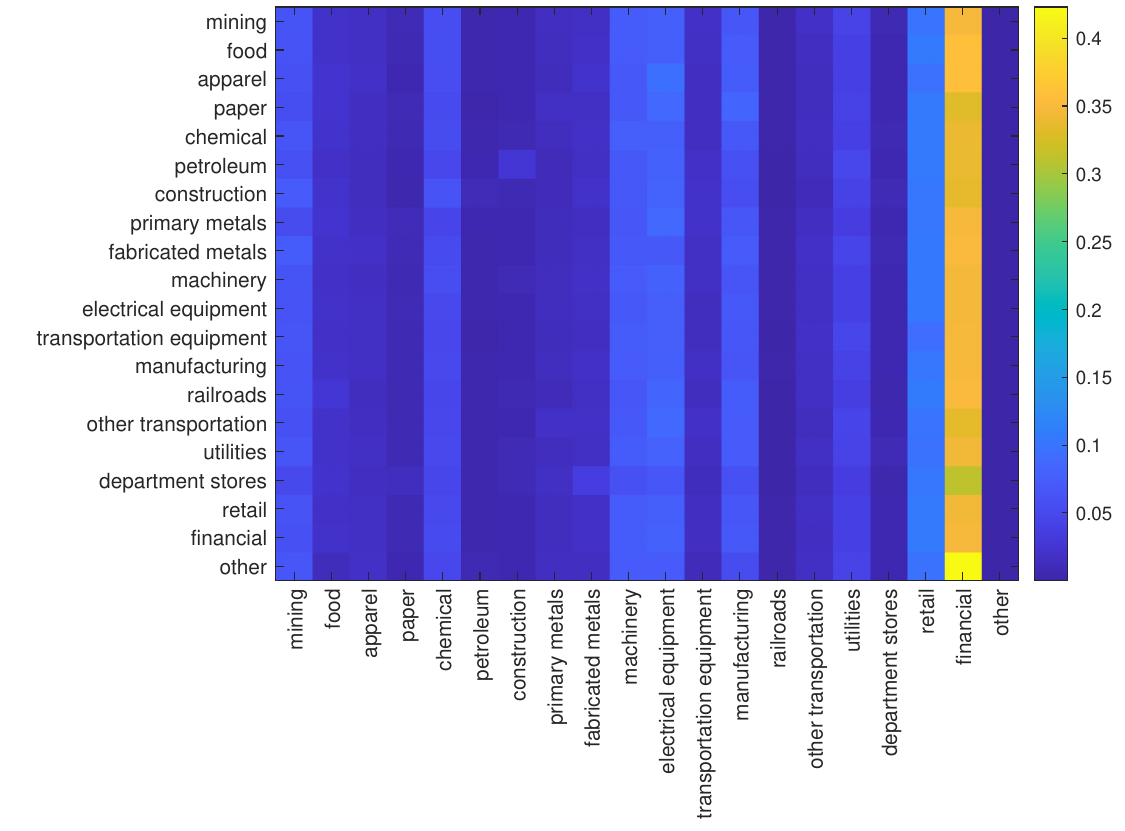}
\caption{Variable Selection Frequency.}
\begin{minipage}{1\linewidth}
\begin{footnotesize}
We report how often the variables from column sectors are selected in the third-stage LASSO regression for firms on row sectors. The numbers are normalized by the total number of firms in each sector.
\end{footnotesize}
\end{minipage}
\label{F:crossparcorr}
\end{figure}

\subsection{Forecasting US Industrial Production}

The second application consists of forecasting monthly US industrial production using a large set of monthly macroeconomic variables. We compare four different models: (1) Autoregressive model; (2) Sparse LASSO Regression (SR); (3) Principal Component Regression (PCR); and (4) \texttt{FarmPredict}.

We use variables from the August 2022 vintage of the FRED-MD database, which is a large monthly macroeconomic dataset designed for empirical analysis in data-rich macroeconomic environments. The dataset is updated in real-time through the FRED database and is available from Michael McCraken's webpage.\footnote{https://research.stlouisfed.org/econ/mccracken/fred-databases/. For further details, we refer to \cite{mMsN2016}.} .

Our sample extends from January 1960 to December 2019 (719 observations), and only variables with all observations in the period are used (122 variables). The dataset is divided into eight groups: (i) output and income; (ii) labor market; (iii) housing; (iv) consumption, orders, and inventories; (v) money and credit; (vi) interest and exchange rates; (vii) prices; and (viii) stock market. Finally, all series are transformed in order to become stationary.

In order to highlight the gains of exploring all relevant information in the dataset, we construct one-step ahead forecasts for the first-order difference of the logarithm of the monthly industrial production index (IP, growth rate): $Y_{IP,t}$.

We compare the following models:
\begin{enumerate}
\item
\textbf{Autoregressive model (\texttt{AR}):}
\[
\widehat{Y}_{IP, t+1|t}^{(\texttt{AR})} = \widehat{\phi}_{0} + \widehat{\phi}_{1}\widehat{Y}_{IP,t} + \ldots + \widehat{\phi}_{p}\widehat{Y}_{IP,t-p+1},
\]
where $\widehat{\phi}_{0},\widehat{\phi}_{1},\ldots,\widehat{\phi}_{p}$ are OLS estimates. The value of $p$ is selected by BIC.

\item
\textbf{Sparse regression (\texttt{SR}):}
\[
\widehat{Y}_{IP, t+1|t}^{(\texttt{SR})} = \widehat{\beta}_{0} + \widehat{\b\beta}_{1}'\b Y_{t} + \ldots + \widehat{\b\beta}_{p}'\b Y_{t-p+1},
\]
$\widehat{\beta}_{0},
\widehat{\b\beta}_{1}\ldots,
\widehat{\b\beta}_{p}$ are LASSO estimates and $\b Y_t=(Y_{1t},\ldots,Y_{nt})$ with $n=122$. The penalty parameter is selected by modified BIC as in \cite{hWbLcL2009}.
\item
\textbf{Principal Component Regression (\texttt{PCR}):}
\[
\widehat{Y}_{IP,t+1|t}^{(\texttt{PCR})} = \widehat{\pi}_{0} +  \widehat{\b\pi}_{1}'\widehat{\b{F}}_t +\cdots+\widehat{\b\pi}_{q}'\widehat{\b{F}}_{t-q+1},
\]
where $\widehat{\b F}_{t}$ is the estimate of the $(k \times 1)$ vector of factors $\b F_t$ given by the first $k$ principal components of $\b Y_t-\widehat{\b\mu}$ with $\widehat{\b\mu}$ being the sample average of $\b Y_t$. The parameters of the model are computed by OLS regression of $Y_{j,t}$ on a constant and lags of $\widehat{\b F}_t$. The lag $q$ is selected by BIC.

\item
\textbf{AR - Principal Component Regression (\texttt{AR-PCR}):}
\[
\widehat{Y}_{IP,t+1|t}^{(\texttt{AR-PCR})} = \widehat{\alpha}_{0} + \widehat{\b\alpha}_{1}'Y_{IP,t} + \ldots + \widehat{\alpha}_{p}Y_{IP,t-p+1} + \widehat{\b\varrho}_{1}'\widehat{\b{F}}_t +\cdots+\widehat{\b\varrho}_{p}'\widehat{\b{F}}_{t-p+1},
\]
where $\widehat{\b F}_{t}$ is the estimate of the $(k \times 1)$ vector of factors $\b F_t$ given by the first $k$ principal components of $\b Y_t-\widehat{\b\mu}$ with $\widehat{\b\mu}$ being the sample average of $\b Y_t$. The parameters of the model are computed by OLS regression of $Y_{j,t}$ on a constant, its own lags and lags of $\widehat{\b F}_t$. The lag orders $p$ and $q$ are selected by BIC.
\item
\textbf{\texttt{FarmPredict}:}
\[
\widehat{Y}_{i,t+1|t}^{(\texttt{FarmPredict})} =\widehat{\b\mu}_{j}+
\widehat{\b\lambda}_i'\widehat{\b P}_{j1}'\widehat{\b F}_t+\cdots+\widehat{\b\lambda}_i'\widehat{\b P}_{jp}'\widehat{\b F}_{t-p+1}+\widehat{\b\theta}_{1i}'\widehat{\b U}_{t} + \ldots + \widehat{\b\theta}_{pi}'\widehat{\b U}_{t-p+1},
\]
where
$\widehat{\b U}_t=\left(\widehat{U}_{1,t},\ldots,\widehat{U}_{n,t}\right)'$ and $\widehat{U}_{i,t}=Y_{i,t}-\widehat{\b\lambda}_i'\widehat{\b F}_t$, $i\in[n]$. The estimates $\widehat{\theta}_{0i},
\widehat{\b\theta}_{1i}\ldots,
\widehat{\b\theta}_{pi}$, $i\in[n]$, are given by LASSO. The penalty parameter is selected by the modified BIC and the value of $p$ is set to 24.
\end{enumerate}

The forecasts are based on a rolling-window framework of a fixed length of 480 observations, starting in January 1960. Therefore, the forecasts start on January 1990. The last forecasts are for December 2019. Note that the \texttt{AR} model only considers information concerning the own past of the variable of interest. \texttt{SR} and \texttt{PCR}/\texttt{AR-PCR}/ expand the information by two opposing routes. While $\texttt{SR}$ uses a sparse combination of the set of variables, $\texttt{PCR}$ and $\texttt{AR-PCR}$ consider a factor structure (dense model). In the case of $\texttt{AR-PCR}$, lags of the dependent variable are also included. \texttt{FarmPredict} combines these two approaches and uses the full information available. The number of factors is set to 1.

Figure \ref{F:rmse} reports the ratios of  cumulative MSE of the \texttt{FarmPredict} model against the cumulative MSE of the other benchmarks over the forecasting period. Several conclusions emerge from the plot. First, \texttt{FarmPredict} outperforms the PCR  model over the entire out-of-sample period. It is also, in general, superior to the AR, SR, and AR-PCR models, apart from 2004 and 2008. During this period, the economy experienced housing and financial crises. Furthermore, the number of out-of-sample forecasting periods was also quite small.  It is clear that the performance of \texttt{FarmPredict} improves drastically after 2008, and over the entire sample, the MSE ratio of the \texttt{FarmPredict} model over the AR benchmark is 0.9080, while the SR, PCR, and AR-PCR models have the following ratios, respectively: 0.9217, 1.0249, and 0.9215. 

\begin{figure}[ht]
\centering
\includegraphics[width=0.9\linewidth,height=0.4\linewidth]{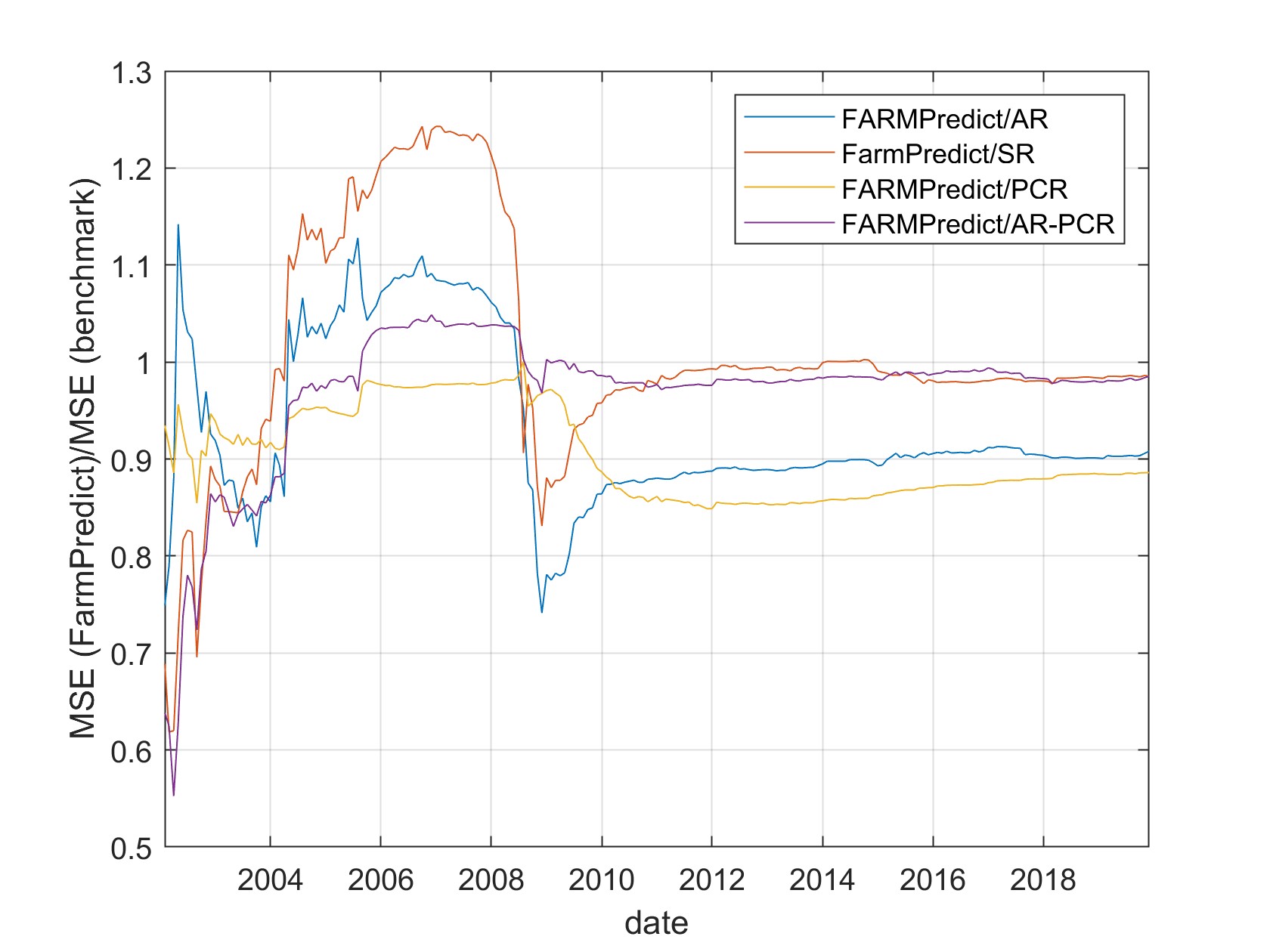}
\caption{Ratios of cumulative MSEs. }
\begin{minipage}{1\linewidth}
\begin{footnotesize}
In the figure we report the cumulative ratios of the mean squared errors (MSE) of the \texttt{FarmPredict} model and the other benchmarks over the rolling windows.
\end{footnotesize}
\end{minipage}
\label{F:rmse}
\end{figure}

\section{Conclusions}\label{S:Conclusions}

We propose a new methodology that bridges the gap between sparse regressions and factor models and evaluates the gains of increasing the information set via factor augmentation. Our proposal consists of several steps. In the first one, we filter the data for known factors (trends, seasonal adjustments, covariates). In the second step, we estimate a latent factor structure. Finally, in the last part of the procedure, we estimate a sparse regression for the idiosyncratic components. We also propose a new test for remaining structures in both high-dimensional covariance and partial covariance matrices. Our test can be used to evaluate the benefits of adding more structure to the model. Our paper has also a number of important side results. First, we proved the consistency of kernel estimation of long-run covariance matrices in high dimensions where both the number of observations and variables grows. Second, we derive the theoretical properties of factor estimation on the residuals of a first-step process. Third, the proposed test can be used as a diagnostic tool for factor models.

We evaluate our methodology with simulations and real data. The simulations show the test has good size and power properties even when the true number of factors is unknown and must be determined from the data. If the number of factors is underestimated, we observe size distortions. In particular, this is the case when the eigenvalue ratio test is used to determine the number of latent factors. The simulations also show that there are major informational gains when combining factor models and sparse regressions in a forecasting exercise. Two applications are considered.

\begin{acks}[Acknowledgments]
Medeiros gratefully acknowledges the partial financial support from CNPq and CAPES. Fan's research is partially supported by ONR grant N00014-22-1-2340 and NSF grants DMS-2210833, DMS-2053832, and DMS-2052926. We are grateful to Caio Almeida, Matteo Barigozzi, Gilberto Boareto, Gustavo Bulhões, Giuseppe Cavaliere, Frank Diebold, Bruno Ferman, Marcelo Fernandes, Claudio Flores, Conrado Garcia, Eric Ghysels, Alexander Giessing, Nathalie Gimenes, Marcelo J. Moreira, Henrique Pires, Yuri Saporito, and Rodrigo Targino for helpful comments. We also thank seminar participants at the SofiE online seminar series, Princeton University, the University of Amsterdam, the University of Pennsylvania, the University of Illinois at Urbana-Champaign, the Federal University of São Carlos, Rutgers University, the University of North Carolina at Chapel Hill, the University of Chicago, the University of California at Riverside, and Columbia University for a number of valuable comments. Finally, we are deeply grateful to Michele Lenza, Eduardo F. Mendes, and Michael Wolf for the careful reading of the paper and the many insightful discussions which led to a much-improved version of this manuscript. This manuscript has been also presented at a number of conferences, and we thank all the participants for their very useful comments.
\end{acks}

\bibliography{ref}

\pagebreak
\begin{supplement}

\setcounter{section}{0}
\setcounter{table}{0}
\setcounter{figure}{0}
\setcounter{page}{1}
\renewcommand{\theequation}{S.\arabic{equation}}
\renewcommand{\thefigure}{S.\arabic{figure}}
\renewcommand{\thetable}{S.\arabic{table}}
\renewcommand{\thesection}{S.\arabic{section}}
\renewcommand{\thetheorem}{S.\arabic{theorem}}
\renewcommand{\thelemma}{S.\arabic{lemma}}
\renewcommand{\thecorollary}{S.\arabic{corollary}}

\stitle{}

\startcontents[sections]

\begin{center}
\vspace{0.5cm}
CONTENTS
\vspace{0.5cm}
\end{center}

\printcontents[sections]{l}{1}{\setcounter{tocdepth}{3}}

\newpage

\section{Introduction}

The goal of this supplement is to provide additional results as well as the proofs of all theoretical results in the main body of the paper. 

This Supplementary Material is organized as follows. We start by discussing guidelines for practical implementation of the method in Section \ref{S:Guide}.  Section \ref{S:Simulation} provides some additional simulation results. Section \ref{App:Exp} presents results for the case with geometric mixing and Exponential Tails. Section \ref{App: Main Results} contains all the proofs of the results on paper. Finally, Section \ref{S:lemmas} collects auxiliary lemmas.

\section{Guide to Practice}\label{S:Guide}

The methodology in this paper involves several steps. The first step consists of identifying known covariates that we may want to control for. It may involve the removal of deterministic trends and seasonal effects, for instance. This can be done either by parametric or nonparametric regressions. It is important to notice, however, that the convergence rates of the estimations in the subsequent steps will be influenced by the convergence rate of the estimation in the first part of the procedure.

After the data are filtered in the first step, one can test for the remaining covariance structure. If the covariance matrix of the filtered data is (almost) diagonal, there is no need to estimate a latent factor structure, and the practitioner may jump directly to the third step.
On the other hand, if the covariance of the filtered data is dense, a latent factor model should be considered, and the number of factors must be determined. To determine the number of factors, we consider either the eigenvalue ratio test of \cite{sAaH2013} or the information criteria put forward in \cite{jBsN2002}. The factors can be estimated by the usual methods.

The next step involves a sparse regression in order to estimate any remaining links between idiosyncratic components. Before running the last step, we may test for a diagonal covariance matrix of the idiosyncratic terms. If the null is not rejected, there is no need for additional estimation. In case of rejection, we can proceed with a LASSO regression. We recommend that the penalty term is selected by some Information Criterion (IC) as advocated by \cite{hWbLcL2009} and \cite{mMeM2016}. If out-of-sample forecasting is the goal, an additional final step is necessary. In this case, lags of factors and idiosyncratic terms must be determined. The practitioner may rely on the usual information criteria available. 

Finally, concerning the estimation of the long-run matrices, the usual methods discussed in the literature can be used here to select the kernel and the bandwidth. We use the simple Bartlett kernel with bandwidth given as $\lfloor T/3 \rfloor$.

\section{Simulation}\label{S:Simulation}

In this section, we report simulation results divided into two parts. In the first one, we evaluate the finite-sample properties of the test for the remaining covariance structure. In the second part, we highlight the informational gains when considering both the common factors and the idiosyncratic component. We simulate 1,000 replications of the following model for various combinations of sample size ($T$) and number of variables ($n$):
\begin{align}
Y_{i,t}&=\b\lambda_i'\b{F}_t+U_{i,t},\,i=[n],\,t=[T],\label{E:simul1}\\
\b{F}_t&=0.8  \b{F}_{t-1} + \b E_{t},\label{E:simu2}\\
U_{i,t}&=   I(i=1)(\theta_{12} U_{2t} + \theta_{13} U_{3t} + \theta_{14} U_{4t} + \theta_{15} U_{5t})+ V_{i,t}
\label{E:simul3}\\
V_{i,t} &=\phi V_{it-1} + \epsilon_{i,t}
\end{align}
where $I(\cdot)$ is the indicator function, $\{\epsilon_{i,t}\}$ is a sequence of independent t-distributed random variables with 10 degrees of freedom, and $\{\b E_{t}\}$ is a sequence of $r$-dimensional mutually independent random vectors t-distributed with 10 degrees of freedom. Furthermore, $\{\epsilon_{i,t}\}$ and $\{\b E_{t}\}$ are mutually independent for all time periods, factors and variables. For each Monte Carlo replication, the vector of loadings is sampled from a Gaussian distribution with mean -6 and standard deviation 0.2 for $i=1$ and mean 2 and unit variance for $i=2,\ldots,n$. The coefficients $\theta_{12}$, $\theta_{13}$, $\theta_{14}$, and $\theta_{15}$ are equal to zero or 0.8, 0.9, -0.7, and 0.5, respectively. We set the number of factors, $r$, equal to 3. $\phi$ can be either 0 or 0.5. Note that equation \eqref{E:simul3} is a special case of equation \eqref{E:U_reg} with $\b{W}_{i,t}=\b{U}_{-i,t}$.

\subsection{Test for Remaining Covariance Structure}\label{S:SimulCov}

We start by reporting results for the test of no remaining structure on the covariance matrix of $\b{U}_t=(U_{1t},\ldots,U_{nt})'$. The null hypothesis considered is that all the covariances between the first variable ($i=1$) and the remaining ones are all zero. For size simulations we set $\theta_{12}=\theta_{13}=\theta_{14}=\theta_{15}=0$ in the DGP. To evaluate the effects of factor estimation as well as the methods in selecting the number of factors, we consider the following scenarios: (1) factors are known, and there is no estimation involved; (2) factors are estimated by principal components, but the number of factors is known; (3) the number of factors is determined by the eigenvalue ratio procedure of \cite{sAaH2013}; (4)-(7) the number of factors is determined by one of the four information criteria proposed by \cite{jBsN2002} as defined by
\[
\begin{matrix*}[l]
\textnormal{IC}_1=\log[S(r)]+r\frac{n+T}{nT}\log\left(\frac{nT}{n+T}\right) &
\textnormal{IC}_2=\log[S(r)]+r\frac{n+T}{nT}\log C_{nT}^2\\
\textnormal{IC}_3=\log[S(r)]+r\frac{\log C_{nT}^2}{C_{nT}^2} &
\textnormal{IC}_4=\log[S(r)]+r\frac{(n+T-k)\log(nT)}{nT}.
\end{matrix*}
\]
where $S(r)=\frac{1}{nT}\|\b R  - \widehat{\b\Lambda}_r\widehat{\b F}_r\|_2^2$ and $C_{nT}:=\sqrt{\min(n,T)}$.

Table \ref{T:sizeSimul1} and \ref{T:sizeSimul2} reports the results of the empirical size of test for different significance levels.  We consider the case of $\phi=0$ in Table \ref{T:sizeSimul1} and $\phi=0.5$ in \ref{T:sizeSimul2}.  The tables present the results when the factors are known in panel (a), the factors are unknown but the number of factors is known in panel (b), or the number of factors is estimated either by the information criterion $\textnormal{IC}_1$ in panel (c) or the eigenvalue ratio procedure in panel (d).

A number of facts emerge from the inspection of the results in Table \ref{T:sizeSimul1}. First, size distortions are small when the factors are known. In this case, the test is undersized when the pair $(n,T)$ is small. When the factor is not known but the true number of factors is available, the size distortions are high only when $T=100$ and $n=50$ due to inaccurate estimation of factors. However, the distortions disappear when the pair $(T,n)$ grows. In this case, the empirical size is similar to the situation reported in Panel (a). The finite performance of the test in the case where the number of factors is selected by information criterion $\textnormal{IC}_1$ is almost indistinguishable from the case reported in Panel (b). However, the results with the eigenvalue ratio procedure are much worse when $T=100$ and $n=50$. In this case, the procedure selects fewer factors than the true number $r=3$. For instance, the procedure selects 2 or fewer factors in 36\%  of the replications. Just as comparison, for $T=100$ and $n=50$, $\textnormal{IC}_1$ underdetermines the number of factors only in 3.10\% of the cases.  The latter also confirms that overestimation of the number of factors will not have a big adversarial effect. For all the other combinations of $T$ and $n$ all the data-driven methods select the correct number of factors in almost all replications.

When the idiosyncratic components are autocorrelated the size distortions are higher, as reported in Table \ref{T:sizeSimul2}. This is mainly caused by the well-known difficulties in the estimation of the long-run covariance matrix.

Table \ref{T:powerSimul1} report the results of the empirical power with $\phi=0$, $\beta_{12}=0.8$,  $\beta_{13}=0.9$, $\beta_{14}=-0.7$, and $\beta_{15}=-0.5$ in the DGP. When the factors are known, the test always rejects the null and the empirical power is one for any significance level. On the other hand, when factors must be estimated but the number of factors is known, the power decreases, as depicted in panel (b) in the table. Nevertheless, for $T=500,700$ the power is reasonably high, especially when the test is conducted at a $10\%$ significance level. For $T=100$, the performance deteriorates as $n$ grows. The results are similar when data-driven procedures are used to determine the number of factors, and the conclusions are mostly the same.

Table \ref{T:powerSimul2} report power results in a similar setting as above but with $\phi=0.5$. The above conclusions are mostly the same if $\phi=0$ or $\phi=0.5$.

The main message of the simulation exercise is that the finite-sample performance of the proposed tests depends on the correct selection of factors. Nevertheless, for the DGP considered here, the usual data-driven methods available in the literature to determine the true number of factors seem to work reasonably well.

\subsection{Informational Gains}\label{S:SimulPredict}

The goal of this simulation is to compare, in a prediction environment, the three-stage method developed in the paper by evaluating the information gains in predicting $Y_{1t}$ by three different methods. First, the predictions are computed from a LASSO regression of $Y_{1t}$ on all the other $n-1$ variables. This is the Sparse Regression (SR) approach. Second, we consider a principal component regression (PCR), i.e., an ordinary least squares (OLS) regression of the variable of interest on factors computed from the pool of other variables. Finally, we consider predictions constructed from the method proposed here, the \texttt{FarmPredict} methodology. Table \ref{T:simulPred} presents the results. The table presents the average mean squared error (MSE) over 5-fold cross-validation (CV) subsamples. As in the size and power simulations, we consider different combinations of $T$ and $n$. We report results for the case where $\theta_{12}=0.8$,  $\theta_{13}=0.9$, $\theta_{14}=-0.7$, and $\theta_{15}=-0.5$ in the DGP.

According to the DGP, the theoretical MSE is 0.25 when all the information is used. When just a factor is used, the MSE is 2.21. From the table is clear that there are significant informational gains when we consider both factors and the cross-dependence between idiosyncratic components. Several conclusions emerge from the table. First, it is clear that when the sample size increases the MSE reduces. This is expected. Second, the PCR's MSE and \texttt{FarmPredict}'s MSE are close to their theoretical values of 2.21 and 0.25 when the sample increases. The performance of the \texttt{FarmPredict} is quite remarkable when $T=500$ or $T=700$ and is always superior to Sparse Regression and PCR.

\begin{table}[tb]
\caption{\textbf{Simulation Results: Size with $\phi=0$.}}
\label{T:sizeSimul1}
\begin{minipage}{0.9\linewidth}
\begin{footnotesize}
The table reports the empirical size of the test of the remaining covariance structure. Panel (a) reports the case where the factors are known, whereas Panel (b) considers that the factors are unknown but the number of factors is known. Panels (c) and (d) present the results when the number of factors is determined, respectively, by the eigenvalue ratio test and the information criterion $IC_1$. Factors are estimated by the usual principal component algorithm. Three nominal significance levels are considered: 0.01, 0.05, and 0.10. The table reports the results for the case where $\phi=0$ in \eqref{E:simul3}.
\end{footnotesize}
\end{minipage}
\resizebox{0.9\linewidth}{!}{
\begin{threeparttable}
\begin{tabular}{lllllllllllll}
\hline
&&\multicolumn{11}{c}{\underline{\textbf{Panel(a): Known factors}}}\\
                 && \multicolumn{3}{c}{\underline{$T=100$}} && \multicolumn{3}{c}{\underline{$T=500$}} && \multicolumn{3}{c}{\underline{$T=700$}} \\
                 && 0.10 & 0.05 & 0.01 && 0.10 & 0.05 & 0.01 && 0.10 & 0.05 & 0.01 \\
\hline
$n=0.5\times T$  && 0.08 & 0.03 & 0.01 && 0.10 & 0.05 & 0.01 && 0.09 & 0.04 & 0.01\\
$n=1\times T$    && 0.06 & 0.02 & 0.00 && 0.07 & 0.03 & 0.01 && 0.10 & 0.05 & 0.01\\
$n=2\times T$    && 0.07 & 0.02 & 0.00 && 0.07 & 0.02 & 0.00 && 0.08 & 0.04 & 0.00\\
$n=3\times T$    && 0.05 & 0.01 & 0.00 && 0.08 & 0.04 & 0.01 && 0.07 & 0.04 & 0.01\\
\\
&&\multicolumn{11}{c}{\underline{\textbf{Panel(b): Known number of factors}}}\\
                 && \multicolumn{3}{c}{\underline{$T=100$}} && \multicolumn{3}{c}{\underline{$T=500$}} && \multicolumn{3}{c}{\underline{$T=700$}} \\
                 && 0.10 & 0.05 & 0.01 && 0.10 & 0.05 & 0.01 && 0.10 & 0.05 & 0.01 \\
\hline
$n=0.5\times T$  && 0.23 & 0.13 & 0.02 && 0.14 & 0.06 & 0.02 && 0.11 & 0.05 & 0.01\\
$n=1\times T$    && 0.13 & 0.06 & 0.01 && 0.09 & 0.04 & 0.01 && 0.12 & 0.05 & 0.01\\
$n=2\times T$    && 0.09 & 0.04 & 0.01 && 0.07 & 0.04 & 0.01 && 0.09 & 0.04 & 0.00\\
$n=3\times T$    && 0.06 & 0.02 & 0.00 && 0.07 & 0.04 & 0.01 && 0.07 & 0.03 & 0.01\\
\\
&&\multicolumn{11}{c}{\underline{\textbf{Panel(c): Information criterion ($\textnormal{IC}_1$)}}}\\
                 && \multicolumn{3}{c}{\underline{$T=100$}} && \multicolumn{3}{c}{\underline{$T=500$}} && \multicolumn{3}{c}{\underline{$T=700$}} \\
                 && 0.10 & 0.05 & 0.01 && 0.10 & 0.05 & 0.01 && 0.10 & 0.05 & 0.01 \\
\hline
$n=0.5\times T$  && 0.23 & 0.12 & 0.02 && 0.13 & 0.06 & 0.02 && 0.10 & 0.06 & 0.01 \\
$n=1\times T$    && 0.13 & 0.08 & 0.02 && 0.10 & 0.03 & 0.01 && 0.11 & 0.05 & 0.01 \\
$n=2\times T$    && 0.11 & 0.05 & 0.01 && 0.07 & 0.04 & 0.01 && 0.10 & 0.05 & 0.01 \\
$n=3\times T$    && 0.08 & 0.03 & 0.01 && 0.07 & 0.03 & 0.01 && 0.07 & 0.03 & 0.01 \\
\\
&&\multicolumn{11}{c}{\underline{\textbf{Panel(d): Eigenvalue ratio}}}\\
                && \multicolumn{3}{c}{\underline{$T=100$}} && \multicolumn{3}{c}{\underline{$T=500$}} && \multicolumn{3}{c}{\underline{$T=700$}} \\
                 && 0.10 & 0.05 & 0.01 && 0.10 & 0.05 & 0.01 && 0.10 & 0.05 & 0.01 \\
\hline
$n=0.5\times T$  && 0.46 & 0.35 & 0.23 && 0.13 & 0.06 & 0.02 && 0.10 & 0.05 & 0.01 \\
$n=1\times T$    && 0.12 & 0.06 & 0.02 && 0.08 & 0.04 & 0.01 && 0.12 & 0.05 & 0.01 \\
$n=2\times T$    && 0.09 & 0.04 & 0.01 && 0.08 & 0.04 & 0.01 && 0.09 & 0.04 & 0.00 \\
$n=3\times T$    && 0.06 & 0.02 & 0.00 && 0.08 & 0.04 & 0.01 && 0.06 & 0.03 & 0.01 \\
\hline
\end{tabular}
\end{threeparttable}}
\end{table}

\begin{table}[htbp]
\caption{\textbf{Simulation Results: Size with $\phi=0.5$.}}
\label{T:sizeSimul2}
\begin{minipage}{0.9\linewidth}
\begin{footnotesize}
The table reports the empirical size of the test of the remaining covariance structure. Panel (a) reports the case where the factors are known, whereas Panel (b) considers that the factors are unknown but the number of factors is known. Panels (c) and (d) present the results when the number of factors is determined, respectively, by the eigenvalue ratio test and the information criterion $IC_1$. Factors are estimated by the usual principal component algorithm. Three nominal significance levels are considered: 0.01, 0.05, and 0.10. The table reports the results for the case where $\phi=0.5$ in \eqref{E:simul3}.
\end{footnotesize}
\end{minipage}
\resizebox{0.9\linewidth}{!}{
\begin{threeparttable}
\begin{tabular}{lllllllllllll}
\hline
&&\multicolumn{11}{c}{\underline{\textbf{Panel(a): Known factors}}}\\
                 && \multicolumn{3}{c}{\underline{$T=100$}} && \multicolumn{3}{c}{\underline{$T=500$}} && \multicolumn{3}{c}{\underline{$T=700$}} \\
                 && 0.10 & 0.05 & 0.01 && 0.10 & 0.05 & 0.01 && 0.10 & 0.05 & 0.01 \\
\hline
$n=0.5\times T$  && 0.09 &0.04 &0.01 &&0.12 &0.07 &0.01 &&0.10 &0.05 &0.01\\
$n=1\times T$    && 0.07 &0.03 &0.00 &&0.07 &0.03 &0.01 &&0.11 &0.06 &0.01\\
$n=2\times T$    && 0.08 &0.02 &0.00 &&0.08 &0.03 &0.00 &&0.09 &0.05 &0.00\\
$n=3\times T$    && 0.05 &0.02 &0.00 &&0.09 &0.04 &0.01 &&0.07 &0.04 &0.01\\
\\
&&\multicolumn{11}{c}{\underline{\textbf{Panel(b): Known number of factors}}}\\
                 && \multicolumn{3}{c}{\underline{$T=100$}} && \multicolumn{3}{c}{\underline{$T=500$}} && \multicolumn{3}{c}{\underline{$T=700$}} \\
                 && 0.10 & 0.05 & 0.01 && 0.10 & 0.05 & 0.01 && 0.10 & 0.05 & 0.01 \\
\hline
$n=0.5\times T$  && 0.25 &0.15 &0.3 &&0.15 &0.07 &0.02 &&0.12 &0.06 &0.01\\
$n=1\times T$    && 0.13 &0.07 &0.01 &&0.09 &0.04 &0.01 &&0.14 &0.06 &0.02\\
$n=2\times T$    && 0.09 &0.04 &0.01 &&0.08 &0.04 &0.01 &&0.09 &0.05 &0.00\\
$n=3\times T$    && 0.08 &0.02 &0.00 &&0.08 &0.04 &0.01 &&0.08 &0.03 &0.01\\
\\
&&\multicolumn{11}{c}{\underline{\textbf{Panel(c): Information criterion ($\textnormal{IC}_1$)}}}\\
                 && \multicolumn{3}{c}{\underline{$T=100$}} && \multicolumn{3}{c}{\underline{$T=500$}} && \multicolumn{3}{c}{\underline{$T=700$}} \\
                 && 0.10 & 0.05 & 0.01 && 0.10 & 0.05 & 0.01 && 0.10 & 0.05 & 0.01 \\
\hline
$n=0.5\times T$  &&0.45 &0.41 &0.29 &&0.15 &0.07 &0.02 &&0.11 &0.06 &0.01\\
$n=1\times T$    &&0.15 &0.09 &0.02 &&0.10 &0.04 &0.01 &&0.14 &0.06 &0.01\\
$n=2\times T$    &&0.09 &0.04 &0.01 &&0.09 &0.04 &0.01 &&0.09 &0.05 &0.00\\
$n=3\times T$    &&0.07 &0.03 &0.00 &&0.10 &0.04 &0.01 &&0.08 &0.03 &0.01\\
\\
&&\multicolumn{11}{c}{\underline{\textbf{Panel(d): Eigenvalue ratio}}}\\
                && \multicolumn{3}{c}{\underline{$T=100$}} && \multicolumn{3}{c}{\underline{$T=500$}} && \multicolumn{3}{c}{\underline{$T=700$}} \\
                 && 0.10 & 0.05 & 0.01 && 0.10 & 0.05 & 0.01 && 0.10 & 0.05 & 0.01 \\
\hline
$n=0.5\times T$  && 0.25 &0.14 &0.04 &&0.14 &0.07 &0.02 &&0.13 &0.06 &0.01\\
$n=1\times T$    && 0.15 &0.07 &0.02 &&0.10 &0.04 &0.01 &&0.13 &0.06 &0.02\\
$n=2\times T$    && 0.11 &0.05 &0.01 &&0.08 &0.05 &0.01 &&0.10 &0.05 &0.00\\
$n=3\times T$    && 0.08 &0.03 &0.01 &&0.09 &0.04 &0.01 &&0.08 &0.03 &0.01\\
\hline
\end{tabular}
\end{threeparttable}}
\end{table}

\begin{table}[tb]
\caption{\textbf{Simulation Results: Power ($\phi=0$).}}
\label{T:powerSimul1}
\begin{minipage}{0.9\linewidth}
\begin{footnotesize}
The table reports the empirical power of the test of the remaining covariance structure. Panel (a) reports the case where the factors are known, whereas Panel (b) considers that the factors are unknown but the number of factors is known. Factors are estimated by the usual principal component algorithm. Three nominal significance levels are considered: 0.01, 0.05, and 0.10.
\end{footnotesize}
\end{minipage}
\resizebox{0.9\linewidth}{!}{
\begin{threeparttable}
\begin{tabular}{lllllllllllll}
\hline
&&\multicolumn{11}{c}{\underline{\textbf{Panel(a): Known factors}}}\\
                 && \multicolumn{3}{c}{\underline{$T=100$}} && \multicolumn{3}{c}{\underline{$T=500$}} && \multicolumn{3}{c}{\underline{$T=700$}} \\
                 && 0.10 & 0.05 & 0.01 && 0.10 & 0.05 & 0.01 && 0.10 & 0.05 & 0.01 \\
\hline
$n=0.5\times T$  && 1 & 1 & 1 && 1 & 1 & 1 && 1 & 1 & 1\\
$n=1\times T$    && 1 & 1 & 1 && 1 & 1 & 1 && 1 & 1 & 1\\
$n=2\times T$    && 1 & 1 & 1 && 1 & 1 & 1 && 1 & 1 & 1\\
$n=3\times T$    && 1 & 1 & 1 && 1 & 1 & 1 && 1 & 1 & 1\\
\\
&&\multicolumn{11}{c}{\underline{\textbf{Panel(b): Known number of factors}}}\\
                 && \multicolumn{3}{c}{\underline{$T=100$}} && \multicolumn{3}{c}{\underline{$T=500$}} && \multicolumn{3}{c}{\underline{$T=700$}} \\
                 && 0.10 & 0.05 & 0.01 && 0.10 & 0.05 & 0.01 && 0.10 & 0.05 & 0.01 \\
\hline
$n=0.5\times T$  && 0.35 & 0.19 & 0.03 && 0.99 & 0.98 & 0.83 & 0.99 & 0.99 & 0.95\\
$n=1\times T$    && 0.20 & 0.08 & 0.01 && 0.82 & 0.60 & 0.11 & 0.95 & 0.81 & 0.32\\
$n=2\times T$    && 0.15 & 0.07 & 0.01 && 0.82 & 0.55 & 0.11 & 0.94 & 0.82 & 0.34\\
$n=3\times T$    && 0.09 & 0.03 & 0.00 && 0.79 & 0.52 & 0.10 & 0.94 & 0.80 & 0.31\\
&&\multicolumn{11}{c}{\underline{\textbf{Panel(c): Eigenvalue ratio}}}\\
                 && \multicolumn{3}{c}{\underline{$T=100$}} && \multicolumn{3}{c}{\underline{$T=500$}} && \multicolumn{3}{c}{\underline{$T=700$}} \\
                 && 0.10 & 0.05 & 0.01 && 0.10 & 0.05 & 0.01 && 0.10 & 0.05 & 0.01 \\
\hline
$n=0.5\times T$  && 0.14 &0.09 &0.02 &&0.99 &0.97 &0.83&& 0.99 &0.99 &0.94 \\
$n=1\times T$    && 0.18 &0.07 &0.01 &&0.84 &0.59 &0.13&& 0.95 &0.81 &0.33 \\
$n=2\times T$    && 0.17 &0.07 &0.01 &&0.83 &0.54 &0.11&& 0.94 &0.82 &0.34 \\
$n=3\times T$    && 0.08 &0.03 &0.00 &&0.82 &0.53 &0.10&& 0.95 &0.81 &0.34 \\
\\
&&\multicolumn{11}{c}{\underline{\textbf{Panel(d): Information criterion ($\textnormal{IC}_1$)}}}\\
                 && \multicolumn{3}{c}{\underline{$T=100$}} && \multicolumn{3}{c}{\underline{$T=500$}} && \multicolumn{3}{c}{\underline{$T=700$}} \\
                 && 0.10 & 0.05 & 0.01 && 0.10 & 0.05 & 0.01 && 0.10 & 0.05 & 0.01 \\
\hline
$n=0.5\times T$  && 0.15 &0.09 &0.02 &&0.99 &0.97 &0.83&& 0.99 &0.99 &0.94 \\
$n=1\times T$    && 0.20 &0.07 &0.01 &&0.84 &0.60 &0.13&& 0.95 &0.81 &0.33 \\
$n=2\times T$    && 0.17 &0.07 &0.01 &&0.83 &0.55 &0.11&& 0.94 &0.82 &0.34 \\
$n=3\times T$    && 0.09 &0.03 &0.00 &&0.82 &0.56 &0.10&& 0.95 &0.81 &0.34 \\
\hline
\end{tabular}
\end{threeparttable}}
\end{table}

\begin{table}[htbp]
\caption{\textbf{Simulation Results: Power ($\phi=0.5$).}}
\label{T:powerSimul2}
\begin{minipage}{0.9\linewidth}
\begin{footnotesize}
The table reports the empirical power of the test of the remaining covariance structure. Panel (a) reports the case where the factors are known, whereas Panel (b) considers that the factors are unknown but the number of factors is known. Factors are estimated by the usual principal component algorithm. Three nominal significance levels are considered: 0.01, 0.05, and 0.10.
\end{footnotesize}
\end{minipage}
\resizebox{0.9\linewidth}{!}{
\begin{threeparttable}
\begin{tabular}{lllllllllllll}
\hline
&&\multicolumn{11}{c}{\underline{\textbf{Panel(a): Known factors}}}\\
                 && \multicolumn{3}{c}{\underline{$T=100$}} && \multicolumn{3}{c}{\underline{$T=500$}} && \multicolumn{3}{c}{\underline{$T=700$}} \\
                 && 0.10 & 0.05 & 0.01 && 0.10 & 0.05 & 0.01 && 0.10 & 0.05 & 0.01 \\
\hline
$n=0.5\times T$  && 1 & 1 & 1 && 1 & 1 & 1 && 1 & 1 & 1\\
$n=1\times T$    && 1 & 1 & 1 && 1 & 1 & 1 && 1 & 1 & 1\\
$n=2\times T$    && 1 & 1 & 1 && 1 & 1 & 1 && 1 & 1 & 1\\
$n=3\times T$    && 1 & 1 & 1 && 1 & 1 & 1 && 1 & 1 & 1\\
\\
&&\multicolumn{11}{c}{\underline{\textbf{Panel(b): Known number of factors}}}\\
                 && \multicolumn{3}{c}{\underline{$T=100$}} && \multicolumn{3}{c}{\underline{$T=500$}} && \multicolumn{3}{c}{\underline{$T=700$}} \\
                 && 0.10 & 0.05 & 0.01 && 0.10 & 0.05 & 0.01 && 0.10 & 0.05 & 0.01 \\
\hline
$n=0.5\times T$  && 0.38 & 0.18 & 0.03 && 1.00 & 1.00 & 0.91 && 1.00 & 1.00 & 1.00\\
$n=1\times T$    && 0.21 & 0.09 & 0.02 && 0.89 & 0.69 & 0.13 && 1.00 & 0.92 & 0.39\\
$n=2\times T$    && 0.18 & 0.07 & 0.01 && 0.98 & 0.59 & 0.13 && 1.00 & 0.96 & 0.36\\
$n=3\times T$    && 0.10 & 0.03 & 0.00 && 0.91 & 0.66 & 0.11 && 1.00 & 0.92 & 0.40\\
&&\multicolumn{11}{c}{\underline{\textbf{Panel(c): Eigenvalue ratio}}}\\
                 && \multicolumn{3}{c}{\underline{$T=100$}} && \multicolumn{3}{c}{\underline{$T=500$}} && \multicolumn{3}{c}{\underline{$T=700$}} \\
                 && 0.10 & 0.05 & 0.01 && 0.10 & 0.05 & 0.01 && 0.10 & 0.05 & 0.01 \\
\hline
$n=0.5\times T$  && 0.15 & 0.11 &0.03 &&1.00 &1.00 &0.96 &&1.00 &1.00 &1.00\\
$n=1\times T$    && 0.22 & 0.08 &0.01 &&0.99 &0.70 &0.16 &&1.00 &0.95 &0.38\\
$n=2\times T$    && 0.17 & 0.07 &0.01 &&0.94 &0.62 &0.12 &&0.98 &0.88 &0.37\\
$n=3\times T$    && 0.09 & 0.03 &0.00 &&0.89 &0.59 &0.11 &&1.00 &0.87 &0.40\\
&&\multicolumn{11}{c}{\underline{\textbf{Panel(d): Information criterion ($\textnormal{IC}_1$)}}}\\
                 && \multicolumn{3}{c}{\underline{$T=100$}} && \multicolumn{3}{c}{\underline{$T=500$}} && \multicolumn{3}{c}{\underline{$T=700$}} \\
                 && 0.10 & 0.05 & 0.01 && 0.10 & 0.05 & 0.01 && 0.10 & 0.05 & 0.01 \\
\hline
$n=0.5\times T$  && 0.15 & 0.11 &0.03 &&1.00 &1.00 &0.96 &&1.00 &1.00 &1.00\\
$n=1\times T$    && 0.22 & 0.08 &0.01 &&0.99 &0.70 &0.16 &&1.00 &0.95 &0.38\\
$n=2\times T$    && 0.17 & 0.07 &0.01 &&0.94 &0.62 &0.12 &&0.98 &0.88 &0.37\\
$n=3\times T$    && 0.09 & 0.03 &0.00 &&0.89 &0.59 &0.11 &&1.00 &0.87 &0.40\\
\hline
\end{tabular}
\end{threeparttable}}
\end{table}

\begin{table}[tb]
\caption{\textbf{Simulation Results: Informational Gains}}
\label{T:simulPred}
\begin{minipage}{0.9\linewidth}
\begin{footnotesize}
The table reports the average mean squared error (MSE) of three different prediction models over 5-fold cross-validation subsamples. The goal is to predict the first variable using information from the remaining $n-1$. Panel (a) considers the case of Sparse Regression (SR) where $Y_{1t}$ is LASSO-regressed on all the other variables. Panel (b) shows the results of Principal Component Regression (PCR). Finally, Panel (c) presents the results of \texttt{FarmPredict}. ``N/A'' means ``not available''. Note that there is no factor selection for Sparse Regression. ``Known Number'' means that the number of factors is known.
\end{footnotesize}
\end{minipage}
\resizebox{0.9\linewidth}{!}{
\begin{threeparttable}
\begin{tabular}{lllllllllllll}
\hline
&&\multicolumn{11}{c}{\underline{\textbf{Panel(a): Sparse Regression (SR)}}}\\
&& \multicolumn{3}{c}{\underline{Known Number}} && \multicolumn{3}{c}{\underline{Eigenvalue Ratio}} && \multicolumn{3}{c}{\underline{Information Criterion ($\textnormal{IC}_1$)}} \\
                 && $T=100$ & $T=500$ & $T=700$ && $T=100$ & $T=500$ & $T=700$ && $T=100$ & $T=500$ & $T=700$ \\
\hline
$n=0.5\times T$  && 0.60 & 0.35 & 0.34 && N/A & N/A& N/A&& N/A& N/A& N/A\\
$n=1\times T$    && 0.42 & 0.38 & 0.32 && N/A & N/A& N/A&& N/A& N/A& N/A\\
$n=2\times T$    && 0.40 & 0.35 & 0.31 && N/A & N/A& N/A&& N/A& N/A& N/A\\
$n=3\times T$    && 0.40 & 0.35 & 0.30 && N/A & N/A& N/A&& N/A& N/A& N/A\\
\\
&&\multicolumn{11}{c}{\underline{\textbf{Panel(b): Principal Component Regression (PCR)}}}\\
&& \multicolumn{3}{c}{\underline{Known Number}} && \multicolumn{3}{c}{\underline{Eigenvalue Ratio}} && \multicolumn{3}{c}{\underline{Information Criterion ($\textnormal{IC}_1$)}} \\
                 && $T=100$ & $T=500$ & $T=700$ && $T=100$ & $T=500$ & $T=700$ && $T=100$ & $T=500$ & $T=700$ \\
\hline
$n=0.5\times T$  &&3.82 &3.12 &3.01 &&4.69 &3.12 &3.01 &&3.26 &3.04 &2.34\\
$n=1\times T$    &&3.09 &2.35 &2.34 &&4.05 &3.35 &3.34 &&3.22 &3.02 &2.32\\
$n=2\times T$    &&3.14 &2.97 &2.21 &&4.13 &3.97 &2.21 &&3.29 &3.21 &2.27\\
$n=3\times T$    &&3.83 &3.00 &2.33 &&3.83 &3.00 &2.33 &&3.12 &3.00 &2.28\\
&&\multicolumn{11}{c}{\underline{\textbf{Panel(c): FarmPredict}}}\\
&& \multicolumn{3}{c}{\underline{Known Number}} && \multicolumn{3}{c}{\underline{Eigenvalue Ratio}} && \multicolumn{3}{c}{\underline{Information Criterion ($\textnormal{IC}_1$)}} \\
                 && $T=100$ & $T=500$ & $T=700$ && $T=100$ & $T=500$ & $T=700$ && $T=100$ & $T=500$ & $T=700$ \\
\hline
$n=0.5\times T$  && 0.50 &0.33 &0.31 &&0.52 &0.33 &0.31 &&0.50& 0.34 &0.30 \\
$n=1\times T$    && 0.32 &0.29 &0.28 &&0.37 &0.29 &0.28 &&0.53& 0.28 &0.27 \\
$n=2\times T$    && 0.27 &0.27 &0.26 &&0.28 &0.27 &0.26 &&0.32& 0.28 &0.28 \\
$n=3\times T$    && 0.22 &0.21 &0.21 &&0.22 &0.21 &0.21 &&0.34& 0.27 &0.27 \\
\hline
\end{tabular}
\end{threeparttable}}
\end{table}

\section{Results for the case with Geometric Mixing and Exponential Tails}\label{App:Exp}

This section extends the results from the main text (Theorems \ref{T:LS}-\ref{T:inference_partial_cov}) to the case  where the random quantities admit an exponential tail and a strong mixing coefficient with exponential decay (see  Assumption \ref{Ass:MomentsExp}(a)-(c) for a precise statement).

Before we begin, we need some additional pieces of notation. For any convex function $\psi:\R^+\to\R^+ $ such that $\psi(0)=0$ and $\psi(x)\to\infty$ as $x\to\infty$ and (real-valued) random variable $X$, we denote its Orlicz-norm by $\viii{x}_\psi$, which is defined by $\viii{X}_\psi:=\inf \left\{C>0:\E\left[\psi\left(\frac{|X|}{C}\right)\right]\leq 1\right\}$.  In particular, we have the $\ell^p$ Orlicz-norm of $X$ by $\viii{X}_p$ for $p\in[0,\infty)$ by setting $\psi(x)=x^p$ and $\viii{X}_{e^\gamma}$ the exponential Orlicz-norm for $\gamma>0$ by setting $\psi(x)=\exp(x^\gamma)-1$ for $\gamma\geq 1$ and $\psi(x)$ is the convex hull of $x\mapsto \exp(x^\gamma)-1$ for  $\gamma \in(0,1)$ (to ensure convexity).  Also,  when $\b X$ is a random vector, we define its Orlicz-norm by $\viii{\b X}_\psi:=\sup_{\|\b{u}\|\leq 1}\viii{\b{u}'\b X}_\psi$.

We strengthen Assumption (\ref{Ass:Moments}.c) by adding the following condition. Recall that $\{\alpha_{m}\}_m$ denotes the strong mixing coefficients  $\{\b Z_t\}_t$ as defined in \eqref{E:strong_mixing_coef}.

\begin{assumption}[\textbf{Moments and Dependency: Exponential Case}]\label{Ass:MomentsExp} . There are universal constants $C, K_1,\gamma_1,\gamma_2>0$ such for all $t,s\in [T]$, $T\geq 2$ and $i\in[n]$,
\begin{enumerate}[(a)]

\item $\viii{\b Z_{t}}_{e^{\gamma_2}}\leq C$
\item  $\alpha_{m}\leq \exp(-K_1m^{\gamma_1})$ for  $1\leq m<T$.
\item  $\gamma < 1$ where $\gamma$ is defined by $1/\gamma=1/\gamma_1+2/\gamma_2$
\item $\viii{n^{-1/2}\left[\b U_s'\b U_t - \E(\b U_s'\b U_t)\right]}_{e^{\gamma_2}}\leq C$
\item $\viii{n^{-1/2}\sum_{i=1}^n\lambda_{j,i}U_{i,t}}_{e^{\gamma_2}}\leq C$
\item $\viii{\|(\b X_i'\b X_i/T)^{-1}\|}_{e^{\gamma_2}}\leq C$
\item  $\frac{(\log n)^{(2/\gamma) - 1}}{T}\leq \frac{1}{C_1}$ where $C_1$ is positive constant only depending on $C$, $K_1$, $\gamma$ and $\gamma_1$.
\end{enumerate}
\end{assumption}

\begin{theorem}\label{T:LS_exp}  Under Assumptions \ref{Ass:DGP},\ref{Ass:Moments}, and \ref{Ass:MomentsExp}:
\[
\|\widehat{\b R}-\b R\|_{\max}\lsim_\P \frac{\sqrt{r}k[\log(n)\log(nTk)]^{1/\gamma_2} \sqrt{\log(nk)}}{\sqrt{T}}.
\]
\end{theorem}

\begin{theorem}\label{T:FM_exp} Under Assumptions \ref{Ass:DGP} --\ref{Ass:MomentsExp} , let $\varrho_R$ be a non-negative sequence of $n$ and $T$ such that $\|\widehat{\b R}-\b R\|_{\max} \lsim_\P \varrho_R$. Then
\begin{enumerate}[(a)]
\item
$\max_{t\leq T}\|\widehat{\b F}_t - \b H \b F_t\|_2\lsim_\P \frac{1}{\sqrt{T}} + \frac{[\log T]^{1/\gamma_2}}{\sqrt{n}} + \varrho_R[\log (nT)]^{1/\gamma_2}$
\item
$\max_{i\leq n}\|\widehat{\b\lambda}_i - \b H \b\lambda_i\|_2\lsim_\P \sqrt{\frac{\log n}{T}} + \frac{1}{\sqrt{n}} + \varrho_R$
\item
$\|\widehat{\b U} - \b U\|_{\max}\lsim_\P (\log T)^{1/\gamma_2}\sqrt{\frac{\log n}{T}}+\frac{(\log T)^{1/\gamma_2}}{\sqrt{n}} +  \varrho_R[\log (nT)]^{1/\gamma_2}$,
\end{enumerate}
provided that $\sqrt{\log n/T} + [\log(nT) \varrho_R]^{1/\gamma_2} \lsim 1$.
\end{theorem}

\begin{theorem}\label{T:Oracle_bounds_exp} Let $\varrho_U$ be a non-negative sequence of $n$ and $T$ such that $\|\widehat{\b U}- \b U\|_{\max} \lsim_\P \varrho_U$ and  assume that Assumptions \ref{Ass:Moments} and \ref{Ass:MomentsExp}  hold. For every $\epsilon>0$ there is a constant $0<C_\epsilon<\infty$ such that if the penalty parameter is set $\xi\geq C_\epsilon\xi_0$, then for any minimizer $\widehat{\b\theta_i}$ of \eqref{E:LASSO_obj_fun},  with probability at least $1-\epsilon$:
\[
\max_{i\in[n]} \left[(\widehat{\b\theta}_i - \b\theta_i)'\E\left(\b W_{i,t} \b W_{i,t}'\right) (\widehat{\b\theta}_{i} - \b\theta_i)
+ \xi\|\widehat{\b\theta}_i - \b\theta_i\|_1\right]\leq 8 \frac{\xi^2s_0}{b^2},
\]
provided that $\frac{\xi_1 s_0}{b}\leq K_\epsilon$ where $s_k:=\max_{i\in[n]}\|\b\theta_i\|_k$ for $k\in\{0,1,2\}$, $K_\epsilon$ is a positive constant only depending on $\epsilon$, and 
\begin{align*}
\xi_0 &:=\left(1+s_2\right)\sqrt{\frac{\log[n(l+1)]}{T}}+\left(1+s_1\right)\left[[\log(nT)]^{1/\gamma_2}\varrho_U + \varrho_U^2\right], \\
\xi_1 &:= \sqrt{\frac{\log [n(l+1)]}{T}}+\left[[\log(nT)]^{1/\gamma_2}\varrho_U + \varrho_U^2\right].
\end{align*}
\end{theorem}

\begin{theorem}\label{T:prediction_bounds_exp} Under Assumptions \ref{Ass:DGP}, \ref{Ass:Moments}, and  \ref{Ass:MomentsExp} , let $\varrho_\gamma, \varrho_U, \varrho_\theta, \varrho_P,\varrho_\lambda$ and $\varrho_F$ be non-negative sequence on $n$ and $T$ such that, uniformly in $i\in[n]$, $\|\widehat{\b\gamma}_i-\b\gamma_i\|_{1} \lsim_\P \varrho_\gamma$, $\|\widehat{\b U}-\b U\|_{\max} \lsim_\P \varrho_U$, $\|\widehat{\b \theta}_i-\b\theta_i\|_{1} \lsim_\P \varrho_\theta$, $\|\widehat{\b P}-\b P\|_{2} \lsim_\P \varrho_P$, $\|\widehat{\b \lambda}_i-\b\lambda_i\|_{2} \lsim_\P \varrho_\lambda$, and $\|\widehat{\b F}_t-\b F_t\|_{2} \lsim_\P \varrho_F$, respectively. Then, for every $t\geq 1$, 
\[
\max_{i\in[n]}\left|\widehat{Y}_{i,t}-\widetilde{Y}_{i,t}\right|\lsim_\P (\varrho_\gamma+\varrho_\theta) (\log n)^{1/\gamma_2} + \varrho_U s_1+ \varrho_P + \varrho_\lambda +\varrho_F,\]
where $s_1$ is defined in Theorem \ref{T:Oracle_bounds}.
\end{theorem}

\begin{theorem}\label{T:inference_cov_exp} 
For $\m{D}\subseteq [n]^2$, let $\widetilde{\b J}:=\sqrt{T}(\widetilde{\b\Sigma}_\m{D} - \b\Sigma_\m{D})$ and $\b{\m{G}}$ be a zero-mean Gaussian vector with the same covariance matrix of $\widetilde{\b J}$, i.e., $\b{\m{G}}\sim N(\b{0},\b\Upsilon_\Sigma)$.
Under Assumptions \ref{Ass:DGP}--\ref{Ass:MomentsExp}, if further
\begin{enumerate}[(a)]
\item
$\{\b U_t:t\in[T]\}$ is fourth-order stationary process for each $T$;
\item
the minimum eigenvalue of $\b\Upsilon_\Sigma$ is greater or equal to $\underline{c}$, for some $\underline{c}>0$,
\end{enumerate}
 then,
\begin{align*}
    \rho(\widetilde{\b J},\b{\m{G}})&\lsim \frac{(\log T)^{\gamma_1 + 1} \log d + \big[\log (dT)\big]^{2/\gamma}(\log d)^2\log T}{\sqrt{T}\underline{c}^2} \\
    &\qquad + \frac{(\log d)^{2} +(\log d)^{3/2}\log T+ \log d (\log T)^{\gamma_1+1}\log (dT)}{T^{1/4}\underline{c}^2},
\end{align*}
where $d:=|\m{D}|$.

\noindent Let $\widehat{\b J}:=\sqrt{T}(\widehat{\b\Sigma}_\m{D} - \b\Sigma_\m{D})$, then
\[
\rho(\widehat{\b J},\b{\m{G}})\lsim\rho(\widetilde{\b J},\b{\m{G}})  + \inf_{\delta>0}\left[\delta_1\sqrt{1\lor \log (d/\delta)} + \P(\|\widehat{\b J}-\widetilde{\b J}\|_\infty>\delta)\right]
.\]
Let $\widetilde{\b \Upsilon}_\Sigma$ be any positive semi definite estimator of $\b\Upsilon_\Sigma$ and $\b{\m{G}}^*|\b X,\b Y\sim N(\b 0,\widetilde{\b \Upsilon}_\Sigma)$, then
\[
\rho(\widehat{\b J},\b{\m{G}}^*)\lsim \rho(\widehat{\b J},\b{\m{G}}) + \inf_{\delta>0}\left[\delta \log d(1\lor |\log d|) + \P(\|\widetilde{\b\Upsilon}_\Sigma - \b\Upsilon_\Sigma\|_{\max}>\delta)\right].
\]
\end{theorem}

\begin{theorem}\label{T:inference_partial_cov_exp} 
For $\m{D}\subseteq [n]^2$, let $\widetilde{\b Q}:=\sqrt{T}(\widetilde{\b\Pi}_\m{D} - \b\Pi_\m{D})$ and $\b{\m{H}}$ be a zero-mean Gaussian vector with the same covariance matrix of $\widetilde{\b Q}$, i.e., $\b{\m{H}}\sim N(\b 0,\b\Upsilon_\Pi)$.
Under the same assumptions and notation of Theorem \ref{T:inference_cov_exp} with $\b\Upsilon_\Sigma$ and $\underline{c}$ replaced by $\b\Upsilon_\Pi$ and $\underline{b}$, respectively in condition (c), we have
\begin{align*}
     \rho(\widetilde{\b Q},\b{\m{H}})&\lsim \frac{(\log T)^{\gamma_1 + 1} \log d + \big[\log (dT)\big]^{2/\gamma}(\log d)^2\log T}{\sqrt{T}\underline{b}^2} \\
    &\qquad + \frac{(\log d)^{2} +(\log d)^{3/2}\log T+ \log d (\log T)^{\gamma_1+1}\log (dT)}{T^{1/4}\underline{b}^2},
\end{align*}
where $d:=|\m{D}|$.

\noindent Let $\widehat{\b Q}:=\sqrt{T}(\widehat{\b\Pi}_\m{D} - \b\Pi_\m{D})$, then
\[
\rho(\widehat{\b Q},\b{\m{H}})\lsim\rho(\widetilde{\b Q},\b{\m{H}})  + \inf_{\delta>0}\left[\delta_1\sqrt{1\lor \log (d/\delta)} + \P(\|\widehat{\b Q}-\widetilde{\b Q}\|_\infty>\delta)\right]
.\]
Let $\widetilde{\b \Upsilon}_\Pi$ be any positive semi definite estimator of $\b\Upsilon_\Pi$ and $\b{\m{H}}^*|\b X,\b Y\sim N(\b 0,\widetilde{\b \Upsilon}_\Pi)$, then
\[
\rho(\widehat{\b Q},\b{\m{H}}^*)\lsim \rho(\widehat{\b Q},\b{\m{H}}) + \inf_{\delta>0}\left[\delta \log d(1\lor |\log d|) + \P(\|\widetilde{\b\Upsilon}_\Pi - \b\Upsilon_\Pi\|_{\max}>\delta)\right].
\]
\end{theorem}

%\begin{align}
%\xi_0 &= (1+\max_i\|\b \theta_i\|_2) \sqrt{\log[n(l+1)]}T^{-1/2}+(1+\max_{i} \|\b\theta_i\|_1)\big[[\log(nT)]^{1/\gamma_2}\eta + \eta^2\big]\nonumber, \\
%\xi_1 &:= \sqrt{\frac{\log [n(l+1)]}{T}} +  [\log (nT)]^{1/\gamma_2}\eta + \eta^2\nonumber  
%\end{align}
%for the exponential case; and 

%in the polynomial case (Assumption \ref{Ass:Moments}(c)) and $g(n)= (\log n)^{1/\gamma_2}$ in the exponential case (Assumption \ref{Ass:Moments}(d)).

\section{Proof of the Theorems} \label{App: Main Results}

\subsection{Proof of Theorems \ref{T:LS} and  \ref{T:LS_exp}}

We first upper bound $|\widehat{R}_{i,t}-R_{i,t}|$. Recall that $\widehat{R}_{i,t}-R_{i,t}=(\widehat{\b\gamma}_i-\b\gamma_i)'\b X_{i,t}$. Then, by subsequent application of Hölder's inequality, we have
\begin{align}\label{E:LS_decomposition0}
    |\widehat{R}_{i,t}-R_{i,t}| 
    \leq \|\widehat{\b\gamma}_i-\b\gamma_i\|\|\b X_{i,t}\|\leq\|\widehat{\b\Sigma}_i^{-1}\| \|\widehat{\b v}_i\|\|\b X_{i,t}\|\leq k\|\widehat{\b\Sigma}_i^{-1}\| \|\widehat{\b v}_i\|_\infty\|\b X_{i,t}\|_\infty,
\end{align}
where $\widehat{\b\Sigma}_i:=\b X_i'\b X_i/T$ and $\widehat{\b v}_i:= \b X_i'\b R_i/T$. Therefore, 
\begin{equation}\label{E:LS_decomposition}
  \|\widehat{\b R}- \b R\|_{\max}\leq k\left(\max_{i}\left\|\widehat{\b\Sigma}_i^{-1}\right\|\right)\left(\max_{i,j}\left|\widehat{v}^{(j)}_i\right|\right)\left(\max_{i,t,j}\left|X^{(j)}_{i,t}\right|\right),
\end{equation}
where $\widehat{v}_i^{(j)}$ and $X_{i,t}^{(j)}$ denote the $j$-th component of $\widehat{\b v}_i$ and $\b X_{i,t}$ respectively. 

Under Assumption \ref{Ass:Moments}, $\viii{\|\widehat{\b\Sigma}_i^{-1}\|}_\psi\leq C$. By definition, $R_{i,t} = \b \lambda_i'\b F_t + U_{i,t}$, then 

\[
\viii{R_{i,t}}_\psi\leq \|\b \lambda_i\| \viii{\b F_t}_\psi + \viii{U_{i,t}}_\psi\leq (1+ \sqrt{r}\|\Lambda\|_{\max})\viii{\b Z_{i,t}} \leq (1+C\sqrt{r})C
\] 
under Assumptions \ref{Ass:Factor_Model} and \ref{Ass:Moments}. Also,  $\b X_{i,t}R_{i,t}$ is a function of $\b Z_{t}$. Then, $\{\b X_{i,t}R_{i,t}\}_t$ is a zero-mean sequence with strong mixing coefficient upper bounded by the strong mixing coefficient of $\{\b Z_t\}_t$. 

For the polynomial case, applying the union bound followed by Markov's inequality we conclude that 
\[
\max_{i}\left\|\widehat{\b\Sigma}_i^{-1}\right\|\lsim_\P n^{1/p}
\quad\textnormal{and}\quad 
\max_{i,t,j}\left|X^{(j)}_{i,t}\right|\lsim_\P (nkT)^{1/p}.
\]

Also, by the Cauchy-Schwartz inequality 
\[
\viii{X_{i,t}^{(j)}R_{i,t}}_{(p+\epsilon)/2}\leq \viii{X_{i,t}^{(j)}}_{p+\epsilon}\viii{R_{i,t}}_{p+\epsilon}\leq (1+C\sqrt{r})C^2.
\]
Then, by Lemma \ref{L:tail_bound_sup_poly}, $\max_{i,j}|\widehat{v}^{(j)}_i|\lsim_\P\frac{\mathscr{R}_\alpha\sqrt{r}(nk)^{2/p}}{\sqrt{T}}$. Plugging the last three probability bounds back on \eqref{E:LS_decomposition} yields the result for the polynomial case. 

Similarly, for the exponential case, the union bound followed by Lemma \ref{L:Orlicz_equivalence} gives us 
\[
\max_{i}\|\widehat{\b\Sigma}_i^{-1}\|\lsim_\P  (\log n)^{1/\gamma_2}\quad 
\textnormal{and}\quad \max_{i,t,j}|X^{(j)}_{i,t}|\lsim_\P [\log(nkT)]^{1/\gamma_2}.
\]
Also, by Lemma \ref{L:CS_exp}, 
\[
\viii{X_{i,t}^{(j)}R_{i,t}}_{e^{\gamma_2/2}}\lsim  \viii{X_{i,t}^{(j)}}_{e^{\gamma_2}}^2\lor \viii{R_{i,t}}_{e^{\gamma_2}}^2\lsim 1.
\] 
Hence, $\max_{i,j}|\widehat{v}^{(j)}_i|\lsim_\P\sqrt{\frac{\log (nk)}{T}}$ by Lemma \ref{L:tail_bound_sup_exp}. Plugging the last three probability bounds back on \eqref{E:LS_decomposition} yields the result for the exponential case, 

\subsection{Proof of Theorems \ref{T:FM} and \ref{T:FM_exp}}

The proof is an adaption of the proof of Theorem 4 and Corollary 1 in \cite{FLM2013}, henceforth FLM, to accommodate (i) the serial dependency (strong mixing sequences), (ii) polynomial tails, and (iii) the estimation error in the sample covariance matrix.  For part (a), we use expression (A.1) in \cite{Bai2003} to obtain the following identity
\begin{equation}\label{E:Bai_identity}
\widehat{\b F}_t - \b H \b F_t = \left(\frac{\b V}{n}\right)^{-1}\left[\frac{1}{T}\sum_{s=1}^T\widehat{\b F}_s\frac{\E(\b U_s'\b U_t)}{n} + \frac{1}{T}\sum_{s=1}^T\left(\widehat{\b F}_s\widetilde{\zeta}_{st} + \widehat{\b F}_s\widetilde{\eta}_{st} +\widehat{\b F}_s\widetilde{\xi}_{st}\right)\right],
\end{equation}
where $\b V$ is a $(r\times r)$ diagonal matrix whose diagonal is given by the eigenvalues of $\widehat{\b\Lambda}'\widehat{\b\Lambda}$ in decreasing fashion with $\widehat{\b\Lambda}:=\widehat{\b R}\widehat{\b F}/T$ ; and $\widetilde{\zeta}_{st},\widetilde{\eta}_{st}$ and $\widetilde{\xi}_{st}$ are defined before Lemma \ref{L:FM_Basic}.

By Assumptions \ref{Ass:Factor_Model}(d) and \ref{Ass:Moments} we have $\|\b R\|_{\max}\leq r\|\b\Lambda\|_{\max}\|\b F\|_{\max} +\|\b U\|_{\max}\lsim_\P g(nT)$. Applying Lemma \ref{L:compatibility_cond} we conclude that $\|\widehat{\b\Sigma} - \widetilde{\b\Sigma}\|_{\max} \lsim_\P \varrho_R[g(nT) + \varrho_R]\lsim_\P 1$. Finally, $\psi
g_\alpha(n)/\sqrt{T}\lsim 1$ also by assumption. Then,  $\|(\frac{\b V}{n})^{-1}\| \lsim_\P 1$ by Lemma \ref{L:10FLM}. Using the results (a)-(d) of Lemma \ref{L:9FLM} we can bound in probability each of the terms in brackets of \eqref{E:Bai_identity} in $\ell_2$ norm, uniformly in $t\leq T$. Result (a) follows since
\[
\max_{t\in[T]}\|\widehat{\b F}_t - \b H \b F_t\|\lsim_\P \left[\tfrac{1}{\sqrt{T}} + \frac{g(T)}{\sqrt{n}} + g(nT)\varrho_R\right],
\]
where $g(x)= x^{1/p}$ under Assumption (\ref{Ass:Moments}.c) and $g(x)=[\log (x)]^{1/\gamma_2}$ under Assumption (\ref{Ass:Moments}.d).

For part (b) we use the fact that $\widehat{\b\Lambda}:=\widehat{\b R}\widehat{\b F}/T$ and set $\widehat{\b F}'\widehat{\b F} = \b I_r$ to write
\begin{equation}\label{E:lambda_identity}
\widehat{\b\lambda}_i - \b H \b\lambda_i = \frac{1}{T}\sum_{t=1}^T\b H \b F_t  \widetilde{U}_{i,t} + 	\frac{1}{T}\sum_{t=1}^T\widehat{R}_{i,t}(\widehat{\b F}_t - \b H\b F_t) + \b H\left(\frac{1}{T}\sum_{t=1}^T\b F_t\b F_t' - \b I_r\right)\b\lambda_i.
\end{equation}
The first term can be upper bounded in $\ell_2$ norm, uniformly in $i\leq n $, by
\[
\sqrt{r}\|\b H\|\max_{i\leq n}\max_{j\leq r}\left|\frac{1}{T}\sum_{t=1}^T F_{jt} \widetilde{U}_{i,t}\right|\lsim_\P  g_1(n)/\sqrt{T} + \varrho_R,
\]
 where the equality follows from Lemma \ref{L:10FLM}(b) and (e), and $g_1(x)= \mathscr{R}_\alpha x^{2/p}$ under Assumption (\ref{Ass:Moments}.c) and $g_1(x)=\sqrt{\log (x)}$ under Assumption (\ref{Ass:Moments}.d). The $\ell_2$-norm of the second term is upper bounded uniformly in $i\leq n $ by
 \[
 \left(\max_{i\leq n}\frac{1}{T}\sum_{t=1}^T\widehat{R}_{i,t}^2 \frac{1}{T}\sum_{t=1}^T\| \widehat{\b F}_t - \b H\b F_t\|^2\right )^{1/2} \lsim_\P \left[\frac{1}{T} +(1/\sqrt{n}+)^2\right]^{1/2},
 \]
where the first term after the equality follows from Lemma \ref{L:10FLM}(d) together with the theorem's assumption and the second term from Lemma \ref{L:8FLM}(e). Finally, the last term of \eqref{E:lambda_identity} is upper bounded by
\[
\|\b H\|\|\max_{i\leq n}\b\lambda_i\|\left\|\frac{1}{T}\sum_{t=1}^T\b F_t\b F_t' - \b I_r\right\|\lsim_\P 1/\sqrt{T},
\]
where the last term is $\lsim_\P 1/\sqrt{T}$ by the maximum inequality and Assumption \ref{Ass:Moments}.
Plugging the last three displays back into \eqref{E:lambda_identity} yields result (b).

For part (c) we have $\|\widehat{\b U}-\b U\|_{\max} = \| \b\Lambda\b F' - \widehat{\b\Lambda}\widehat{\b F}'  +\widehat{\b R} - \b R\|_{\max} \leq \|  \widehat{\b\Lambda}\widehat{\b F}' - \b\Lambda\b F' \|_{\max}   +\|\widehat{\b R} - \b R\|_{\max}$. The last term is $\lsim_\P \varrho_R$ by assumption. For the first term we use the decomposition
\begin{align}\label{E:factor_decomposition}
\widehat{\b\lambda}_i'\widehat{\b F}_t - \b\lambda_i'\b F_t&=(\widehat{\b\lambda}_i - \b H\b\lambda_i)'(\widehat{\b F}_t - \b H\b F_t) +  (\b H\b\lambda_i)'(\widehat{\b F}_t - \b H\b F_t) \nonumber\\
 &\qquad+  (\widehat{\b\lambda}_i - \b H\b\lambda_i)'\b H\b F_t +\b\lambda_i'(\b H'\b H - \b I_r)\b F _t.
\end{align}
Therefore, we can upper bound the left hand side as
\begin{align*}
|\widehat{\b\lambda}_i'\widehat{\b F}_t - \b\lambda_i'\b F_t|&\leq\|\widehat{\b\lambda}_i - \b H\b\lambda_i\|\|\widehat{\b F}_t - \b H\b F_t\|  + \|\b H\b\lambda_i\|\|\widehat{\b F}_t - \b H\b F_t\| \\
 &\qquad+  \|\widehat{\b\lambda}_i - \b H\b\lambda_i\|\|\b H\b F_t\| +\|\b\lambda_i\|\|\b F _t\|\|\b H'\b H - \b I_r\|.
\end{align*}

Now, we bound in probability, uniformly in $i\leq n$ and $t\leq T$, each of the four terms above. The first one is given by parts (a) and (b). $\max_{i\leq n}\|\b H\b\lambda_i\|\leq \|\b H\| \max_{i\leq n}||\b \lambda_i\|\lsim_\P r\|\b\Lambda\|_{\max} \lsim_\P 1$ by Lemma \ref{L:10FLM}(b) and  Assumption \ref{Ass:Factor_Model}(d). Thus, the second term is bounded by part (a). For the third term, $\max_{t\leq T}\|\b H\b F_t\|\leq \|\b H\|\max_{t\leq T}||\b F_t\|\lsim_\P g(T)$ by Lemma \ref{L:10FLM}(b) and Assumption \ref{Ass:Factor_Model}. Finally, $\|\b H'\b H - \b I_r\| \lsim_\P 1/\sqrt{T} + 1/\sqrt{n} + \varrho_R$ by Lemma \ref{L:10FLM}(c). Hence, the last term is $\lsim_\P g(T)(1/\sqrt{T} + 1/\sqrt{n} + \varrho_R)$ by Assumptions \ref{Ass:Factor_Model}(d) and \ref{Ass:Moments}.

\subsection{Proof of Theorems \ref{T:Oracle_bounds} and \ref{T:Oracle_bounds_exp}}

For $i\in[n]$,  let $\m{Q}_i(\b a):=\|\widehat{\b U}_{i,\cdot}-\b a' \widehat{\b W}_{i,\cdot}\|_2^2/T$ for $\b a\in\R^d$. We have that $L(\widehat{\b\theta}_i)+\xi\|\widehat{\b\theta}_i\|_1\leq \m{Q}(\b \theta_i) +\xi\|\b \theta_i\|_1$ by definition of $\widehat{\b\theta}_i$. Also,  since $\m{Q}_i(\b\theta)$ is a quadratic function,  it implies that $(\widehat{\b\theta}_i -\b\theta_i)'\widehat{\b\Sigma}_i(\widehat{\b\theta}_i -\b\theta_i)\leq  2\widehat{\b E}_i'(\widehat{\b\theta}_i -\b\theta_i)+\xi(\|\b\theta_i\|_1 - \|\widehat{\b\theta}_i\|_1)$ where $\widehat{\b\Sigma}_i:= \widehat{\b W}_{i,\cdot} \widehat{\b W}_{i,\cdot}'/T$ and $\widehat{\b E}_i :=(\widehat{\b U}_{i,\cdot}-\b \theta_i' \widehat{\b W}_{i,\cdot})'\widehat{\b W}_{i,\cdot}/T$.  By Holder's inequality, we have $|\widehat{\b E}_i'(\widehat{\b\theta}_i -\b\theta_i)|\leq \|\widehat{\b E}_i\|_\infty\|\widehat{\b\theta}_i -\b\theta_i\|_1$. Then, for $\xi\geq 4\|\widehat{\b E}_i\|_\infty$ we have
 \begin{equation}\label{E:basic_inequality}
 (\widehat{\b\theta}_i -\b\theta_i)'\widehat{\b\Sigma}_i(\widehat{\b\theta}_i -\b\theta_i)\leq \xi/2\|\widehat{\b\theta}_i -\b\theta_i\|_1 +\xi(\|\b\theta_i\|_1 - \|\widehat{\b\theta}_i\|_1).
 \end{equation}
For any index set $\S\subseteq [n]$, by the decomposability of the $\ell_1$ norm (refer to Definition 1 in \cite{Negahban2012}) followed by the triangle inequality, we have $\|\widehat{\b\theta_i}\|_1= \|\widehat{\b\theta}_{i,\S}\|_1 +\|\widehat{\b\theta}_{i,\S^c}\|_1 \geq \|\b\theta_{i,\S}\|_1 - \|\widehat{\b\theta}_{i,\S} -\b\theta_{i,\S}\|_1+\|\widehat{\b\theta}_{i,\S^c}\|_1 $ and $\|\widehat{\b\theta}_i -\b\theta_i\|_1 = \|\widehat{\b\theta}_{i,\S} -\b\theta_{i,\S}\|_1 + \|\widehat{\b\theta}_{i,\S^c} -\b\theta_{i,\S^c}\|_1\leq \|\widehat{\b\theta}_{i,\S} -\b\theta_{i,\S}\|_1 + \|\widehat{\b\theta}_{i,\S^c} -\b\theta_{i,\S^c}\|_1  $. Plugging it back in \eqref{E:basic_inequality} yields
 \begin{equation}\label{E:basic_inequality_after_decomposition}
 2(\widehat{\b\theta}_i-\b\theta_i)'\widehat{\b\Sigma}_i(\widehat{\b\theta}_i -\b\theta_i) + \xi\|\widehat{\b\theta}_{i,\S^c} -\b\theta_{i,\S^c}\|_1\leq
  3\xi\|\widehat{\b\theta}_{i,\S} -\b\theta_{i,\S}\|_1 +4\xi\|\b\theta_{i,\S^c}\|_1.
 \end{equation}
We then conclude that any minimizer of \eqref{E:LASSO_obj_fun} obeys $\widehat{\b\theta}_i -\b\theta_i \in \C(\m{S}_i,3):=\{\b x\in\R^d:\|\b x_{\S^c_i}\|_1\leq 3\|\b x_{\S_i}\|_1\|_1\}$ where $\S_i:=\{j:\theta_{i,j}\neq 0\}$ is the support of $\b\theta_i$.

Recall the definition of the compatibility constant appearing in \cite{vandegeer2009}, reproduced below for convenience.
\begin{Def}\label{D:GIF}
For an $n\times n$ matrix $\b{M}$, a set $\S\subseteq\{1,\ldots,n\}$  and a scalar $\zeta \geq 0$, the compatibility constant is given by
\begin{equation}\label{E:GIF}
\kappa(\b{M},\S,\zeta):=\inf\left\{\frac{\|\b{x}\|_{\b{M}}\sqrt{|\S|}}{\|\b{x}_\S\|_1}:\b{x}\in\R^n:\|\b{x}_{\S^c}\|_1\leq\xi\|\b{x}_{\S}\|_1 \right\},
\end{equation}
where $\|\b{x}\|_{\b{M}} = \sqrt{\b{x}'{\b{M}} \b{x}}$.
Moreover, we say that $(\b{M},\S,\zeta)$ satisfies the compatibility condition if $\kappa(\b M,\S,\zeta)>0$.
\end{Def}
Notice that the square of the compatibility constant is close related to the minimum of the $\ell_1$-norm of the eigenvalues of $\b M$, restricted to a cone in $\R^n$. By definition of the compatibility constant $\widehat{\kappa}_i:=\kappa( \widehat{\b\Sigma}_i,\m{S}_0,3)$, we have that $\|\widehat{\b\theta}_{i,\S} -\b\theta_{i,\S}\|_1\leq \sqrt{(\widehat{\b\theta}_i -\b\theta_i)'\widehat{\b\Sigma}_i(\widehat{\b\theta}_i -\b\theta_i)} \sqrt{|\S_i|}/\widehat{\kappa}_i$. Apply this inequality to \eqref{E:basic_inequality_after_decomposition} and use the fact that $4ab<a^2 + 4b^2$ for non-negative $a,b\in\R$, to obtain
\begin{equation}\label{E:oracle_hat}
(\widehat{\b\theta}_i -\b\theta_i)'\widehat{\b\Sigma}_i(\widehat{\b\theta}_i -\b\theta_i) + \xi\|\widehat{\b\theta}_{i} -\b\theta_i\|_1\leq 4\xi^2|\S_i|/\widehat{\kappa}_i^2,
\end{equation}
provided that $\xi\geq 4\|\widehat{\b E}_i\|_\infty$.

Let $\b\Sigma_i:=\E(\b W_{i,t} \b W_{i,t}')$ and $\kappa_i:=\kappa( \b\Sigma_i,\m{S}_0,3)$ .  Note that $\min_i\kappa_i^2\geq b^2>0$ under Assumption (\ref{Ass:Moments}.g) because
\[\lambda_{\min}\left[\E\left(\b{\m{U}}_t\b{\m{U}}_t'\right)\right]\leq\min_i \lambda_{\min}(\b\Sigma_i)\leq \min_i \min_{\b x\in\C(\m{S}_i,3)}\left\{\frac{\b x'\b\Sigma_i \b x}{\b x'\b x}\right\}\leq \min_i \kappa_i^2,\]
where in the first inequality we use Cauchy interlacing theorem and in the last on we use the fact that $\b x' \b x \geq \b x_{\m{S}_i}' \b x_{\m{S}_i}$ and the bound $\|\b x_{\m{S}_i}\|_1\leq \sqrt{|\m{S}_i}\|\b x_{\m{S}_i}\|_2$.
Also, we may use Lemma \ref{L:compatibility_cond} with $\zeta =3$ and $\alpha =1/2$ to assert that if $\|\widehat{\b\Sigma}_i-\b\Sigma_i\|_{\max}\leq \xi_1$ for some $\xi_1>0$, such that $32\xi_1 |\m{S}_i|/\kappa_i^2\leq 1$ then, $\widehat\kappa_i^2 \geq \kappa_i^2/2$ and $\b x'\b\Sigma_i \b x/2\leq  \b x'\widehat{\b \Sigma}_i \b x$ for $\b x\in\C(\m{S}_i,3)$.  Use the last three inequalities in \eqref{E:oracle_hat} to conclude that, for $i\in[n]:$
\begin{equation}\label{E:oracle}
\tfrac{1}{2}(\widehat{\b\theta}_i -\b\theta_i)'\b\Sigma_i(\widehat{\b\theta}_i -\b\theta_i) + \xi\|\widehat{\b\theta}_{i} -\b\theta_i\|_1\leq 8\xi^2|\S_i|/b^2,
\end{equation}
provided  that $\xi\geq 4\|\widehat{\b E}_i\|_\infty$ and $\|\widehat{\b\Sigma}_i-\b\Sigma_i\|_{\max}\leq \xi_1$ such that $32\xi_1 |\m{S}_i|/b\leq 1$.

\subsubsection{Probabilistic Bounds}

We now bound in probability the events $\{\xi\geq 4\|\widehat{\b E}_i\|_\infty\}$ and $\{\|\widehat{\b\Sigma}_i-\b\Sigma_i\|_{\max}\leq \xi_1\}$,  uniformly in $i\in[n]$. For the former, by definition of $\b W_{i,t}$,  we have $\max_{i,t}\|\b W_{i,t}\|_\infty\leq \|\b U\|_{\max}$ and $\max_{i,t}\|\widehat{\b W}_{i,t} - \b W_{i,t}\|_{\infty}\leq \|\widehat{\b U} - \b U\|_{\max}$.  Also,  by the definition of $V_{i,t}$, the triangle inequality followed by Hölder's inequality,  we have $\max_{i,t}|V_{i,t}|\leq \max_{i,t}|U_{i,t}| +\max_{i,t}\|\b\theta_i\|_1\b|\b W_{i,t}\|_\infty\leq (1+\max_i\|\b\theta_i\|_1)\|\b U\|_{\max}$.  Similarly, we conclude $\max_{i,t}|\widehat{V}_{i,t} - V_{i,t}| \leq  (1+\max_i\|\b\theta_i\|_1)\|\widehat{\b U} - \b U\|_{\max}$.

First, we decompose
\[T\widehat{\b E}_i=\b W_{i,\cdot} \b V_{i,\cdot} + \b W_{i,\cdot} (\widehat{\b V}_{i,\cdot}-\b V_{i,\cdot})+(\widehat{\b W}_{i,\cdot}  -\b W_{i,\cdot} )\b V_{i,\cdot}+(\widehat{\b W}_{i,\cdot}  -\b W_{i,\cdot} )(\widehat{\b V}_{i,\cdot}-\b V_{i,\cdot}),\]
and bound each term individually as
\begin{align*}
\max_i\|\b W_{i,\cdot} (\widehat{V}_{i,\cdot}-V_{i,\cdot})/T\|_\infty&\leq\max_{i,t}\|\b W_{i,t}\|_\infty\max_{i,t}|\widehat{V}_{i,t} - V_{i,t}| \\ &\leq (1+\max_{i} \|\b\theta_i\|_1)\|\b U\|_{\max}\|\widehat{\b U} - \b U\|_{\max},
\end{align*}
\begin{align*}
\max_i\|(\widehat{\b W}_{i,\cdot}  -\b W_{i,\cdot} )V_{i,\cdot}/T\|_\infty&\leq\max_{i,t}\|\widehat{\b W}_{i,t}  -\b W_{i,t} \|_\infty\max_{i,t}| V_{i,t}| \\ &\leq (1+\max_{i} \|\b\theta_i\|_1)\|\b U\|_{\max}\|\widehat{\b U} - \b U\|_{\max},
\end{align*}
and
\begin{align*}
\max_i\|(\widehat{\b W}_{i,\cdot}  -\b W_{i,\cdot} )(\widehat{V}_{i,\cdot}-V_{i,\cdot})/T\|_\infty&\leq\max_{i,t}\|\widehat{\b W}_{i,t} -\b W_{i,t}\|_\infty\max_{i,t}|\widehat{V}_{i,t} - V_{i,t}| \\ &\leq (1+\max_{i} \|\b\theta_i\|_1)\|\widehat{\b U} - \b U\|_{\max}^2.
\end{align*}

Then,
\begin{align}\label{E:deviation_bound_decomposition}
\max_{i}\|\widehat{\b E}_i\|_\infty&\leq \max_{i}\|\b W_{i,\cdot} \b V_{i,\cdot} /T\|_\infty\\
&\qquad +(1+\max_{i} \|\b\theta_i\|_1)(2\|\b U\|_{\max}\|\widehat{\b U} - \b U\|_{\max} +\|\widehat{\b U} - \b U\|_{\max}^2 )\nonumber.    
\end{align}

We now bound $\viii{\b W_{i,t}}_\psi $ and $\viii{V_{i,t}}_\psi$. For the former, $\viii{\b W_{i,t}}_\psi \leq \sup_t \viii{\b U_{\cdot,t}}_\psi$ and for the latter
\begin{align}
    \viii{V_{i,t}}_\psi &= \viii{U_{i,t}-\b\theta_i' \b W_{i,t}}_\psi\nonumber\\
    &\leq  \viii{U_{i,t}}+\viii{\b\theta_i' \b W_{i,t}}_\psi\nonumber\\
    &\leq  \viii{U_{i,t}}+\|\b \theta_i\|_2\viii{ \b W_{i,t}}_\psi\nonumber\\
    &\leq  (1+\|\b \theta_i\|_2)\sup_t\viii{ \b U_{\cdot,t}}_\psi.\label{E:V_bound}
\end{align}

By definition, we have that $\b W_{i,t}V_{i,t} = f_i(\b U_{\cdot,t},\dots, \b U_{\cdot, t-l})$ for some mapping $f_i$, for each $i\in[n]$ and $t\in[T]$. Then, by Lemma \ref{L:mixing_lag}, we have that the strong mixing coefficients of the sequence $\{W_{i,t}^{(j)}V_{i,t}\}_t$ are such that $\widetilde{\alpha}_{m}\leq \alpha_{(m-l)\lor 0}$, where $W_{i,t}^{(j)}$ denote the $j$-th component of $\b W_{i,t}$ for $j\in[d]$. 

Under Assumption \ref{Ass:Moments}(c), by the Cauchy-Schwartz inequality and the last two bounds, we have
\begin{align*}
\viii{W_{i,t}^{(j)}V_{i,t}}_{p/2+\epsilon/2}&\leq \viii{W_{i,t}^{(j)}}_{p+\epsilon}\viii{V_{i,t}}_{p+\epsilon}\\
&\leq (1+\|\b \theta_i\|_2)\sup_t\viii{ \b U_{\cdot,t}}_{p+\epsilon}^2\leq C^2\left(1+\max_i\|\b \theta_i\|_2\right).
\end{align*}

Therefore, by Lemma \ref{L:tail_bound_sup_poly} followed by Lemma \ref{L:mixing_lag}, we conclude that for $p\in[4,\infty)$
\[
\max_{i,j}|\sum_{t=1}^T W_{i,t}^{(j)}V_{i,t}| \lsim_\P\left(1+\max_i\|\b \theta_i\|_2\right) \mathscr{L}_\alpha\left[n(l+1)\right]^{2/p}\sqrt{T},
\]
where $\mathscr{L}_\alpha:=\big(\sum_{t=0}^{T-1}(m+1)^{(p/2)-2}\alpha_{(m-l)\lor 0}^{1-p/(p+\epsilon)}\big)^{2/p}$. By assumption $\|\widehat{\b U}- \b U\|_{\max} \lsim_\P\varrho_U$, and the union bound followed by Markov's inequality give us $\|\b U\|_{\max} \lsim_\P nT^{1/p}$. Then, from \eqref{E:deviation_bound_decomposition}, we obtain
\begin{align}
\max_{i}\|\widehat{\b E}_i\|_\infty&\lsim_\P (1+\max_i\|\b \theta_i\|_2) \mathscr{L}_\alpha\big[n(l+1)\big]^{2/p}T^{-1/2}\\
&\qquad +(1+\max_{i} \|\b\theta_i\|_1)\left[(nT)^{1/p}\varrho_U + \varrho_U^2\right]\nonumber=:\xi_0.  
\end{align}

Similarly, under Assumption \ref{Ass:Moments}(d) and Lemma \ref{L:CS_exp} we have
\begin{align*}
\viii{W_{i,t}^{(j)}V_{i,t}}_{e^{\gamma_2/2}}&\lsim\viii{W_{i,t}^{(j)}}_{e^{\gamma_2}}^2\lor \viii{V_{i,t}}_{e^{\gamma_2}}^2\\
&\lsim  (1+\|\b \theta_i\|_2)\sup_t\viii{ \b U_{\cdot,t}}_{e^{\gamma_2}}
\lsim(1+\max_i\|\b \theta_i\|_2).
\end{align*}

Therefore, by Lemma \ref{L:tail_bound_sup_exp} we conclude that 
\[
\max_{i,j}\left|\sum_{t=1}^T W_{i,t}^{(j)}V_{i,t}\right| \lsim_\P(1+\max_i\|\b \theta_i\|_2) \sqrt{T\log [n(l+1)]}.
\]
The union bound followed by Lemma \ref{L:Orlicz_equivalence} give us $\|\b U\|_{\max} \lsim_\P [\log (nT)]^{1/\gamma_2}$. Finally, we use \eqref{E:deviation_bound_decomposition} once again to obtain
\begin{align}
\max_{i}\|\widehat{\b E}_i\|_\infty&\lsim_\P (1+\max_i\|\b \theta_i\|_2) \sqrt{\log[n(l+1)]}T^{-1/2}\\
&\qquad +(1+\max_{i} \|\b\theta_i\|_1)\left\{[\log(nT)]^{1/\gamma_2}\varrho_U + \varrho_U^2\right\}\nonumber:=\xi_0.  
\end{align}

We now bound $\max_i\|\widehat{\b\Sigma}_i-\b\Sigma_i\|_{\max}$. Let $\widetilde{\b \Sigma}_i:=\b W_{i,\cdot}\b W_{i,\cdot}'/T$. From the definition of $\b W_{i,t}$ and Lemma \ref{L:compatibility_cond} we have that 
\[\max_i\|\widehat{\b\Sigma}_i-\widetilde{\b\Sigma}_i\|_{\max}\leq \|\widehat{\b U} - \b U\|_{\max}(2\|\b U\|_{\max}  + \|\widehat{\b U} - \b U\|_{\max}).\]
Also by definition, we have that $\b W_{i,t}\b W_{i,t}' = g_i(\b U_{\cdot,t},\dots, \b U_{\cdot, t-l})$ for some function $g_i$ for each $i\in[n]$ and $t\in[T]$. Then, by Lemma \ref{L:mixing_lag} we have that the strong mixing coefficients of the sequence $\{\b W_{i,t}\b W_{i,t}'\}_t$, $\check{\alpha}_{m}\leq \alpha_{(m-l)\lor 0}$. Let $W_{i,t}^{(j)}$ denote the $j$-th component of $\b W_{i,t}$ for $j\in[d]$.

Under Assumption \ref{Ass:Moments}(c) and  the Cauchy-Schwartz inequality we have
\[
\viii{W_{i,t}^{(j_1)}W_{i,t}^{(j_2)}}_{p/2+\epsilon/2}\leq \viii{W_{i,t}^{(j_1)}}_{p+\epsilon}\viii{W_{i,t}^{(j_2)}}_{p+\epsilon}\leq \sup_t\viii{ \b U_{\cdot,t}}_{p+\epsilon}^2\leq C^2;\qquad j_1,j_2\in[d].
\]

Therefore, by Lemma \ref{L:tail_bound_sup_poly} followed by Lemma \ref{L:mixing_lag}, we conclude that for $p\in[4, \infty)$
\[\max_{i,j_1,j_2}\left|\sum_{t=1}^T W_{i,t}^{(j)} W_{i,t}^{(j_2)} -\E\left(W_{i,t}^{(j)} W_{i,t}^{(j_2)}\right)\right| \lsim_\P \mathscr{L}_\alpha\left[n(l+1)\right]^{4/p}\sqrt{T},\]
where $\mathscr{L}_\alpha$ is defined above. Hence, by the triangle inequality
\begin{align}
    \max_i\|\widehat{\b\Sigma}_i-\b\Sigma_i\|_{\max}&\leq  \max_i\|\widetilde{\b\Sigma}_i-\b\Sigma_i\|_{\max} + \max_i\|\widehat{\b\Sigma}_i-\widetilde{\b\Sigma}_i\|_{\max}\\
    &\lsim_\P \frac{\mathscr{L}_\alpha\big[n(l+1)\big]^{4/p}}{\sqrt{T}} +  (nT)^{1/p}\varrho_U + \varrho_U^2=:\xi_1\nonumber.
\end{align}

Similarly, under Assumption \ref{Ass:Moments}(d) and Lemma \ref{L:CS_exp} we have, for $j_1,j_2\in[d]$, 
\[
\viii{W_{i,t}^{(j_1)}W_{i,t}^{(j_2)}}_{e^{\gamma_2/2}}\lsim \viii{W_{i,t}^{(j_1)}}_{e^{\gamma_2}}\lor \viii{W_{i,t}^{(j_2)}}_{e^{\gamma_2}}\leq C_{\gamma_2}\sup_t\viii{ \b U_{\cdot,t}}_{e^{\gamma_2}}^2\leq C_{\gamma_2}C.
\]
Thus, by Lemma \ref{L:tail_bound_sup_exp} we conclude that 
\[\max_{i,j_1,j_2}\left|\sum_{t=1}^T W_{i,t}^{(j)} W_{i,t}^{(j_2)} -\E\left(W_{i,t}^{(j)} W_{i,t}^{(j_2)}\right)\right|\lsim_\P \sqrt{T\log [n(l+1)]}.\]
Therefore, by the triangle inequality
\begin{align}
    \max_i\|\widehat{\b\Sigma}_i-\b\Sigma_i\|_{\max}&\leq  \max_i\|\widetilde{\b\Sigma}_i-\b\Sigma_i\|_{\max} + \max_i\|\widehat{\b\Sigma}_i-\widetilde{\b\Sigma}_i\|_{\max}\\
    &\lsim_\P \sqrt{\frac{\log [n(l+1)]}{T}} +  [\log (nT)]^{1/\gamma_2}\varrho_U + \varrho_U^2=:\xi_1\nonumber.
\end{align}

\subsection{Proof of Theorems \ref{T:prediction_bounds} and \ref{T:prediction_bounds_exp}}
Decompose the prediction error as
\begin{align*}
\widehat{Y}_{i,t}-\widetilde{Y}_{i,t}&=(\widehat{\b\gamma}_i-\b\gamma_i)'\b X_{i,t} +\widehat{\b\lambda}_i'\widehat{\b P}\widehat{\b G}_t - \b\lambda_i'\b P\b G_t +\widehat{\b\theta}_i' \widehat{\b W}_{i,t} - \b\theta_i'\b W_{i,t} \\
&= (\widehat{\b\gamma}_i-\b\gamma_i)'\b X_{i,t}+\big[\b\lambda_i'\b P + (\widehat{\b\lambda}_i'\widehat{\b P}-\b\lambda_i'\b P)\big](\widehat{\b G}_t-\b G_t)- (\widehat{\b\lambda}_i'\widehat{\b P}-\b\lambda_i'\b P)\b G_t\\
&\qquad +\big[\b\theta_i +(\widehat{\b\theta}_i- \b\theta_i) \big]' (\widehat{\b W}_{i,t} - \b W_{i,t}) + (\widehat{\b\theta}_i- \b\theta_i)'\b W_{i,t},
\end{align*}
and
\begin{align*}
    \widehat{\b\lambda}_i'\widehat{\b P}-\b\lambda_i'\b P &=  \big[\b\lambda_i +(\widehat{\b\lambda}_i - \b\lambda_i)\big](\widehat{\b P} - \b P) + (\widehat{\b\lambda}_i - \b\lambda_i)\b P.
\end{align*}
For the first term we have
\begin{align*}
   \max_{i\in[n]}\left|(\widehat{\b\gamma}_i-\b\gamma_i)'\b X_{i,t}\right|\leq \max_{i\in[n]} \|\widehat{\b\gamma}_i-\b\gamma_i\|_1\max_{i\in[n]}\|\b X_{i,t}\|_\infty \lsim_\P \varrho_\gamma\max_{i\in[n]}\|\b X_{i,t}\|_\infty).
\end{align*}
For the second term, 
\[
\max_{i\in[n]}|\widehat{\b\lambda}_i'\widehat{\b P}-\b\lambda_i'\b P |\lsim_\P (\varrho_ P + \varrho_\lambda)
\] 
and, therefore,  
\[
\max_{i\in[n]}|\widehat{\b\lambda}_i'\widehat{\b  P}\widehat{\b G}_t - \b\lambda_i'\b P\b G_t|\lsim_\P (\varrho_ P + \varrho_\lambda +\varrho_F).
\]

Finally, for the last term, Hölder's inequality yields
\begin{align*}
    |\widehat{\b\theta}_i' \widehat{\b W}_{i,t} - \b\theta_i'\b W_{i,t}|
    &\leq(\|\b\theta_i\|_1 +\|\widehat{\b\theta}_i- \b\theta_i\|_1 ) \|\widehat{\b W}_{i,t} - \b W_{i,t}\|_\infty + \|\widehat{\b\theta}_i- \b\theta_i\|_1\|\b W_{i,t}\|_\infty.
\end{align*}
As a consequence, 
\[\max_{i\in[n]}|\widehat{\b\theta}_i' \widehat{\b W}_{i,t} - \b\theta_i'\b W_{i,t}|\lsim_\P \max_{i\in[n]} \|\b \theta_i\|_1\varrho_U + \varrho_\theta\max_{i\in[n]}\|\b W_{i,t}\|_\infty
\]
and
\begin{align*}
\max_{i\in[n]}|\widehat{Y}_{i,t}-\widetilde{Y}_{i,t}|&\lsim_\P \varrho_\gamma\max_{i\in[n]}\|\b X_{i,t}\|_\infty + \varrho_U\max_{i\in[n]}  \|\b \theta_i\|_1 + \varrho_\theta\max_{i\in[n]}\|\b W_{i,t}\|_\infty + \varrho_ P + \varrho_\lambda +\varrho_F.
\end{align*}
Finally, $\max_{i\in[n]}\|\b X_{i,t}\|_\infty\lsim_\P n^{1/p}$ in the polynomial case and $\max_{i\in[n]}\|\b X_{i,t}\|_\infty\lsim_\P (\log n)^{1/\gamma_2}$ in the exponential one. Similarly, $\max_{i\in[n]}\|\b W_{i,t}\|_\infty\lsim_\P d_W^{1/p}$ in the polynomial case and $\max_{i\in[n]}\|\b W_{i,t}\|_\infty\lsim_\P(\log d_W)^{1/\gamma_2}$ in the exponential one.

\subsection{Proof of Theorems \ref{T:inference_cov} and \ref{T:inference_cov_exp}}

For $\m{D}\in[n]^2$,  let $\b N_t:=(N_{1,t},\dots,N_{d,t})' :=\left[U_{i,t}U_{j,t} - \E \left(U_{i,t}U_{j,t}\right):(i,j)\in\m{D}\right]$ and $\widehat{\b N}_t:=\left[\widehat{U}_{i,t}\widehat{U}_{j,t} - \E\left(U_{i,t}U_{j,t}\right):(i,j)\in\m{D}\right]$  where $d:=|\m{D}|$.  Then,
\[
\widetilde{\b J} = \frac{1}{\sqrt{T}}\sum_{t\in[T]}\b N_t\quad\textnormal{and}\quad \widehat{\b J} = \frac{1}{\sqrt{T}}\sum_{t\in[T]}\widehat{\b N}_t.
\]
For the polynomial case (Theorems \ref{T:inference_cov}), by Cauchy Schwartz inequality and Assumption \ref{Ass:Moments}, \[
\viii{U_{i,t}U_{j,t}}_{(p+\epsilon)/2}\leq\viii{U_{i,t}}_{p+\epsilon}\viii{U_{j,t}}_{p+\epsilon}\leq C^2.
\]
Then, $\viii{N_{j,t}}_{(p+\epsilon)/2}\leq 2C^2$ for $j\in[d]$; for the exponential case (Theorem \ref{T:inference_cov_exp}). By Lemma \ref{L:CS_exp}  and Assumption \ref{Ass:MomentsExp}, $\viii{U_{i,t}U_{j,t}}_{e^{\gamma_2/2}}\lsim 1$. Then, $\viii{N_{j,t}}_{e^{\gamma_2/2}}\lsim 1$ for $j\in[d]$. Also, the mixing coefficient of $\{\b N_t,t\in[T]\}$ are upper bounded by those of $\{\b U_t,t\in[T]\}$ by Lemma \ref{L:mixing_lag}.
Therefore, we can apply Theorem \ref{T:CLT_HD}(a) to bound $\rho(\widetilde{\b J},\b{\m{G}})$ with ``$p=p/2$'' and ``$\epsilon= \epsilon/2$'' for the polynomial case and the first result of Theorem \ref{T:inference_cov} follows. Similarly,  we can apply Theorem \ref{T:CLT_HD}(b) to bound $\rho(\widetilde{\b J},\b{\m{G}})$ with ``$\gamma_1=\gamma_1$'' and ``$\gamma_2= \gamma_2/2$'' for the exponential case and the first result of Theorem \ref{T:inference_cov_exp} follows.

By triangle inequality followed Lemma \ref{L:useful_decompostion} (applied twice) and Corollary 1 in \cite{cck_anti-concentration} we have
\begin{align*}
    \rho(\widehat{\b J},\b{\m{G}})&\leq \rho(\widehat{\b J}, \widetilde{\b J}) + \rho(\widetilde{\b J},\b{\m{G}})\\
    &\lsim \rho(\widetilde{\b J},\b{\m{G}}) + \inf_{\delta_1>0}\left\{\Delta(\b{\m{G}},\delta_1) + \P(\|\widehat{\b J} - \widetilde{\b J}\|_\infty>\delta_1)\right\}\\
    &\lsim\rho(\widetilde{\b J},\b{\m{G}}) + \inf_{\delta_1>0}\left\{\delta_1\sqrt{1\lor \log d} + \P(\|\widehat{\b J} - \widetilde{\b J}\|_\infty>\delta_1)\right\},
\end{align*}
where $\Delta(\cdot,\cdot)$ is defined by \eqref{E:levy_coef}. Recall that $\b\Upsilon$ is the variance of $\b N_t$ and  $ \b{\m{G}}^*|\b X,\b Y\sim N(\b 0,\widehat{\b\Upsilon})$. Then, on the event $\{\|\widehat{\b\Upsilon} - \b\Upsilon\|_{\max}\leq \delta_2\}$ for $\delta_2>0$ by Theorem 1.1 in \cite{fang_koike} (see also Lemma 2.1 in \cite{clt_hd} for this particular application)
\[\sup_{t\in\R}|\P(\|\b{\m{G}}\|_\infty\leq t)-\P(\| \b{\m{G}}^*\|_\infty\leq t|\b X,\b Y)|\lsim\delta_2 \log d(1\lor |\log d|),\]
and
\[\rho(\b{\m{G}}, \b{\m{G}}^*)\lsim \inf_{\delta_2>0}\left\{\delta_2 \log d(1\lor |\log d|) + \P(\|\widehat{\b\Upsilon} - \b\Upsilon\|_{\max}>\delta_2)\right\}.\]

\subsection{Proof of Corollary \ref{C:inference_cov}}

Since $S^*|\b X,\b Y$ has no point-mass, we have that $\tau = \P(S^*\leq c^*(\tau)|\b X,\b Y)$ almost surely for every $\tau\in(0,1)$ and $\m{D}\in[n]^2$. Then, under the Null, using the notation as in the proof of Theorem \ref{T:inference_cov}, we write for $\delta_1,\delta_2>0$:
\begin{align*}
\sup_\m{D}\sup_\tau\left|\P\left[S\leq c^*(\tau)\right] -\tau\right| &=  \sup_\m{D}\sup_\tau\left|\P\left[S\leq c^*(\tau)\right] -\P(S^*\leq c^*(\tau)\right|\\
&\leq  \sup_\m{D}\rho(\widehat{\b J}, \b Z^*)\\
&\lsim \varrho_\Sigma + \delta_1\sqrt{1\lor \log n} + \sup_\m{D}\P(\|\widehat{\b J} - \widetilde{\b J}\|_\infty>\delta_1)\\
&\quad + \delta_2 \log n(1\lor |\log n|) + \sup_\m{D}\P(\|\widehat{\b\Upsilon} - \b\Upsilon\|_{\max}>\delta_2).
\end{align*}

Let $\gamma_1$ and $\gamma_2$ denote the rate appearing on Lemmas \ref{L:mean_bound} and \ref{L:NW_bound}, respectively, with $\varrho_U$ equals to the rate of Collorary \ref{C:LS+FM} and $\m{D}=[n]^2$ . By assumption, $\varrho_\Sigma\lor  (\log n)^{3/2}\gamma_1\lor(\log n)^3\gamma_2 = o(1)$ as  $T,n\to\infty$ when we set $\delta_1=\gamma_1\log n$ and $\delta_2=\gamma_2\log n$. Therefore, all the terms in the last expression vanish in probability, and the result follows.

\subsection{Proof of Theorems \ref{T:inference_partial_cov} and \ref{T:inference_partial_cov_exp}}

The proof is parallel to the proof of Theorem \ref{T:inference_cov}. We outline the main differences. For $\m{D}\in[n]^2$,  let $\b K_t:=(K_{1,t},\dots,K_{d,t})' :=\left[V_{i,j,t}V_{j,i,t} - \E (V_{i,j,t}V_{j,i,t}):(i,j)\in\m{D}\right]$ and $\widehat{\b K}_t:=\left[\widehat{V}_{i,j,t}\widehat{V}_{j,i,t} - \E (V_{i,j,t}V_{j,i,t}):(i,j)\in\m{D}\right]$  where $d:=|\m{D}|$.  Then,
\[
\widetilde{\b Q} = \tfrac{1}{\sqrt{T}}\sum_{t\in[T]}\b K_t\quad\textnormal{and}\quad \widehat{\b Q} = \tfrac{1}{\sqrt{T}}\sum_{t\in[T]}\widehat{\b K}_t.\]
For the polynomial case (Theorem \ref{T:inference_partial_cov}), by Cauchy Schwartz inequality, the bound in \eqref{E:V_bound} with $\b W_{i,t} = \b U_{-ij,t}$, and Assumption \ref{Ass:Moments}, we write \[
\viii{V_{i,j,t}V_{j,i,t}}_{(p+\epsilon)/2}\leq\viii{V_{i,j,t}}_{p+\epsilon}\viii{V_{j,i,t}}_{p+\epsilon}\leq C^2_\chi,
\]
where $C_\chi:=(1+\max_{i,j\in\m{D}}\|\b \chi_{i,j}\|_2)C$. Then, $\viii{K_{j,t}}_{(p+\epsilon)/2}\leq 2C^2_\chi$ for $j\in[d]$. Also, the mixing coefficient of $\{\b K_t,t\in[T]\}$ are upper bounded by those of $\{\b U_t,t\in[T]\}$ by Lemma \ref{L:mixing_lag}.
Therefore, we can apply Theorem \ref{T:CLT_HD}(a) to bound $\rho(\widetilde{\b Q},\b{\m{H}})$ with ``$p=p/2$'' and ``$\epsilon= \epsilon/2$'' and the first result of Theorem \ref{T:inference_partial_cov} follows. The other two results can be obtained as in the proof do Theorem \ref{T:inference_cov}, with $\widehat{\b J}$, $\b{\m{G}}$, and $\b{\m{G}}^*$ replaced by $\widehat{\b Q}$, $\b{\m{H}}$, and $\b{\m{H}}^*$, respectively.

For the exponential case (Theorem \ref{T:inference_partial_cov_exp}), we replace the Cauchy Schwartz inequality and Assumption \ref{Ass:Moments} used above by Lemma \ref{L:CS_exp} and Assumption \ref{Ass:MomentsExp}, respectively, and conclude that 
\[
\viii{V_{i,j,t}V_{j,i,t}}_{e^{\gamma_2/2}}\lsim\viii{V_{i,j,t}}_{e^{\gamma_2}}^2\lor \viii{V_{j,i,t}}_{e^{\gamma_2}}^2\leq C^2_\chi. 
\]
Finally, we apply Theorem \ref{T:CLT_HD}(b) to bound $\rho(\widetilde{\b Q},\b{\m{H}})$ with ``$\gamma_1 = \gamma_1$'' and ``$\gamma_2= \gamma_2/2$'' and the first of Theorem \ref{T:inference_partial_cov_exp} follows. 

\subsection{Proof of Corollary \ref{C:inference_partial_cov}}

The proof is identical to the proof of Corollary \ref{C:inference_cov} with Lemmas \ref{L:mean_bound} and \ref{L:NW_bound} replaced by Lemmas \ref{L:V_bond} and \ref{L:NW_bound_partial}, respectively.

\section{High-Dimension Central Limit Theorem for Strong Mixing sequences}

The setup for this section is the following.  Let $\{\b X_1,\dots, \b X_T\}$ be a sequence of zero-mean random vectors taking value in $\R^d$, whose strong mixing coefficient we denote by $\{\alpha_n:n\in\N\}$. Define $\b S=T^{-1/2}\sum_{t=1}^T\b X_t$ and let $\b Z\sim N(0,\b\Sigma)$ where $\b\Sigma = \E(\b S \b S')$. The next result upper bounds
\[
\rho(\b S, \b Z):=\sup_{A\in\m{R}}|\P( \b S\in A) - \P(\b Z \in A)|,
\]
where $\m{R}$ is the class of all rectangles in the from $\bigtimes_{j=1}^d (a_j,b_j]$ for some $-\infty\leq a_j\leq b_j\leq \infty$ and $j\in[d]$. We also define for $s>0$
\begin{equation}\label{E:levy_coef}
    \Delta(\b Z,s):= \sup_{z\in\R^d}\{\P(\b Z\leq z+s) - \P(\b Z\leq z)\}.
\end{equation}

Note that we may assume, without loss of generality, that all the elements on the diagonal of $\b\Sigma$ are positive. Otherwise, we could exclude the entries of $\b S$ (and $\b Z)$ that are $0$ almost surely. Also, since $\rho(\b D\b S, \b D\b Z)=
\rho(\b S, \b Z)$ for any diagonal matrix $\b D$, we may assume that all the elements in the diagonal of $\b\Sigma$ are $1$, i.e.,  $\b\Sigma$ is a correlation matrix. Finally,  we assume that $\b \Sigma$ is positive definite and denote by $\sigma_*^2\in(0,1)$ its smallest eigenvalue. 

\begin{theorem}\label{T:CLT_HD}[High-Dimension Central Limit Theorem for Strong Mixing sequences] 

\begin{enumerate}[(a)]
    \item \underline{Polynomial Case}: If for some $p\in[4,\infty)$ and $\epsilon>0$, $\viii{X_{i,t}}_{p+\epsilon}\leq B$ for $i\in[d]$ and $t\in [T]$ for some constant $B$ that might dependent on $T$ and $p$; and $\alpha_n\leq  K n^{-r}$ for $r\geq 0$ and a constant $K$ that might depend on $d$, then
\begin{align*}
    \rho\left(\b S,\b Z\right)&\lsim \frac{B^2}{\sigma_*^2}\left[\frac{\mathscr{A}_2}{\sqrt{T}} + \frac{\mathscr{A}_1}{T^{(1/2)-\kappa}}  +\frac{K^{1-2/p}}{T^{(r/2)(1-2/p) -1}}\right]\log T\log d\\
    &\qquad +\frac{(B\mathscr{A}_{\sqrt{T},4})^2(\log d)^{3/2}\log T}{T^{1/4}\sigma_*^2} + \frac{(B\mathscr{A}_{\sqrt{T},p})^2d^{2/p}(\log d)^2\log T}{T^{1/2-1/p}\sigma_*^2}\nonumber\\
    &\qquad +\frac{\left[ B\mathscr{A}_{\sqrt{T},p} d(\log d)^{(3/2)p-4}\log T \log (dT)  \right]}{T^{1/4}\sigma_*^{p/(p-2)}}^{1/(p-2)},\nonumber\\
    &\qquad+ \frac{1}{T^{r\kappa -1/2}}\left[d^{1/(p+1)}(\mathscr{A}_{T^{\kappa},p}B)^{D_p}\sqrt{1\lor \log d}+ K \right]\nonumber,
\end{align*}
where $\kappa:=\frac{1/2 + D_p/4}{r+D_p/2}\land 1/2$, $D_p := \frac{p}{p+1}$,  $\mathscr{A}_{k,p}:=\left[\sum_{n=0}^{k-1} (n+1)^{p/2-2}\alpha_n^{\epsilon/(p+\epsilon)}\right]^{\tfrac{2}{p}}$ for $k\in\N$, $\mathscr{A}_1:=\sum_{n=0}^{\lfloor T^\kappa \rfloor}\alpha_n^{1-2/p}$, and $\mathscr{A}_2:=\sum_{n=1}^{\lfloor \sqrt{T} \rfloor}n \alpha_n^{1-2/p}$. In particular, if $r>(\tfrac{p+\epsilon}{\epsilon})(p/2-1)\lor \tfrac{2}{1-2/p}$ then $\mathscr{A}_{k,p},\mathscr{A}_{1}$ and $\mathscr{A}_{2}$ can be bounded by constant that does not depend on $T$ or $d$.

\item \underline{Exponential Case}: If $\{X_{i,t}: t\in[T]\}$ fulfills the conditions of Lemma \ref{L:tail_bound_exp} for $i\in [d]$, then for $T\geq 16\lor (C\log d)^{1/\gamma -1/2}$, $d\geq 2$, and $d\sqrt{T}\geq \exp(1/\gamma) - 1 $, 
\begin{align*}
    \rho\left(\b S_{X},\b Z\right)&\lsim \frac{(\log T)^{\gamma_1 + 1} \log d + \big[\log (dT)\big]^{2/\gamma}(\log d)^2\log T}{\sqrt{T}\sigma_*^2} \\
    &\qquad + \frac{(\log d)^{2} +(\log d)^{3/2}\log T+ \log d (\log T)^{\gamma_1+1}\log (dT)}{T^{1/4}\sigma_*^2}.
\end{align*}
\end{enumerate}
\end{theorem}

\begin{proof}
We adapt the classical ``big block-small block'' technique commonly used to prove central limit theorems with dependence. Recall we assume $T\geq 4$ and consider two sequences of positive integers $a:=a_T$ and $b:=b_T$ such that $b\leq a$ and $a+b\leq T/2$. Let $m:=\lceil T/(a+b)\rceil$ and  define for $j\in[m-1]$ consecutive blocks of size $a$ and $b$
with index set $\m{A}_j:=\{((j-1)(a+b)+1,\dots, (j-1)(a+b)+a\}$ and $\m{B}_j:=\{(j-1)(a+b)+a+1,\dots j(a+b)\}$.
Finally, set $\m{A}_{m}:=\{m(a+b) +1, \dots, T\}$, which might be empty. Define,
\[
\b A_j := \sum_{t\in \m{A}_j}\b X_t,\quad j\in[m]\qquad\textnormal{and}\qquad\b B_j = \sum_{t\in\m{B}_j} \b X_t,\qquad j\in[m-1]
\]
such that
\[
\b S:= \b S_X:=\tfrac{1}{\sqrt{T}}\sum_{t=1}^T \b X_t =\tfrac{1}{\sqrt{T}}\sum_{j=1}^{m} \b A_j + \tfrac{1}{\sqrt{T}}\sum_{j=1}^{m-1} \b B_j=:\b S_A + \b S_B.
\]
Also, let $\{\widetilde{\b A}_j:j\in[m]\}$ be an independent sequence such that $\b A_j$ and $\widetilde{\b A}_j$ have the same distribution for $j\in[m]$, and define $\b S_{\widetilde{A}} := T^{-1/2}\sum_{j=1}^{m} \widetilde{\b A}_j$.

We start by applying Lemma \ref{L:useful_decompostion} with $X=\b S_X$, $Y=\b S_A$ and $d(x,y)=\|x-y\|_\infty$ to obtain, for $A\in\m{R}$ and $s>0$,
\begin{align*}
 |\P(\b S_X\in A)- \P(\b S_A\in A)|\leq \P(\|\b S_B\|_\infty>s) + \P(\b S_A\in A^s\setminus A)\lor \P(\b S_A\in A\setminus A^{-s}).
\end{align*}
The right hand side of the last expression can be upper bounded as
\begin{align*}
    \P(\b S_A\in A^s\setminus A) &= \P(\b Z\in A^s\setminus A) + \big[\P(\b S_A\in A^s\setminus A)-\P(\b Z\in A^s\setminus A)\big]\\
    &\leq \P(\b Z\in A^s\setminus A)+2\rho(\b S_A,\b Z)\\
    &\leq \P(\b Z\in A^s\setminus A)+2\rho(\b S_A,\b S_{\widetilde{A}}) + 2\rho(\b S_{\widetilde{A}},\b Z).
\end{align*}
We can proceed similarly to conclude that $\P(\b S_A\in A\setminus A^{-s})\leq \P(\b Z\in A\setminus A^{-s})+2\rho(\b S_A,\b S_{\widetilde{A}}) + 2\rho(\b S_{\widetilde{A}},\b Z)$. Since $\P(\b Z\in A^s\setminus A)\lor\P(\b Z\in A\setminus A^{-s})\leq 2\Delta(\b Z, s)$, we can take the supremum over $A\in\m{R}$ to write, for every $s>0$,
\begin{align*}
\rho(\b S_X,\b S_A)\leq \P(\|\b S_B\|_\infty>s)+2\Delta(\b Z, s) +2\rho(\b S_A,\b S_{\widetilde{A}}) + 2\rho(\b S_{\widetilde{A}},\b Z).
\end{align*}
By the triangle inequality,
\begin{align*}
\rho\left(\b S_X,\b Z\right)&\leq \rho\left(\b S_X,\b S_A\right)  + \rho\left(\b S_A,\b S_{\widetilde{A}}\right) +\rho\left(\b S_{\widetilde{A}},\b Z\right).
\end{align*}
Finally, combine the last two displays we obtain the inequality
\begin{align}
\rho\left(\b S_X,\b Z\right)&\leq 3\rho\left(\b S_{\widetilde{A}},\b Z\right)  +2\Delta(\b Z, s) + \P\left(\|\b S_B \|_\infty>s\right) + 3 \rho\left(\b S_A,\b S_{\widetilde{A}}\right).\label{E:GA_decompose2}
\end{align}
We now proceed to bound each of the terms on the right-hand side.

\textbf{Step 1.} Recall that $\{\widetilde{\b A}_j:j\in[m]\}$ is a zero-mean independent sequence, so we deal with the first term in \eqref{E:GA_decompose2} applying Theorem 2.2 in \cite{clt_hd} in the form
\[
\b S_{\widetilde{A}}:=\frac{1}{\sqrt{T}}\sum_{j=1}^m\widetilde{\b A}_j = \frac{1}{\sqrt{m}}\sum_{j=1}^m\b W_j\quad\textnormal{and}\quad \b W_j:=\sqrt{\frac{m}{T}}\widetilde{\b A}_j,
\]
which give us
 \begin{align}\label{eq:GA_decomposition}
    \rho\left(\b S_{\widetilde{A}},\b Z\right)&\lsim \log m\left[\frac{(\log d) M_\Sigma}{\sigma_*^2} +\frac{(\log d)^{3/2}\sqrt{\mu_4}}{m\sigma_*^2} + \frac{(\m{M}\log d)^2}{m\sigma_*^2}\right]\\
    &\qquad + \inf_{x>0}\left[\frac{\log d\sqrt{\log m \log (dm)H(x)}}{\sqrt{m}\sigma_*^2} + \frac{x(\log d)^{3/2}}{\sqrt{m}\sigma_*}\right],\nonumber
\end{align}
where $M_\Sigma:=\|\b\Sigma_{\b S_{X}} -\b\Sigma_{\b S_{\widetilde{A}}}\|_{\max}$, $\mu_4 := \max_{i\in[d]}\sum_{j=1}^m\viii{W_{j,i}}_4^4$, $\mathcal{M}:=\viii{\max_{j,i} |W_{j,i}|}_4$ and $H(x):=\max_{j\in[m]}\E\|\b W_j\|_\infty^4\1\{\|\b W_j\|_\infty>x\}$ for $x>0$.

From Lemma \ref{L:tail_bound_poly} we have, for $p\in[2,\infty)$ and $j\in[m]$, $\viii{\widetilde{ A}_{j,i}}_{p}=\viii{ A_{j,i}}_{p}=\viii{\sum_{t\in\m{A}_j}X_{t,i}}_{p}\lsim\sqrt{a+b}\mathscr{A}_{a,p} L_{p}\lsim\sqrt{a}\mathscr{A}_{a,p} L_{p}$ where $L_p:=\max_{i\in[d]}\max_{t\in [T]}\viii{X_{t,i}}_{p+\epsilon}$,  $W_{j,i}$ and $X_{t,i}$ denotes the $i$-th element of $\b W_j$ and $\b X_t$ respectively,  for $i\in[d]$. Thus, for $i\in[d]$, $j\in[m]$ and $p\in[2,\infty)$.
\[\viii{ W_{j,i}}_p\lesssim L_p\mathscr{A}_{a,p}\sqrt{\frac{ma}{T}}\leq L_p\mathscr{A}_{a,p}=:R_p.\]

Expression (1,12a) and (1.12b) in \cite{rio2017asymptotic} give us, for $t,s\in[T]$, $i,k\in[d]$ and $p\in[1,\infty]$,
\[
\left|\E(X_{t,i}X_{s,k})\right| \leq 2\alpha_{|t-s|}^{1-2/(p+\epsilon)}\|X_{t,i}\|_{p+\epsilon}\|X_{t,k}\|_{p+\epsilon}\leq 2L_p^2\alpha_{|t-s|}^{1-2/p}.
\]

Recall that $\b S_{\widetilde{A}}^G\sim N(\b 0,\b\Sigma_{\b S_{\widetilde{A}}})$ and $\b Z\sim N(\b 0,\b\Sigma_{\b S_{X}})$, where
\[
\b\Sigma_{\b S_{\widetilde{A}}}=\frac{1}{T}\sum_{j=1}^m \sum_{t\in\mathcal{A}_j} \sum_{s\in\mathcal{A}_j} \E(\b X_t \b X_s')\quad\textnormal{and}\quad \b\Sigma_{\b S_{X}}=\frac{1}{T} \sum_{t=1}^T \sum_{s=1}^T \E( \b X_t \b X_s').
\]
Define $\b\Delta_\ell:= \sum_{i=1}^{a-1}\sum_{j=1}^{a-1} \E (X_{i+\ell} X'_{a+j+\ell})$ for $\ell\geq 1$, then using the bound above we have
\[\|\b\Delta_\ell\|_{\max}\leq 2L_p^2\sum_{m=1}^{a-1} m \alpha_m^{1-2/p},\quad  \]
Then,
\begin{align*}
    M_\Sigma & \leq \tfrac{4 m}{T} \max_\ell\|\b\Delta_\ell\|_{\max} + \frac{m}{T}\max_{j\in[m]}\|\E (\b B_j\b B_j')\|_{\max} + \frac{1}{T}\sum_{|t-s|>a} \E (\b X_t \b X_s') \\
    & \lsim \frac{L_p^2}{T}\left[m\mathscr{A}_2 + mb\mathscr{A}_1(b)  + T^2\alpha_a^{1-2/p}\right],
\end{align*}
where $\mathscr{A}_1(m):=\sum_{n=0}^{m-1}\alpha_n^{1-2/p}$ for $m\geq 1$ and $\mathscr{A}_2:=\sum_{n=1}^{a-1}n \alpha_n^{1-2/p}$.

Also,
\begin{align*}
\mu_4 &\leq m\max_{i\in[d]}\max_{j\in[m]}\viii{W_{j,i}}_4^4\lsim m R_4^4\\
\mathcal{M}&\leq \viii{\max_{j,i} |W_{j,i}|}_p\leq  \left(md\max_{j,i} |W_{j,i}|^p\right)^{1/p}\lsim R_p(md)^{1/p}\\
    H(x) &\leq \max_{j\in[m]}\E\|\b W_j\|_\infty^p/x^{p-4}\lsim R_p^p d/x^{p-4}
\end{align*}
Plug these bounds back into \ref{eq:GA_decomposition} and set $x=\left[\frac{\log m \log (dm) R_p d}{\sigma_*^2\log d}\right]^{1/(p-2)}$ to equate the terms inside the infimum. As a consequence, we are left with 
\begin{align}\label{eq:GA_decomposition_poly}
    \rho\left(\b S_{\widetilde{A}},\b Z\right)&\lsim \frac{L_p^2\left[m\mathscr{A}_2 + mb\mathscr{A}_1(b)  + T^2\alpha_a^{1-2/p}\right]\log m \log d}{T\sigma_*^2}\\
    &\qquad +\frac{R_4^2(\log d)^{3/2}\log m}{\sqrt{m}\sigma_*^2} + \frac{R_p^2(md)^{2/p}(\log d)^2\log m}{m\sigma_*^2}\nonumber\\
    &\qquad +\frac{\left[ R_p d(\log d)^{(3/2)p-4}\log m \log (dm)  \right]}{\sqrt{m}\sigma_*^{p/(p-2)}}^{1/(p-2)}.\nonumber
\end{align}

\textbf{Step 2.} By Nazarov's inequality for Gaussian random vectors (Theorem 1 in \cite{nazarov}) we can bound the second term in \eqref{E:GA_decompose2} as
\[
\Delta\left(\b Z,s\right)\lsim s\sqrt{ 1\lor \log d}.
\]
Similarly to Step 1, Lemma \ref{L:tail_bound_poly} give us for $j\in[m]$, $i\in[d]$ and $p\in[2,\infty)$
\[\viii{\widetilde{ B}_{j,i}}_{p}=\viii{ B_{j,i}}_{p}=\viii{\sum_{t\in\m{B}_j}X_{t,k}}_{p}\lsim\sqrt{b}\mathscr{A}_{b,p} L_p,\]
where $B_{j,k}$ and $N_{t,k}$ denotes the $i$-th element of $\b B_j$ and $\b X_t$ respectively,  for $i\in[d]$; and $\mathscr{A}_{b}:=\left[\sum_{0\leq m <b} (m+1)^{p/2-2}\alpha_m^{\epsilon/(p+\epsilon)}\right]^{\tfrac{2}{p}}$. We can then bound the third term in \eqref{E:GA_decompose2} by Markov's inequality as 
\[\P\left(\|\b S_B\|_\infty>s\right) \leq d
\left(\frac{\sqrt{mb/T}\mathscr{A}_{b,p} L_{p}}{s}\right)^{p}.\]

We can set $s=s^*:=\big[d^{1/p}\sqrt{mb/T}\mathscr{A}_{b,p} L_{p}\big]^{p/(p+1)}$ to equate (up to log terms) the two terms above containing $s$ and obtain
\begin{equation*}
   \Delta\left(\b Z,s\right) + \P\left(\|\b S_B\|_\infty>s\right) \lsim \left(d^{1/p}\sqrt{\tfrac{mb}{T}}\mathscr{A}_{b,p} L_{p}\right)^{\tfrac{p}{p+1}}\sqrt{1\lor \log d}.
\end{equation*}

\textbf{Step 3.} Notice that any measurable $A\subseteq\R^2$ we have $|\P[(\b A_1, \b A_2)\in A] - \P[\widetilde{\b A}_1,\widetilde{\b A}_2\in A]|\leq \alpha_b$ where  $\{\alpha_n,n\in \N\}$. Then, the last term in \eqref{E:GA_decompose2} can be upper bounded by $(m-1)\alpha_b$ by induction, i.e,
\begin{equation*}
  \rho\left(\b S_A,\b S_{\widetilde{A}}\right) \leq (m-1)\alpha_b.
\end{equation*}

\textbf{Step 4.}  Applying
the bound derived in the steps above back into \eqref{E:GA_decompose2} we are left with

\begin{align}\label{E:GA_decompose3}
    \rho\left(\b S_{X},\b Z\right)&\lsim \frac{L_p^2\left[m\mathscr{A}_2 + mb\mathscr{A}_1(b)  + T^2\alpha_a^{1-2/p}\right]\log m \log d}{T\sigma_*^2}\\
    &\qquad +\frac{R_4^2(\log d)^{3/2}\log m}{\sqrt{m}\sigma_*^2} + \frac{R_p^2(md)^{2/p}(\log d)^2\log m}{m\sigma_*^2}\nonumber\\
    &\qquad +\frac{\left[ R_p d(\log d)^{(3/2)p-4}\log m \log (dm)  \right]}{\sqrt{m}\sigma_*^{p/(p-2)}}^{1/(p-2)},\nonumber\\
    &\qquad+ \left(d^{1/p}\sqrt{\tfrac{mb}{T}}\mathscr{A}_{b,p} L_{p}\right)^{\tfrac{p}{p+1}}\sqrt{1\lor \log d}  + m\alpha_b\nonumber\\
  &=:(I) + (II) + (III) + (IV) + (V) + (VI)\nonumber.
\end{align}

We now must choose sequences $a$ and $b$ to balance all the terms appearing on the right-hand side. Let $a =1\lor \lfloor T^{\gamma_a}\rfloor$ and $b =1\lor \lfloor T^{\gamma_b}\rfloor$ for $0\leq \gamma_b\leq \gamma_a <1$ . Then $m\asymp T^{1-\gamma_a}$, $ma/T\asymp 1$, $mb/T\asymp T^{\gamma_b-\gamma_a}$. Comparing $(I)$, $(II)$ and $(IV)$, it seems we cannot do much better than setting $\gamma_a=\gamma_a^*:=1/2$.
Since $\alpha_n \leq K n^{-r}$ then, $m\alpha_b\lesssim K T^{1-\gamma_{a}^*-r\gamma_b} = KT^{1/2-r\gamma_b}$. To equate $(V)$ with $(VI)$ (up to $\log$ terms) and respect the fact that $0\leq \gamma_b\leq \gamma_a^*$ we set
\[
\gamma_b = \gamma_b^*:= \left(\frac{1/2+D_p/4}{r+D_p/2}\right)\land 1/2;\qquad D_p:=\frac{p}{p+1}.
\]
Then,
\begin{equation*}
   (V) + (VI) \lsim \frac{1}{T^{\phi(p,r)}}\left[d^{1/(p+1)}(\mathscr{A}_{b,p} L_{p})^{D_p}\sqrt{1\lor \log d}+ K\right],
\end{equation*}
where
\[\phi(r,p) := \gamma_a^* + r\gamma_b^* - 1=\begin{cases}
\tfrac{D_p(r-1)}{4(r+D_p/2)} &; r\geq 1\\
\frac{r-1}{2} &;  0\leq r< 1.
\end{cases}  \]
Therefore,
\begin{align}\label{E:GA_decompose4}
    \rho\left(\b S_{X},\b Z\right)&\lsim \frac{L_p^2}{\sigma_*^2}\left[\frac{\mathscr{A}_2}{\sqrt{T}} + \frac{\mathscr{A}_1(T^{\gamma_b^*})}{T^{(1/2)-\gamma_b^*}}  +\frac{K_d^{1-2/p}}{T^{(r/2)(1-2/p) -1}}\right]\log T\log d\\
    &\qquad +\frac{(L_4\mathscr{A}_{\sqrt{T},4})^2(\log d)^{3/2}\log T}{T^{1/4}\sigma_*^2} + \frac{(L_p\mathscr{A}_{\sqrt{T},p})^2d^{2/p}(\log d)^2\log T}{T^{1/2-1/p}\sigma_*^2}\nonumber\\
    &\qquad +\frac{\left[ L_p\mathscr{A}_{\sqrt{T},p} d(\log d)^{(3/2)p-4}\log T \log (dT)  \right]}{T^{1/4}\sigma_*^{p/(p-2)}}^{1/(p-2)},\nonumber\\
    &\qquad+ \frac{1}{T^{\phi(p,r)}}\left[d^{1/(p+1)}(\mathscr{A}_{T^{\gamma_b^*},p} L_{p})^{D_p}\sqrt{1\lor \log d}+ K \right]\nonumber,
\end{align}
which concludes the proof of part (a).

The proof for part (b) runs parallel to the proof of part (a). In step 1, we replace Lemma \ref{L:tail_bound_poly} by Lemma \ref{L:tail_bound_exp} to conclude that $\viii{A_{i,j}}_{e^\gamma}\lsim \sqrt{a}$ and thus $\viii{W_{i,j}}_{e^\gamma}\lsim 1$, which in turn allow us to obtain
\begin{align*}
   M_\Sigma & \lsim \frac{m}{T}+ \frac{mb}{T}  + T\exp\left[-K_1(1-2/p)a^{\gamma_1}\right];\qquad p\in[2,\infty)\\
\mu_4 &\leq m\max_{i\in[d]}\max_{j\in[m]}\viii{W_{j,i}}_4^4\lsim m.
\end{align*}
We apply Lemma \ref{L:Orlicz_basic} (b) followed (d) to write
\begin{align*}
\mathcal{M}&\lsim \viii{\max_{j,i} |W_{j,i}|}_{e^\gamma}\lsim  \psi_{e^\gamma}^{-1}(md)\max_{j,i} \viii{W_{j,i}}_{e^\gamma}\lsim \psi_{e^\gamma}^{-1}(md).
\end{align*}
Similarly, we obtain $\E\|\b W_j\|_\infty^p =  \viii{\max_{i\in[d]}|W_{j,i}|}_p^p\lsim \viii{\max_{i\in[d]}|W_{j,i}|}_{e^{\gamma}}^p\lsim \big[\psi_{e^\gamma}^{-1}(d)\big]^p$. Set $x=K_x\sqrt{\log d}$ for some $K_x>0$ large enough such that, by Lemma \ref{L:tail_bound_sup_exp}, we have for $d\geq 2$ and $a\geq 4\lor K_x(\log d)^{2/\gamma-1}$,
\[
\P(\|\b W_j\|_\infty>x)=\P\left(\left\|\sum_{t\in\m{A}_j} \b  X_t\right\|_\infty>K_x\sqrt{\frac{T}{ma}}\sqrt{a\log d}\right)\leq (da)^{1-K_x^\gamma} + 2d^{1-K^2_x}.
\]
Then, by Cauchy-Schwartz inequality, 
\[
H(x)\leq \max_{j\in[m]}\sqrt{\E\|\b W_j\|_\infty^8\P(\|\b W_j\|_\infty>x)}\lsim \left[\psi_{e^\gamma}^{-1}(d)\right]^4\left[(da)^{1-K_x^\gamma} + 2d^{1-K_x^2}\right]^{1/2}
\]

Plug these bounds back into \eqref{eq:GA_decomposition}, we are left with 
\begin{align}\label{eq:GA_decomposition_exp}
    \rho\left(\b S_{\widetilde{A}},\b Z\right)&\lsim \left\{\frac{m}{T} + \frac{mb}{T} + T\exp\left[-K_1(1-2/p)a^{\gamma_1}\right]\right\}\frac{\log m \log d}{\sigma_*^2}\\
    &\qquad +\frac{(\log d)^{3/2}\log m}{\sqrt{m}\sigma_*^2} + \frac{\big[\psi_{e^\gamma}^{-1}(md)\big]^{2}(\log d)^2\log m}{m\sigma_*^2}\nonumber\\
    &\qquad +\frac{\log d\sqrt{\log m \log (dm)}\big[\psi_{e^\gamma}^{-1}(d)\big]^2\big[(da)^{1-K^\gamma} + 2d^{1-K^2}\big]^{1/4}}{\sqrt{m}\sigma_*^2}\nonumber\\
    &\qquad + \frac{(\log d)^{2}}{\sqrt{m}\sigma_*}.\nonumber
\end{align}

For step 2, we can then bound the third term in \eqref{E:GA_decompose2} by Lemma \ref{L:tail_bound_sup_exp}. Set $s=K_s\sqrt{\frac{mb}{T} \log d \log T}$ for some $K_s>0$ large enough so that, for $d\geq 2$, and $mb\geq 4\lor K_s(\log d)^{2/\gamma-1} $, we have
\begin{align*}
   \P\left(\|\b S_B\|_\infty>s\right)&=\P\left(\left\|\sum_{j=1}^{m-1}\sum_{t\in\m{B}_j} \b  X_t\right\|_\infty>K_s\sqrt{mb\log T\log d}\right)\\
   &\leq (dmb)^{1-K_s^\gamma(\log T)^{\gamma/2}} + 2d^{1-K_s^2\log T}.
\end{align*}
By plugging the bounds above back into \eqref{E:GA_decompose2} we obtain
\begin{align}\label{eq:decomposition_exp}
    \rho\left(\b S_{X},\b Z\right)&\lsim \left\{\frac{m}{T} + \frac{mb}{T} + T\exp\left[-K_1(1-2/p)a^{\gamma_1}\right]\right\}\frac{\log m \log d}{\sigma_*^2}\\
    &\qquad +\frac{(\log d)^{3/2}\log m}{\sqrt{m}\sigma_*^2} + \frac{\big[\psi_{e^\gamma}^{-1}(md)\big]^{2}(\log d)^2\log m}{m\sigma_*^2}\nonumber\\
    &\qquad +\frac{\log d\sqrt{\log m \log (dm)}\big[\psi_{e^\gamma}^{-1}(d)\big]^2\big[(da)^{1-K^\gamma_x} + 2d^{1-K_x^2}\big]^{1/4}}{\sqrt{m}\sigma_*^2}\nonumber\\
    &\qquad + \frac{(\log d)^{2}}{\sqrt{m}\sigma_*} + (dmb)^{1-K_s^\gamma(\log T)^{\gamma/2}} + 2d^{1-K_s^2\log T}.\nonumber\\
    &\qquad + \log d\sqrt{\tfrac{mb}{T} \log T}  + m\exp(-K_1b^{\gamma_1})\nonumber.
\end{align}
For the same reason discussed in the polynomial case, we cannot do much better than setting $a \asymp T^{1/2}$. Then,  $m\asymp T^{1/2}$. Set $b\asymp (\log T)^{\gamma_1}$ so that the last term is at most $T^{-1/4}$. Note that the numerator of the fourth term above can be bounded uniformly in $d$ and $T$ by a constant only depending on $\gamma$ and the choice of $K_x>1$. Also, $\left(\frac{1}{d}\right)^{K^2_s\log T-1}\lsim T^{-1}$ provided that $K_s>1$ as $d\geq 2$, and $\left[\frac{1}{d\sqrt{T}(\log T)^{\gamma_1}}\right]^{K^\gamma(\log T)^{\gamma/2}-1}\lsim T^{-1/2}$ by taking $K_s\geq 2^{1/\gamma}$. Finally, $\psi_{e^\gamma}^{-1}(x) \lsim  (\log x )^{1/\gamma}$, for $x\geq \exp(1/\gamma)-1$ by Lemma \ref{L:Orlicz_basic}(a). Result (b) follows.

\begin{align}\label{eq:decomposition_exp_2}
    \rho\left(\b S_{X},\b Z\right)&\lsim \frac{(\log T)^{\gamma_1 + 1} \log d}{\sqrt{T}\sigma_*^2}\\
    &\qquad +\frac{(\log d)^{3/2}\log T}{T^{1/4}\sigma_*^2} + \frac{\big[\psi_{e^\gamma}^{-1}(d\sqrt{T})\big]^{2}(\log d)^2\log T}{\sqrt{T}\sigma_*^2}\nonumber\\
    &\qquad +\frac{\log d\sqrt{\log T \log (d\sqrt{T})}\big[\psi_{e^\gamma}^{-1}(d)\big]^2\big[(d\sqrt{T})^{1-K^\gamma} + 2d^{1-K^2}\big]^{1/4}}{T^{1/4}\sigma_*^2}\nonumber\\
    &\qquad + \frac{(\log d)^{2}}{T^{1/4}\sigma_*} + \left(\frac{1}{d\sqrt{T}(\log T)^{\gamma_1}}\right)^{K^\gamma(\log T)^{\gamma/2}-1} + \left(\frac{1}{d}\right)^{K^2\log T-1}.\nonumber\\
    &\qquad + \frac{\log d (\log T)^{\gamma_1+1}}{T^{1/4}} \log\left(1\lor d\right) +\frac{1}{T^{1/4}}\nonumber.
\end{align}
Note that the numerator of the fourth term above can be bounded uniformly in $d$ and $T$ by a constant only depending on $\gamma$ and the choice of $K>1$. Also, $\left(\frac{1}{d}\right)^{K^2\log T-1}\lsim T^{-1}$ as $d\geq 2$ and $K>1$, and $\left[\frac{1}{d\sqrt{T}(\log T)^{\gamma_1}}\right]^{K^\gamma(\log T)^{\gamma/2}-1}\lsim T^{-1/2}$ by taking $K\geq 2^{1/\gamma}$. Finally, $\psi_{e^\gamma}^{-1}(x) \lsim  (\log x )^{1/\gamma}$ for $x\geq \exp(1/\gamma)-1$ by Lemma \ref{L:Orlicz_basic}(a). Result (b) follows.
\end{proof}

\section{Auxiliary Lemmas}\label{S:lemmas}

\subsection{Concentration Inequalities for strong mixing sequences}

\begin{lemma}\label{L:mixing_lag} Let $\{\b X_t: t\in\Z\}$ be a sequence of $d$-dimensional random vector with strong mixing coefficient $\{\alpha^{\b X}_n:n\in\N\}$.  For a non-negative integer $k$, define $\{\b Y_t:=f(\b X_t,\dots,  \b X_{t-k}): t\in\Z\}$ for some measurable $f:\R^{d(k+1)}\to \R^q$ and denote its strong mixing coefficient by $\{\alpha^{\b Y}_n:n\in\N\}$.  Then $\alpha^{\b Y}_n\leq \alpha^{\b X}_{(n-k)\lor 0}$ for $n\in\N$.
\end{lemma}
\begin{proof}
Let $\m{X}_s^t$ and $\m{Y}_s^t$ be the $\sigma$-algebra generated by $(X_s,\dots, X_t)$ and  $(Y_s,\dots, Y_t)$, respectively where $-\infty\leq s\leq t\leq \infty$. Since $\m{Y}_{-\infty}^t\subseteq \m{X}_{-\infty}^t$ and $\m{Y}_{t+n}^\infty\subseteq \m{X}_{t+n-k}^\infty$, we have $\alpha^{\b Y}_n\leq \alpha^{\b X}_{n-k}$ for $n\geq k$ and $\alpha^{\b Y}_n\leq \alpha^{\b X}_0$ for $0\leq n <k$.
\end{proof}

\begin{lemma}\label{L:tail_bound_poly}[based on Corollary 1.1 and Theorem 6.3 in \cite{rio2017asymptotic}] Let $S_T = \sum_{t=1}^T X_t$ where $\{X_t:t\in [T]\}$ is a sequence of zero mean real-valued random variables  such that $\viii{X_t}_p<\infty$ for $t\in[T]$, and  the mixing coefficients given by $\{\alpha_m:0\leq m <T\}$. Then, for $p\in[2,\infty)$ and $T\in\N$, 
\[\viii{S_T}_p\leq a_p\sqrt{T}L_{2,\alpha} +  b_p T^{1/p}L_{p,\alpha},\]
where $L_{s,\alpha}:=\left\{\int_0^1[\alpha^{-1}(u)\land T]^{s-1} Q^{s}(u){ du}\right\}^{1/s}$ for $s>0$,  $a_p$ and $b_p$ are positive constants only depending on $p$, $\alpha^{-1}(u):=\sum_{0\leq m<T}\1\{u\leq \alpha_m\}$, $Q:=\max_{t\in[T]} Q_t$ and $Q_t(u):=\sup\{x\in\R:\P(|X_t|>x)<u\}$.

If further,  $X_t\in\m{L}^q$ for $t\in[T]$ for some $q\in(p, \infty]$, then
\[\viii{S_T}_p\leq c_{p,q} \mathscr{A}_{p,q}\max_{t\in[T]}\viii{X_t}_q\sqrt{T},\]
where $\mathscr{A}_{p,q}:=\left[\sum_{0\leq m <T} (m+1)^{p-2}\alpha_m^{1-p/q}\right]^{\tfrac{1}{p}}$ and $c_{p,q}$ is a positive constant only depending on $p$ and $q$.   

If further, $\alpha_{m}\leq (m+1)^{-r}$ for some constant $r\geq 0$  and all  $m\in\N$ then

\[\viii{S_T}_p\leq d_{p,q,r}\mathscr{R}_{p,q,r}(T)\max_{t\in[T]}\viii{X_t}_q;\qquad \mathscr{R}_{p,q,r}(T):=
\begin{cases}
   \sqrt{T} &; r>\nu\\
   (\log T+  \lambda)^{\tfrac{1}{p}}\sqrt{T} &; r=\nu\\
   T^{[\frac{1}{2}+\frac{p-1 -r(1-p/q)}{p}]\land 1} &; r<\nu,
   \end{cases}
\]
where $\nu:=\frac{(p-1)}{1-p/q}$ and $d_{p,q,r}$ is a positive constant only depending on $p$, $q$ and $r$ and $\lambda\approx 1.1$.
\end{lemma}
\begin{proof}
The first result, for $p=2$, follows from Corollary 1.1 in \cite{rio2017asymptotic} by setting $a_2=b_2 = 1$ and bounding
\[
\sum_{t=1}^T\int_0^1[\alpha^{-1}(u)\land T] Q^2_t(u){ du}\leq T\int_0^1[\alpha^{-1}(u)\land T] \left(\max_{t\in[T]}Q_t\right)^2(u){ du}.
\]
For $p>2$ we rely on Theorem 6.3 in \cite{rio2017asymptotic}, the concavity of $x\mapsto x^{1/p}$ and the result for $p=2$ to write $\viii{S_T}_p\leq a_p\sqrt{T}L_{2,\alpha} +  b_p T^{1/p}L_{p,\alpha}$ with $a_p = 2^{1+2(p+1)/p}p^{1/p}\sqrt{p+1}$ and $b_p=\left[\frac{p}{p-1}4^{p+1}(p+1)^{p-1}\right]^{1/p}$.

For the second result, Markov's inequality give us $\P(|X_t|\geq x)\leq\left(\viii{X_t}_q/x\right)^q$ for $x>0$. Then, $Q_t(u)\leq  \frac{\viii{X_t}_q}{u^{1/p}}$ and $Q(u):=\max_{t\in[t]}Q(u)\leq  \frac{\mu_q}{u^{1/p}}$, where $\mu_q:=\max_{t\in[T]}\viii{X_t}_q$. Then, for $s\in[2,p]$, we have
\begin{align*}
    L_{s,q}\leq \mu_q\left\{\int_0^1[\alpha^{-1}(u)\land T]^{s-1} u^{-s/q}\d u\right\}^{1/s}\leq \mu_q\left\{\int_0^1[\alpha^{-1}(u)\land T]^{p-1} u^{-p/q}\d u\right\}^{1/s}.
\end{align*}
Combine equation (C.10) in in \cite{rio2017asymptotic}, the bound $Q(u)\leq \frac{\mu_q}{u^{1/p}}$ and the last expression to obtain
\begin{align*}
    L_{s,q}\leq \mu_q\left(\frac{p-1}{1-p/q}\right)^{1/s}\left[\sum_{m=0}^{T-1} (m+1)^{p-2} \alpha_m^{1-s/q}\right]^{1/s},\qquad s\in[2,p].
\end{align*}
We apply this bound to the first result, such that the second result follows with $c_{p,q}=2(a_p\lor b_p)\left(\frac{p-1}{1-p/q}\right)^{1/s}$.

For the last result,  if $\alpha_{m}\leq (m+1)^{-r}$ then $\mathscr{A}_{p,q}^p\leq \sum_{m=1}^T m^{p-2-r(1-p/q)}$. For $r>\nu:=\frac{p-1}{1-p/q}$, the sum is convergent as $T\to\infty$, let $C_1:=\lim _{T\to\infty}\mathscr{A}_{p,q}(T)$. For $r=\nu$ we have the harmonic series which is known to diverge at rate slower than $\log T + \lambda$ where $\lambda\approx 1.1$. Finally,  for $r<\nu$ we have $T^{-1}\sum_{m=1}^T (m/T)^{p-2-r(1-p/q)}\leq \int_0^1 x^{p-2-r(1-p/q)}\d x =\frac{1}{p-1-r(1-p/q)}=:C_2^p$. Thus,
\[
\mathscr{A}_{p,q}\leq 
\begin{cases}
   C_1 &; r>\nu\\
   (\log T+  \lambda)^{1/p} &; r=\nu;\\
   C_2 T^{\frac{p-1 -r(1-p/q)}{p}} &; r<\nu.
   \end{cases}
\]
Apply this last bound to the second result yields the last result with $d_{p,q,r}=c_{p,q}(1\lor C_1\lor C_2)$.
\end{proof}

\begin{lemma}\label{L:tail_bound_exp}[based on Theorem 1 in \cite{MPR2011} ] Let $S_T = \sum_{t=1}^T X_t$ where $\{X_t:t\in [T]\}$ is a sequence of zero mean real-valued random variables  such that
\begin{enumerate}[(a)]
\item There exist two positive constants $\gamma_1$  and $K_1$
such that the strong mixing coefficients of the sequence satisfy $\alpha(m)\leq \exp(-K_1m^{\gamma_1})$ for any $1\leq m<T$ and $T\geq 2$, 
\item There exist two positive constants $\gamma_2$  and $K_2$
such that $\sup_{1\leq t\leq T,T\in \N}\viii{X_t}_{e^{\gamma_2}}\leq K_2$,
\item $\gamma < 1$ where $\gamma$ is defined by $1/\gamma=1/\gamma_1+1/\gamma_2$.
\end{enumerate}
Then, there exist positive constants $C_1$, $C_2$,  $C_3$ and $C_4$ depending only on $ K_2, K_1, \gamma$ and $\gamma_1$ such that, for $x>0$ and $T\geq 4$,
\[\P(|S_T|\geq x)\leq T\exp\left(-\frac{x^\gamma}{C_1}\right) + \exp\left[-\frac{x^2}{C_2(1+TV)}\right] + \exp\left\{-\frac{x^2}{C_3T}\exp\left[\frac{x^{\gamma(1-\gamma)}}{C_4(\log x)^\gamma}\right]\right\},\]
where $V$ is a finite constant.

In particular, there a constant $C_5$ only depending on $C_3, C_4$ and $\gamma$ such that, for $x>1$,  
\begin{align*}
\P(|S_T|\geq x)\leq T\exp\left(-\frac{x^\gamma}{C_1}\right) + \exp\left[-\frac{x^2}{C_2(1+TV)}\right]+\exp\left(-\frac{x^2}{C_5 T}\right).
\end{align*}
Furthermore, for some constant $C_6$ only depending on $C_i$ for $i\in[4]$, $\gamma_1$, $\gamma_2$ and $V$, 
\[\viii{S_T}_{\psi_\gamma} \leq C_6\sqrt{T},\quad\textnormal{and}\quad \psi_\gamma(x):=K_\gamma x\1_{0\leq x< a_\gamma} + (\exp x^\gamma -1)\1_{x\geq a_\gamma}\]
where $a_\gamma:=(1/\gamma)^{1/\gamma}$ and  $K_\gamma:=(\exp a_\gamma^\gamma -1)/a_\gamma$.
\end{lemma}
\begin{proof}
We verify the conditions (2.6), (2.7), and (2.8) of Theorem 1 in \cite{MPR2011}, henceforth MPR.  Assumption (b) ensures that $\E\left[\exp(|X_i /K_2|^{\gamma_2} )\right]\leq 2$ then following Remark 2 in MPR. Condition (2.7) is satisfied with $b = K_2$. From expression (2.5) in MPR we have $\tau(m)\leq 2\int_0^{2\alpha(m)} Q(u) du$, where $Q:=\sup_{1\leq t\leq T, T\in\N}  Q_{|X_t|}$ with $Q_{|X_t|}$ denoting the quantile function of $|X_t|$.  Cauchy-Schwartz inequality gives us 
$\tau(m)\leq 2 \sqrt{2\alpha(m)\int_0^1 Q^2(u) du}$.  Given Assumption (b), the integral is finite. Let denote it by $M^2$.  By Assumption (a) we have that $\tau(m)\leq 2 \sqrt{2} M\exp\left(-\frac{1}{2} K_1m^{\gamma_1}\right)$, then condition (2.6) is satisfied with $a=2 \sqrt{2} M$ and $c=\frac{1}{2} K_1$. The first result follows.

For the second result, note that the function $x\mapsto x^{(1-\gamma)}/\log x$ is continuously differentiable and  coercive for $x>1$,  hence it attains its minimum over $x\in(1,\infty)$ which is given by $C_\gamma:=e(1-\gamma)>0$.  Define $C_5:=C_3/\exp\left(\frac{C_\gamma^\gamma}{C_4}\right)$, then for $x>1$ the last term of the first result can be upper bounded by $\exp\left(-\frac{x^2}{C_5 T}\right)$,  and the second result follows.

For the last result we have, by Fubini's Theorem and $C>0$,
\begin{align*}
  \E\left[\psi_\gamma\left(\frac{S_T}{\sqrt{T}C}\right)\right]
  &= \int_0^\infty\P\left[|S_T|> C\psi^{-1}(x)\sqrt{T}\right]d x\\
  &=\int_0^{a_\gamma}\P\left[|S_T|> \frac{C}{K_\gamma} x\sqrt{T}\right]d x + \int_{a_\gamma}^\infty \P\left\{|S_T|> C \left[\log (1+x)\right]^{1/\gamma} \sqrt{T}\right\}d x\\
  &=: I_1 + I_2.
\end{align*}
Lemma \ref{L:tail_bound_poly} give us $\viii{T^{-1/2}S_T}_p\leq K_p$ for any $p\in[2,\infty)$. Then, for the first integral we have
\[
I_1\leq \int_0^{\infty}\P\left[|T^{-1/2}S_T|> \frac{C}{K_\gamma} x\right]d x= \frac{K_\gamma}{C}\viii{T^{-1/2}S_T}_1\leq \frac{K_\gamma}{C}\viii{T^{-1/2}S_T}_2\leq \frac{K_\gamma K_2}{C}.
\]
For the second one, note that $a_\gamma >1$, then
\begin{align*}
    I_2 &= \int_{a_\gamma}^\infty T\left(1+x\right)^{-\tfrac{C^\gamma T^{\gamma/2}}{C_1}} dx  +\int_{a_\gamma}^\infty  \exp\left(-\frac{C^2[\log (1+x)]^{2/\gamma}}{C_2(T^{-1}+V)}\right)+\exp\left(-\frac{x^2}{C_5 T}\right) d x\\
    &\leq \int_{a_\gamma}^\infty T\left(1+x\right)^{-\tfrac{C^\gamma T^{\gamma/2}}{C_1}} dx  +\int_{a_\gamma}^\infty  \exp\left(-\frac{C^2[\log (1+x)]^{2/\gamma}}{C_7}\right) dx:=I_3 + I_4,
\end{align*}
where $C_7:=C_2(1+V)\lor C_5$. For $C^\gamma/C_1> 1$, we can bound $I_3$ as 
\[
I_3 = \frac{T(1+a_\gamma)^{1-\tfrac{C^\gamma T^{\gamma/2}}{C_1}}}{\tfrac{C^\gamma T^{\gamma/2}}{C_1}-1}\leq \frac{T^{1-\gamma/2}(1+a_\gamma)^{1-T^{\gamma/2}}}{\tfrac{C^\gamma }{C_1}-1}\leq \frac{M_\gamma}{\tfrac{C^\gamma }{C_1}-1},
\]
where $M_\gamma:=\sup_{T\geq 4}\big[T^{1-\gamma/2}(1+a_\gamma)^{1-T^{\gamma/2}}\big]$  which is finite because $a_\gamma>0$. For the $C^2/C_7>1$, $\exp\left(-\frac{C^2[\log (1+x)]^{2/\gamma}}{C_7}\right)\leq \exp\left(-[\log (1+x)]^{2/\gamma}\right)$ and
\begin{align*}
    \int_{a_\gamma}^\infty\exp\left(-[\log (1+x)]^{2/\gamma}\right) d x &\leq  ((a_\gamma\lor 2) - a_\gamma) + \int_{a_\gamma\lor 2}^\infty\exp\left(-[\log (1+x)]^{2/\gamma}\right) d x\\
    &\leq 2+ \int_{0}^\infty\exp\left(-[\log (1+x)]^{2}\right) d x\leq 4,
\end{align*}
where we use $\log 3>1$ and the last integral equals $e^{1/4}\sqrt{\pi}(\mathsf{erf}(1/2) +1)/2\leq 2$ with $\mathsf{erf}$ is Gauss error function. Then, by the dominated convergence theorem, we conclude that $\lim_{C\to\infty} I_4 = 0$, hence there exist $C_8>0$ such $I_4\leq 1/3$ for $C\geq C_8$.

Therefore, $\E\psi_\gamma(\tfrac{S_T}{\sqrt{T}C})\leq I_1 + I_3+ I_4\leq 1$ for $C \geq  C_6:=3K_\gamma K_2\lor (2C_1)^{1/\gamma} \lor (2C_7)^{1/2}\lor  \big[C_1(1+3M)\big]^{1/\gamma}\lor C_8$.
\end{proof}

\begin{lemma}\label{L:tail_bound_sup_poly} Let $\b S_T = \sum_{t=1}^T \b X_t$ where $\{\b X_t:=(X_{1,t},\dots, X_{n,t})': t\in[T]\}$ is a sequence of zero mean $n$-dimensional random vector such that $\viii{X_{i,t}}_{p+\epsilon}<\infty $ for $i\in[n]$ and $t\in[T]$  for $p\in[2,\infty)$ and $\epsilon>0$. Then,  for $T\in \N$ and $x>0$, 
\[\P(\|\b S_T\|_\infty\geq x)\leq \sum_{i=1}^n\left(\frac{c_{p,\epsilon}\sqrt{T} \mathscr{A}_i(T)\max_{t\in[T]}\viii{X_{i,t}}_{p+\epsilon}}{x}\right)^p,\]
where $\mathscr{A}_i(T):=\left[\sum_{0\leq m <T} (m+1)^{p-2}\alpha_{i,m}^{1-p/(p+\epsilon)}\right]^{\tfrac{1}{p}}$, $\alpha_{i,m}$ is the strong mixing coefficient of $\{X_{i,t}\}_t$ for $1\leq i\leq n$, and $c_p,\epsilon$ is a constant depending only on $p$ and $\epsilon$.
\end{lemma}
\begin{proof}
The result follows from the union bound, Markov inequality, and Lemma \ref{L:tail_bound_poly}.
\end{proof}

\begin{lemma}\label{L:tail_bound_sup_exp} Let $\b S_T = \sum_{t=1}^T \b X_t$ where $\{\b X_t:=(X_{1,t},\dots, X_{n,t})':1\leq t\leq T\}$  be a sequence of zero mean $n$-dimensional random variables and write $\b S_T$ for its sum. Assume: 
\begin{enumerate}[(a)]
\item There exist two positive constants $\gamma_1$  and $K_1$
such that the strong mixing coefficients $\alpha_{i,m}\leq \exp(-K_1m^{\gamma_1})$ for  $1\leq m<T$, $1\leq i\leq n$, and $T\geq 2$;
\item There exist two positive constants $\gamma_2$  and $K_2$
such that $\sup_{1\leq t\leq T,T\in \N}\viii{X_{i,t}}_{e^{\gamma_2}}\leq K_2$,
\item $\gamma < 1$ where $\gamma$ is defined by $1/\gamma=1/\gamma_1+1/\gamma_2$.
\end{enumerate}
Then there exist positive constants $C_1$ and $C_2$ depending only on $K_2, K_1, \gamma, \gamma_1$,  and the covariance structure such that, for $n\geq 2$, $T\geq 4\lor C_1(\log n)^{(2/\gamma) - 1}$,  and $K\geq 1/\sqrt{C_1C_2(\log 2)^{2/\gamma}}$, 
\begin{align*}
\P(\|\b S_T\|_\infty\geq K\sqrt{C_2 T\log n})\leq (nT)^{1-K^\gamma} +2n^{1-K^2}.
\end{align*}
In particular, $\|\b S_T\|_\infty\lsim_\P  \big(\sqrt{T\log n}\big)$ whenever $\frac{(\log n)^{(2/\gamma) - 1}}{T}=o(1)$.
\end{lemma}
\begin{proof} Write $\b S_T = (S_{1,T},\dots, S_{n,T})'$ where  $S_{i,T}=\sum_{t=1}^T X_{i,t}$ for $i\in[n]$.  The union bound followed by the second result in Lemma \ref{L:tail_bound_exp} yields, for every $x>1$ and $T\geq 4$, 
 \begin{align*}
\P(\|\b S_T\|_\infty\geq x)&\leq n\max_{i}\P(\|\b S_{i,T}\|_\infty\geq x)\\
&\leq  n\left[T\exp\left(-\frac{x^\gamma}{C_1}\right) + \exp\left(-\frac{x^2}{C_2(1+TV)}\right) +\exp\left(-\frac{x^2}{C_5T}\right)\right].
\end{align*}

Set $x = K\left[(C_1\log(nT))^{1/\gamma}\lor \sqrt{C_6T\log (n)}\right]$ for some $K>0$ where $C_6:= (C_2 +V)\lor C_5$.  For large enough $T$, we have that $x>1$ and, therefore,  $\P(\|\b S_T\|_\infty\geq x)\leq (nT)^{1-K^\gamma} +2n^{1-K^2}$. Notice that that the first in brackets is no larger than the second one provided that $(C_1\log(nT))^{2/\gamma}\leq C_6 T\log (n)$, which in turn is implied by $T\geq C_7(\log n)^{(2/\gamma) - 1}$ where $C_7 = C_6^{-1}(2C_1)^{2/\gamma}$. Therefore, for $T\geq C_7(\log n)^{(2/\gamma) - 1}$ we have that $x = K\sqrt{C_6 T\log n}$ and $x>1$ whenever $K\geq 1/\sqrt{C_6C_7(\log 2)^{2/\gamma}}$.
\end{proof}

\subsection{Factor Model Estimation}

\begin{lemma}\label{L:No_estimation} Let $a_j$ and $b_j$ denote the $j$-th eigenvalue in decreasing order of $\b\Sigma$ and $\b \Lambda\b \Lambda'$  respectively. Then, under Assumption \ref{Ass:Factor_Model}(b) and $(c)$: (a) $b_j\asymp n$ for $1\leq j\leq r$; (b) $\max_{j\leq n}|a_j-b_j\lsim 1$; and (c) $a_j\asymp n$ for $1\leq j\leq r$.
\end{lemma}
\begin{proof}
Result $(a)$ follows from the fact that the $r$ eigenvalues of $\b\Lambda'\b\Lambda$ are also (the only $r$ non-zero) eigenvalues of $\b\Lambda\b\Lambda'$ and Assumption \ref{Ass:Factor_Model}(b). Part $(b)$ follows from Weyl's inequality that implies $\max_{j\leq n}|a_j-b_j|\leq \|\b\Sigma - \b\Lambda\b\Lambda'\|\lsim 1$, where the last equality follows from Assumption \ref{Ass:Factor_Model}(c). Finally, result $(c)$ follows from part $(a)$ and $(b)$ and the (reverse) triangle inequality.
\end{proof}

Recall that $\b\Sigma$ is the $(n\times n)$ covariance matrix of $\b U_t = \b Z_t -\b\Gamma \b W_t$. Let $\widetilde{\b\Sigma}:=\frac{1}{T}\sum_{t=1}^T \b U_t \b U_t'$  and $\widehat{\b\Sigma}$ the same as $\widetilde{\b\Sigma}$ but with $\b \Gamma$ replaced by the estimator $\widehat{\b\Gamma}$. Let $\widehat{a}_j$  denote the $j$-th eigenvalue in decreasing order of $\widehat{\b\Sigma}$.

\begin{lemma}\label{L:estimation} Under the Assumptions \ref{Ass:Factor_Model} and \ref{Ass:Moments}, let $\varrho_1$ be a non-negative sequence of $n$ and $T$ such that $\| \widehat{\b\Sigma}-\widetilde{\b\Sigma}\|_{\max}\lsim_\P \varrho_1$,  then:
\begin{enumerate}[(a)]
    \item $\| \widehat{\b\Sigma}-\b\Sigma\|_{\max} \lsim_\P \varrho_1 +  g_\alpha(n)/\sqrt{T}$ 
    \item $\max_{j\leq n}|\widehat{a}_j -a_j| \lsim_\P n(\varrho_1 +  g_\alpha(n)/\sqrt{T})$
    \item $\widehat{a}_j\asymp_P n$ for $j\leq r$ provided that $\varrho_1 +  g_\alpha(n)/\sqrt{T}\lsim 1$,
\end{enumerate}
where $g_\alpha(n)= \mathcal{A}_\alpha n^{4/p}$ under Assumptions (\ref{Ass:Moments}.c) and $g_\alpha(n)= \sqrt{\log n}$ under Assumptions (\ref{Ass:Moments}.d).
\end{lemma}
\begin{proof}
Part (a) follows by the triangle and Lemma \ref{L:tail_bound_sup_poly} or Lemma \ref{L:tail_bound_sup_exp} , since $ \|\widehat{\b\Sigma}-\b\Sigma\|_{\max}\leq\| \widehat{\b\Sigma}-\widetilde{\b\Sigma}\|_{\max} +\| \widetilde{\b\Sigma}-\b\Sigma\|_{\max} \lsim_\P \varrho_1 + g(n)/\sqrt{T}$. Part (b) follows from Weyl's inequality, the fact that $\| \widehat{\b\Sigma}-\b\Sigma\|\leq n\| \widehat{\b\Sigma}-\b\Sigma\|_{\max}$ and part $(a)$. Part $(c)$ follows from the triangle inequality combined with part $(b)$ and Lemma \ref{L:No_estimation}(c).
\end{proof}

The Lemmas \ref{L:FM_Basic}-\ref{L:10FLM} below are an adaption of Lemmas 8--10 in \cite{FLM2013}, henceforth FLM, to include the estimation error in the sample covariance matrix. To avoid confusion and make it easier for the reader to follow through with the changes we use the same notation adopted in FLM. In particular, if $\delta_{i,t}$ denotes the $(i,t)$ element of $\b \Delta:=\widehat{\b R} - \b R$ then $\widetilde{U}_{i,t} = U_{i,t} + \delta_{i,t}$ for $i\in[n]$ and $t\in[T]$. We consider that $\|\b\Delta\|_{\max}\lsim_\P \varrho_R$ for some non-negative sequence $\varrho_R$ depending on $n$ and $T$.

Define:
\begin{align*}
\widetilde{\zeta}_{st}&:=\frac{\widetilde{\b U}_s' \widetilde{\b U}_t	}{n}-\frac{\E(\b U_s'\b U_t)}{n} = \left(\frac{\b U_s' \b U_t	}{n}-\frac{\E(\b U_s'\b U_t)}{n}\right) +\left(\frac{\b U_s' \b \delta_t	}{n} + \frac{\b\delta_s' \b U_t}{n} + \frac{\b\delta_s' \b\delta_t	}{n} \right)=:	\zeta_{st}+\zeta_{st}^*\\	
\widetilde{\eta}_{st}&:=\frac{\b F_s'\sum_{i=1}^n \b\lambda_i \widetilde{U}_{i,t}	}{n} = \frac{\b F_s'\sum_{i=1}^n \b\lambda_i U_{i,t}	}{n}  + \frac{\b F_s'\sum_{i=1}^n \b\lambda_i \delta_{i,t}	}{n} =:\eta_{st}+\eta_{st}^*\\	
\widetilde{\xi}_{st}&:=\frac{\b F_t'\sum_{i=1}^n \b\lambda_i \widetilde{U}_{is}	}{n} = \frac{\b F_t'\sum_{i=1}^n \b\lambda_i U_{is}	}{n} + \frac{\b F_t'\sum_{i=1}^n \b\lambda_i \delta_{is}}{n} = 	\xi_{st}+\xi_{st}^*.
\end{align*}

\begin{lemma}\label{L:FM_Basic} Under Assumption \ref{Ass:Moments}:
\begin{enumerate}[(a)]
    \item $\zeta_{st}\lsim_\P 1/\sqrt{n}$
    \item $\eta_{st}\lsim_\P 1/\sqrt{n}$
    \item $\xi_{st}\lsim_\P 1/\sqrt{n}$
    \item $\zeta_{st}^*\lsim_\P \varrho_R +  \varrho_R^2$ 
    \item $\max_{s,t\leq T}\zeta_{st}^*\lsim_\P g(nT)\varrho_R +  \varrho_R^2$
    \item $\eta_{st}^*\lsim_\P \varrho_R$
    \item $\xi_{st}^*\lsim_\P \varrho_R$,
\end{enumerate}
 where $g(x)= x^{1/p}$ under Assumptions (\ref{Ass:Moments}.c) and $g(x)= [\log x]^{1/\gamma_2}$ under Assumptions (\ref{Ass:Moments}.d);

\end{lemma}
\begin{proof} Parts $(a)$-$(c)$ are straightforward. For $(d)$ we have that $\frac{1}{n}\b U_s' \b U_t\lsim_\P 1$ and $\frac{1}{n}\b \delta_s' \b \delta_t\leq\|\b\Delta\|_{\max}^2\lsim_\P \varrho_R^2$. Then, the other two terms in parentheses in the definition of $ \zeta_{st}^*$ are $\lsim_\P \varrho_R$ by the Cauchy-Schwartz inequality. Part $(e)$ and $(f)$ follows by similar arguments. Therefore,
\[
\max_{t\leq T}\frac{1}{T}\sum_{s=1}^T\left(\frac{1}{n}\b\delta_s'\b U_t\right)^2= \max_{t\leq T}\frac{1}{n^2}\b U_t'\left(\frac{1}{T}\sum_{s=1}^T\b\delta_s\b\delta_s'\right)\b U_t\leq \|\b\Delta\|_{\max}^2\left(\max_{t\leq T}\|\b U_t\|_1/n\right)^2
\]
and
\[
\zeta_{st}^*\leq \|\b U_s\|_\infty\|\b \delta_t\|_\infty + \|\b U_t\|_\infty\|\b \delta_s\|_\infty + \|\b\delta_t\|_\infty\|\b\delta_s\|_\infty\leq 2\|\b U\|_{\max}\|\b\Delta\|_{\max} + \|\b\Delta\|_{\max}^2.
\]
\end{proof}

\begin{lemma} \label{L:8FLM} Under Assumption \ref{Ass:Moments}:
\begin{enumerate}[(a)]
    \item $\frac{1}{T}\sum_{t=1}^T\left[\frac{1}{nT}\sum_{s=1}^T\widehat{F}_{js}\E(\b U_s'\b U_t)\right]^2 \lsim_\P 1/T$
    \item $\frac{1}{T}\sum_{t=1}^T\frac{1}{T}\sum_{s=1}^T\widehat{F}_{js}\widetilde{\zeta}_{st}^2\lsim_\P \left(1/\sqrt{n}+\varrho_R +\varrho_R^2\right)^2$
    \item$ \frac{1}{T}\sum_{t=1}^T\left[\frac{1}{T}\sum_{s=1}^T\widehat{F}_{js}\widetilde{\eta}_{st}\right]^2\lsim_\P \left(1/\sqrt{n}+\varrho_R\right)^2$
    \item $\frac{1}{T}\sum_{t=1}^T\left[\frac{1}{T}\sum_{s=1}^T\widehat{F}_{js}\xi_{st}\right]^2\lsim_\P \left(1/\sqrt{n}+\varrho_R\right)^2$
    \item $\frac{1}{T}\sum_{t=1}^T\left\|\widehat{\b F}_t -\b H \b F_t\right\|^2\lsim_\P 1/T+(1/\sqrt{n}+\varrho_R +\varrho_R^2)^2$.
\end{enumerate}

\end{lemma}
\begin{proof}
Part (a) is unaltered by the presence of a pre-estimation, so it follows directly from Lemma 8(a) in FLM. For part (b), we have that for $s,l\in[n]$ and $j\in[r]$ by Cauchy-Schwartz inequality
\[
\frac{1}{T}\sum_{t=1}^T\left[\frac{1}{T}\sum_{s=1}^T\widehat{ F}_{js}\widetilde{\zeta}_{st}\right]^2\leq\left[\frac{1}{T^2}\sum_{s,l=1}^T\left(\frac{1}{T}\sum_{t=1}^T\widetilde{\zeta}_{st}\widetilde{\zeta}_{lt} \right)^2\right]^{1/2}.
\]
Since $\widetilde{\zeta}_{st}=\zeta_{st}+\zeta_{st}^*\lsim_\P 1/\sqrt{n} +\varrho_R +\varrho_R^2)$ by Lemma \ref{L:FM_Basic},  the term in parentheses is $\lsim_\P (1/\sqrt{n}+\varrho_R+\varrho_R^2)^2$. Result $(b)$ follows.  For (c), by the triangle inequality and Lemma 8(c) in FLM, we have that  $\|\sum_{i=1}^n\lambda_{ji}\widetilde{u}_{i,t}\|\leq \|\sum_{i=1}^n\lambda_{ji} U_{i,t}\| +\|\sum_{i=1}^n\lambda_{ji}\delta_{i,t}\|\lsim_\P \sqrt{n} + n\varrho_R$. Then, we conclude
\[\frac{1}{T}\sum_{t=1}^T[\frac{1}{T}\sum_{s=1}^T\widehat{\b F}_{s}\widetilde{\eta}_{st}]^2\leq\frac{1}{Tn^2}\sum_{t=1}^T\|\sum_{i=1}^n\b U_{i,t}\lambda_i\|^2\lsim_\P 1/n+\varrho_R/\sqrt{n}+ \varrho_R^2 .\]
The proof of part (d) is analogous to (c) and is omitted.
For (e), let $[\widehat{\b F}_t -\b H \b F_t]_j$ denote the $j$-th entry of $\widehat{\b F}_t -\b H \b F_t$. Since $\b V/n$ is bounded away from zero by Lemma \ref{L:estimation}(c), the fact that $(a+b+c+d)^2\leq 4(a^2 + b^2 + c^2 + d^2)$ and using \eqref{E:Bai_identity}, we have that $\max_{j\leq r}T^{-1}\sum_t[\widehat{\b F}_t -\b H \b F_t]_j$ is upper bounded by some constant $C<\infty$ times
\begin{align*}
&\left\{\max_{j\leq r}\frac{1}{T}\sum_{t=1}^T \left[\frac{1}{T}\sum_{s=1}^T\widehat{F}_{js}\frac{\E(\b U_s'\b U_t)}{n}\right]^2 +
\max_{j\leq r}\frac{1}{T}\sum_{t=1}^T \left(\frac{1}{T}\sum_{s=1}^T\widehat{ f}_{js}\widetilde{\zeta}_{st}\right)^2\right.\\
&\left.\qquad +\max_{j\leq r}\frac{1}{T}\sum_{t=1}^T \left(\frac{1}{T}\sum_{s=1}^T\widehat{F}_{js}\widetilde{\eta}_{st} \right)^2 + \max_{j\leq r}\frac{1}{T}\sum_{t=1}^T\left(\frac{1}{T}\sum_{s=1}^T\widehat{ f}_{js}\widetilde{\xi}_{st}\right)^2\right\}.
\end{align*}
The result follows by applying the bounds from (a)--(d) to each of the terms above.

\end{proof}

\begin{lemma} \label{L:9FLM} Under Assumptions \ref{Ass:Moments} and \ref{Ass:Factor_Model}:
\begin{enumerate}[(a)]
    \item  $\max_{t\leq T}\|\frac{1}{nT}\sum_{s=1}^T\widehat{\b F}_s\E(\b U_s'\b U_t)\| \lsim_\P 1/\sqrt{T})$
\item $\max_{t\leq T}\|\frac{1}{T}\sum_{s=1}^T\widehat{\b F}_s\widetilde{\zeta}_{st}\|\lsim_\P \left[g(T)/\sqrt{n} +g(nT)\varrho_R + \varrho_R^2\right]$ 
\item $\max_{t\leq T}\|\frac{1}{T}\sum_{s=1}^T\widehat{\b F}_{s}\widetilde{\eta}_{st}\|\lsim_\P \left[g(T)/\sqrt{n} +\varrho_R\right]$
\item $\max_{t\leq T}\|\frac{1}{T}\sum_{s=1}^T\widehat{\b F}_{s}\xi_{st}\|\lsim_\P \left[g(T)(1/\sqrt{n} + \varrho_R)\right]$,
\end{enumerate}
where $g(x)= x^{1/p}$ under Assumption (\ref{Ass:Moments}.c) and $g(x)=[\log x]^{1/\gamma_2}$ under Assumption (\ref{Ass:Moments}.d).
\end{lemma}
\begin{proof}
Part (a) is unaltered by pre-estimation, so it follows directly from Lemma 9(a) in FLM. For part (b), from the Cauchy-Schwartz inequality, we have
\[\max_{t\leq T}\left\|\frac{1}{T}\sum_{s=1}^T\widehat{\b F}_s\widetilde{\zeta}_{st}\right\|\leq\left(\frac{1}{T}\sum_{s=1}^T\|\widehat{\b F}_s\|^2 \max_{t\leq T}\frac{1}{T}\sum_{s=1}^T\widetilde{\zeta}_{st}^2\right)^{1/2}.\]
The first summation inside the parentheses equals $r$ due to the normalization. For the second summation, by the triangle inequality, we have $\max_{t\leq T}\frac{1}{T}\sum_{s=1}^T\widetilde{\zeta}_{st}^2\leq \max_{t\leq T}\frac{1}{T}\sum_{s=1}^T\zeta_{st}^2 + 2\max_{t\leq T}\frac{1}{T}\sum_{s=1}^T\zeta_{st}\zeta_{st}^* + \max_{t\leq T}\frac{1}{T}\sum_{s=1}^T{\zeta_{st}^*}^2$. For the first term, the maximum inequality followed by Assumption \ref{Ass:Factor_Model}(e) yields
\[\max_{t\leq T}\frac{1}{T}\sum_{s=1}^T\zeta_{st}^2 \lsim_\P \left[\psi_{p/2}^{-1}(T)\max_{s,t}\|\zeta^2\|_{\psi_{p/2}}\right] \lsim_\P \left[\psi_{p/2}^{-1}(T)\max_{s,t}\|\zeta\|_{\psi}^2\right]\lsim_\P \left[\frac{g_1(T)}{n}\right].\]
where $g_1(x)= x^{2/p}$ or $g_1(x)=[\log x]^{2/\gamma_2}$.
The last one is $\lsim_\P (g(nT)\varrho_R+\varrho_R^2)^2$ by Lemma \ref{L:FM_Basic}(d). Then, by Cauchy-Schwartz, we have that
$\max_{t\leq T}\frac{1}{T}\sum_{s=1}^T\widetilde{\zeta}_{st}^2\lsim_\P (\sqrt{g_1(T)/n} +g(nT)\varrho_R+\varrho_R^2)^2$ and result (b) follows.

For (c), by the triangle inequality we have that $\max_{t\leq T}\|\frac{1}{n}\sum_{i=1}^n\b\lambda_i\widetilde{U}_{i,t}\|\leq \max_{t\leq T}\|\frac{1}{n}\sum_{i=1}^n\b\lambda_iU_{i,t}\| + \max_{t\leq T}\|\frac{1}{n}\sum_{i=1}^n\b\lambda_i\delta_{i,t}\|$. For the first term, the maximum inequality followed by Assumption \ref{Ass:Factor_Model}(f) yields
\[
\max_{t\leq T}\left\|\frac{1}{n}\b\Lambda'\b U_t\right\| \lsim_\P \left[\frac{g(T)}{\sqrt{n}}\max_t\left\|\frac{1}{\sqrt{n}}\b\Lambda'\b U_t\right\|\right] \lsim_\P \left[g(T)/\sqrt{n}\right].
\]
The second term is upper bounded by $r\|\b\Lambda\|_{\max}\|\b\Delta\|_{\max} \lsim_\P \varrho_R)$ by Assumption \ref{Ass:Factor_Model}(d). We obtain the result since
\begin{equation}\label{E:eta}
\max_{t\leq T}\left\|\frac{1}{T}\sum_{s=1}^T\widehat{\b F}_{s}\widetilde{\eta}_{st}\right\|\leq \left\|\frac{1}{T}\sum_{s=1}^T \widehat{\b F}_s  \b F_s' \right\| \max_{t\leq T}\left\|\frac{1}{n}\sum_{i=1}^n\b\lambda_i\widetilde{U}_{i,t}\right\| \lsim_\P \left[\frac{g(T)}{\sqrt{n}} +\varrho_R\right].
\end{equation}
For (d), by the triangle inequality, $\|\frac{1}{nT}\sum_s\sum_i \b\lambda_i\widetilde{U}_{is} \widehat{\b F}_{s}\|\leq \|\frac{1}{nT}\sum_s\sum_i \b\lambda_iU_{is} \widehat{\b F}_{s}\| + \|\frac{1}{nT}\sum_s\sum_i \b\lambda_i\delta_{is} \widehat{F}_{is}\|$. Lemma 9(d) of FLM shows that the first term is $\lsim_\P 1/\sqrt{n}$. For the second term, for each $j\in[r]$:
\[
\left\|\frac{1}{nT}\sum_s\sum_i \b\lambda_i\delta_{is} \widehat{F}_{js}\right\|^2\leq \left(\frac{1}{T}\sum_{s=1}^n\left\|\frac{1}{n}\sum_{i=1}^n \b\lambda_i\delta_{is}\right\|^2 \widehat{F}_{js}\right)\left(\frac{1}{T}\sum_{s=1}^n \widehat{F}_{js}^2\right)\lsim_\P \varrho_R^2).
\]
Thus, $\|\frac{1}{nT}\sum_s\sum_i \b\lambda_i\widetilde{U}_{i,t} \widehat{\b F}_{s}\|\lsim_\P 1/\sqrt{n} + \varrho_R)$ and by Cauchy-Schwartz inequality:
\begin{equation}\label{E:xi}
\max_{t\leq T}\left\|\frac{1}{T}\sum_{s=1}^T\widehat{\b F}_{s}\xi	_{st}\right\|\leq \max_{t\leq T}\left\|\b F_t\right\|\left\|\frac{1}{nT}\sum_s\sum_i \b\lambda_i\widetilde{U}_{i,t} \widehat{\b F}_{s}\right\|\lsim_\P \left[g(T)(1/\sqrt{n} + \varrho_R)\right].
\end{equation}
\end{proof}

\begin{lemma} \label{L:10FLM} Let $\varrho_1 +  \psi^{-1}(n^2)/\sqrt{T} \lsim 1$ where $\varrho_1$ is defined in Lemma \ref{L:estimation}. Then, under Assumption \ref{Ass:Moments} we have:
\begin{enumerate}[(a)]
    \item $\|\b V^{-1}\| \lsim_\P 1/n$
    \item $\|\b H\|\lsim_\P 1$
    \item $\|\b H'\b H - \b I_r\|_F \lsim_\P 1/\sqrt{T} + 1/\sqrt{n} + \varrho_R$
    \item $\max_{i\leq n}\frac{1}{T}\sum_{t=1}^T \widehat{R}_{i,t}^2\lsim_\P \varrho_R(h(nT) + \varrho_R) +g_\alpha(n)/\sqrt{T} +  1$
    \item $\max_{i\leq n}\max_{j\leq r}\frac{1}{T}\sum_{t=1}^T F_{jt} \widetilde{U}_{i,t} \lsim_\P g_\alpha(n)/\sqrt{T} + \varrho_R$,
\end{enumerate}
where $g_\alpha(x) = \mathscr{A}_\alpha x^{2/p}$ or $g(x) = \sqrt{\log x}$, and $h(x) = x^{1/p}$ or $h(x) = [\log x]^{1/\gamma_2}$.
\end{lemma}
\begin{proof} We have that $\b V^{-1} =\diag(1/\widehat{a}_1,\dots, 1/\widehat{a}_r)$ and $1/\widehat{a}_j \asymp_P 1/n$ for $j\leq r$ by Lemma \ref{L:estimation}(c). The result (a) then follows. The normalization tell us  $\|\widehat{\b F}\|=\sqrt{T}$, Lemma 11(a) in FLM give us $\|\b F\|\lsim_\P \sqrt{T})$, $\|\b\Lambda'\b\Lambda\|=\widetilde{a}_1 \asymp n$ by Lemma \ref{L:No_estimation}(a), and from part (a) we have $\|\b V^{-1}\|\lsim_\P 1/n)$. Result (b) then follows since by definition $\b H:=T^{-1}\b V^{-1}\widehat{\b F}'\b F\b\Lambda'\b\Lambda$.
For (c), we have by the triangle inequality,
\[\|\b H'\b H - \b I_r\|_F\leq \|\b H'\b H - \b H' \b F'\b F/T \b H\|_F + \|\b H'\b F'\b F/T\b H - \b I_r\|_F\]
For the first term, we have
\[\|\b H'(\b I_r -  \b F'\b F/T) \b H\|_F\leq\|\b H\|^2\|\b I_r -  \b F'\b F/T\|_F \lsim_\P 1/\sqrt{T}. \]
The second term is equal to $\|\b H'\b F'\b F/T\b H -  \widehat{\b F}' \widehat{\b F}/T\|_F$.

For (d) we have
\begin{align*}
\max_{i\leq n}\frac{1}{T}\sum_{t=1}^T \widehat{R}_{i,t}^2  &\leq \max_{i\leq n}\frac{1}{T}\sum_{t=1}^T (\widehat{R}_{i,t}^2 -R_{i,t}^2) + \max_{i\leq n}\frac{1}{T}\sum_{t=1}^T R_{i,t}^2 -\E(R_{i,t}^2) +\max_{i\leq n}\frac{1}{T}\sum_{t=1}^T \E(R_{i,t}^2)\\
&\leq \max_{i,t} |\widehat{R}_{i,t}^2 -R_{i,t}^2| + \max_{i\leq n}\frac{1}{T}\sum_{t=1}^T R_{i,t}^2 -\E(R_{i,t}^2) +\max_{i,t}\E(R_{i,t}^2).
\end{align*}
The last term is $\lsim 1$ by Assumption \ref{Ass:Moments}(c) or (d), the middle term $\lsim_\P g(n)/\sqrt{T}$. The first term is no larger then $\|\b\Delta\|_{\max}(2\|\b R\|_{\max} + \|\b\Delta\|_{\max}) \lsim_\P \varrho_R(h(nT) + \varrho_R))$. The result (d) then follows.

For (e) we have for each $j\leq r$:
\begin{align*}
    \max_{i\leq n}\max_{j\leq r}\left|T^{-1}\sum_t F_{jt}\widetilde{U}_{i,t}\right|&\leq \max_{i,j}\left|T^{-1}\sum_t F_{jt}U_{i,t}\right| + \max_{i,j}\left|T^{-1}\sum_t F_{jt}\delta_{i,t}\right|\\
    &\leq \max_{i,j}\left|T^{-1}\sum_t F_{jt}U_{i,t}\right| + \max_{i,j}\left(T^{-1}\sum_t F_{jt}^2 T^{-1}\sum_t\delta_{i,t}^2\right)^{1/2}
\end{align*}
The first term is $\lsim_\P g(n)/\sqrt{T}$ by the maximum inequality and Assumption \ref{Ass:Moments} and the second is $\lsim_\P \varrho_R$.
\end{proof}

\subsection{Penalized Regression Results}

\begin{lemma}\label{L:compatibility_cond} Let $\b\Sigma_{\b U}:=\b U'\b U/T$ and $\b\Sigma_{\b V}:=\b V'\b V/T$ where $\b U$ and $\b V$ are $(T\times n)$ matrices, then
\[\|\b\Sigma_{\b U}-\b\Sigma_{\b V}\|_{\max}\leq \|\b U-\b V\|_{\max}(2\|\b V\|_{\max}+\|\b U-\b V\|_{\max}).\]
Furthermore, if $\|\b\Sigma_{\b U}-\b\Sigma_{\b V}\|_{\max}\leq \alpha\kappa(\b\Sigma_{\b V},\S,3)/(|\S|(1+\zeta)^2)$ for $\S\subseteq [n]$, $\zeta>0$ and $\alpha\in[0,1]$, then for $\|\b x_{\S^c}\|_1\leq \zeta \|\b x_{\S}\|_1$ we have
\[(1-\alpha)\b x'\b\Sigma_{\b V}\b x\leq \b x'\b\Sigma_{\b U} \b x\leq (1+\alpha)\b x'\b\Sigma_{\b V}\b x.\]
Take the infimum of the expression above over $\{\b x\in\R^n:\|\b x_{\S^c}\|_1\leq \zeta \|\b x_{\S}\|_1\}$ to conclude
\[(1-\alpha)\kappa^2(\b\Sigma_{\b V},\S,\zeta)\leq \kappa^2(\b\Sigma_{\b U},\S,\zeta)\leq(1+\alpha)\kappa^2(\b\Sigma_{\b V},\S,\zeta).\]
\end{lemma}
\begin{proof} By the (reverse) triangle inequality we have $\|\b U\|_{\max} - \|\b V\|_{\max}\leq \|\b U-\b V\|_{\max}$. Also, $\|\b\Sigma_{\b U}-\b\Sigma_{\b V}\|_{\max}=\max_{1\leq i,j\leq n }|T^{-1}\sum_{t=1}^T U_{i,t}U_{jt} - V_{i,t} V_{ijt}|\leq \max_{i,j,t}|U_{i,t}U_{jt} - V_{i,t} V_{jt}|$ and $
|U_{i,t}U_{jt} - V_{i,t} V_{jt}|\leq |(U_{i,t}-V_{i,t})U_{jt} + (U_{jt}-V_{jt})V_{i,t}|\leq \|\b U-\b V\|_{\max}(\|\b U\|_{\max}+\|\b V\|_{\max}).
$. Combine the 3 bounds to obtain the first result.

For the second part of the lemma notice that for any $\b x\in\R^n$ we have $|\b x'\b\Sigma_{\b U} \b x - \b x'\b\Sigma_{\b V} \b x| = |\b x'(\b\Sigma_{\b U} -\b\Sigma_{\b V} )\b x|\leq\|\b\Sigma_{\b U}-\b\Sigma_{\b V}\|_{\max}\|\b x\|_1^2$. Also, if $\|\b x_{\S^c}\|_1\leq \zeta \|\b x_{\S}\|_1$ we have that $\|\b x\|_1=\|\b x_{\S}\|_1 + |\b x_{\S^c}\|_1\leq (1+\zeta) \|\b x_{\S}\|_1\leq (1+\zeta)\sqrt{\b x'\b\Sigma_{\b V}\b x |\S|}/\kappa(\b\Sigma_{\b V},\S,\zeta)$ where the last inequality follows from the definition of compatibility condition. Thus $|\b x'\b\Sigma_{\b U} \b x - \b x'\b\Sigma_{\b V} \b x| \leq \|\b\Sigma_{\b U}-\b\Sigma_{\b V}\|_{\max}(1+\zeta)^2\b x'\b\Sigma_{\b V}\b x |\S|/\kappa(\b\Sigma_{\b V},\S,\zeta)\leq\b x'\b\Sigma_{\b V}\b x/2 $, where the last inequality follows from the definition of compatibility condition. Therefore, we have that $(1-\alpha)\b x'\b\Sigma_{\b V}\b x\leq \b x'\b\Sigma_{\b U} \b x\leq (1+\alpha)\b x'\b\Sigma_{\b V}\b x$ whenever $\|\b x_{\S^c}\|_1\leq \zeta \|\b x_{\S}\|_1$.
\end{proof}

\begin{lemma} \label{L:bound_LASSO} Let $\b W:= (\b U,\b V)$ and $\b Z:=(\b X,\b Y)$ be $T\times (n+1)$ matrices such that $\|\b W - \b Z\|_{\max}\leq C_1$ and $\|\b Z\|_{\max}\leq C_2$, then for any $\b \delta\in\R^n$ we have
\[\|\b U'(\b V - \b U\b\delta)/T-\b X'(\b Y-\b X\b\delta)/T\|_\infty\leq  (1+\|\b\delta\|_1)C_1(2C_2+C_1) \]
\end{lemma}
\begin{proof} For convenience let $q:=V-U\delta\in\R^T$ and $r:=Y-X\delta\in\R^T$, then Hölder's inequality gives us $\|r\|_\infty\leq (1+\|\delta\|_1)\|Z\|_{\max}\leq (1+\|\delta\|_1)C_2$ and $\|q-r\|_\infty\leq (1+\|\delta\|_1)\|W-Z\|_{\max}\leq (1+\|\delta\|_1)C_1$. From the (reverse) triangle inequality we obtain $\|q\|_\infty\leq \|q-r\|_\infty + \|r\|_\infty\leq (1+\|\delta\|_1)(C_1 + C_2)$. Now, following the same steps in the proof of the previous Lemma, we can upper bound the right-hand side of the display by $\|U-X\|_{\max}\|q\|_\infty + \|q-r\|_\infty\|X\|_{\max}$, which in turn can be upper bounded by  the left-hand side of the display.
\end{proof}

\subsection{Inference Procedure Results}

\begin{lemma}\label{L:useful_decompostion} Let $X$ and $Y$ be random elements defined in the same probability space $(\Omega,\m{F},\P)$ taking values in the metric space $(S,d)$. Then for measurable $A$ and $\delta\geq 0$
\begin{align*}
-\P(Y\in A\setminus A^{-\delta}) - \P(d(X,Y)>\delta) &\leq \P(X\in A) - \P(Y\in A)\\
&\leq \P(Y\in A^\delta\setminus A) + \P(d(X,Y)>\delta),
\end{align*}
where $A^\delta:=\{x\in S:d(x,A)\leq \delta\}$, $A^{-\delta}:=S\setminus(S\setminus A)^\delta$ and $d(x,A):=\inf_{y\in A} d(x,y)$.

Let $\m{A}$ be a class of measurable subsets of $S$ then
\[
\rho_\m{A}(X,Y)\leq \inf_{\delta>0}\big[ \P(d(X,Y)>\delta)+\Delta_\mathcal{A}(Y,\delta) \big],  
\]
where $\rho_\m{A}(X,Y):=\sup_{A\in \m{A}}|\P(X\in A) -\P(Y\in A)|$ and $\Delta_\m{A}(Y,\delta):=\sup_{A\in \m{A}}\P(Y\in A^\delta\setminus A^{-\delta})$.

In particular, if $d$ is induced by a norm $\|\cdot\|$ we have for all $t\in\R$ and $\delta>0$
\begin{align*}
-\P(t-\delta\leq\|Y\|\leq t) - \P(\|X-Y\|>\delta) &\leq \P(\|X\|\leq t) - \P(\|Y\|\leq t)\\
&\leq \P(t<\|Y\|\leq t+\delta) + \P(\|X-Y\|>\delta).
\end{align*}

\begin{proof}
For the right hand side, we use $\{X\in A\}\cap\{d(X,Y)\leq \delta\}\subseteq\{Y\in A^\delta\}$ to write
\begin{align*}
\P(X\in A) &=  \P(X\in A, d(X,Y)\leq \delta)  + \P(X\in A,d(X,Y)> \delta)\\
&\leq \P(Y\in A^\delta)  + \P(d(X,Y)> \delta)\\
&=  \P(Y\in A) + \P(Y\in A^\delta\setminus A)  + \P(d(X,Y)> \delta).
\end{align*}
For the left hand side, we use $\{Y\in A^{-\delta}\}\cap\{d(X,Y)\leq \delta\}\subseteq\{X\in A\}$ to write
\begin{align*}
\P(X\in A) &\geq   \P(Y\in A^{-\delta}, d(X,Y)\leq \delta)  \\
&\geq \P(Y\in A^{-\delta})  - \P(d(X,Y)> \delta)\\
&=  \P(Y\in A) - \P(Y\in A\setminus A^{-\delta})  - \P(d(X,Y)> \delta).
\end{align*}
The first result follows. For the second one we have
\begin{align*}
|\P(X\in A) - \P(Y\in A)|&\leq  \P(d(X,Y)>\delta) +\P(Y\in A^\delta\setminus A)\lor \P(Y\in A\setminus A^{-\delta})\\
&\leq  \P(d(X,Y)>\delta) +\P(Y\in A^\delta\setminus A^{-\delta}).
\end{align*}
Take the supremum over $A\in\m{A}$ we have $\rho_\m{A}(X,Y)\leq \P(d(X,Y)>\delta)+\Delta_\mathcal{A}(Y,\delta)$. Take the infimum over $\delta>0$
to obtain the second result.

For the last one, take $A=B_t$ where $B_t$ is $\|\cdot\|$-closed ball of radius $t$ centered at the origin for $t\geq 0$ and the empty set otherwise, then $\{X\in A\}=\{\|X\|\leq t\}$ and $\{Y\in A\}=\{\|Y\|\leq t\}$. Also $\{Y\in A^\delta\setminus A\} = \{t<\|Y\|\leq t+\delta\}$ since $A^\delta = B_{t+\delta}$ and $\{Y\in A\setminus A^{-\delta}\} = \{t-\delta\leq\|Y\|\leq t\}$ since $A^{-\delta}$ is the open ball of radius $t-\delta$
\end{proof}
\end{lemma}

\begin{lemma}\label{L:mean_bound} Let $\|\widehat{\b U}- \b U\|_{\max}\lsim_\P \varrho_U$ for some positive sequence of $n$ and $T$, then
\[
\left\| \frac{1}{\sqrt{T}} (\widehat{\b U}\widehat{\b U}' - \b U\b U') \right\|_{\max}\lsim_\P \sqrt{T}\varrho_U^2+ \mathscr{R}_\alpha\left(\frac{ n^{9/p}}{\sqrt{T}}+\frac{n^{6/p}}{T^{1/2-1/p}}\right) +\frac{1}{n^{1/2-9/p}}.
\]
\end{lemma}
\begin{proof} By the triangle inequality we have
\begin{align*}
 \left\| \frac{1}{\sqrt{T}} (\widehat{\b U}\widehat{\b U}' - \b U\b U') \right\|_{\max}\leq   \left\| \frac{1}{\sqrt{T}} (\widehat{\b U} - \b U)(\widehat{\b U} - \b U)' \right\|_{\max} + 2  \left\| \frac{1}{\sqrt{T}} \b U(\widehat{\b U} - \b U)' \right\|_{\max}.
\end{align*}
For the first term, we have
\begin{align*}
       \left\| \frac{1}{\sqrt{T}} (\widehat{\b U} - \b U)(\widehat{\b U} - \b U)' \right\|_{\max}\leq\sqrt{T}\|\widehat{\b U} - \b U\|_{\max}^2 \lsim_\P \sqrt{T}\varrho_U^2).
\end{align*}
For the second term we use decomposition \eqref{E:factor_decomposition} to write
\begin{align*}
       \frac{1}{\sqrt{T}}\sum_{t=1}^T U_{i,t}(\widehat{U}_{jt} - U_{jt}) &= \frac{1}{\sqrt{T}}\sum_{t=1}^T U_{i,t}(\widehat{\b\lambda}_{j}'\widehat{\b F}_t - \b\lambda_{j}'\b F_t + \widehat{R}_{jt} - R_{jt})\\
       &= \left[(\widehat{\b\lambda}_j - \b H \b\lambda_j) +\b H\b\lambda_j\right]'\frac{1}{\sqrt{T}}\sum_{t=1}^T U_{i,t}(\widehat{\b F}_t- \b H\b F_t)\\
       &\quad + \left[(\widehat{\b\lambda}_j - \b H \b\lambda_j) +(\b H'\b H - I_r)\b\lambda_j\right]'\frac{1}{\sqrt{T}}\sum_{t=1}^T U_{i,t}\b F_t\\
       &\quad + (\widehat{\b\gamma}_j -  \b\gamma_j)'\frac{1}{\sqrt{T}}\sum_{t=1}^T U_{i,t}\b X_{jt}.
\end{align*}
Apply Cauchy-Schwartz inequality in each term followed by the triangle inequality we obtain
\begin{align}\label{E:split}
      \left\| \frac{1}{\sqrt{T}} \b U(\widehat{\b U} - \b U)' \right\|_{\max}
       &\leq \left[\max_{j\leq n}\|\widehat{\b\lambda}_j - \b H \b\lambda_j\| +\sqrt{r}\|\b H\|\|\b\Lambda\|_{\max}\right]\max_{i\leq n}\left\|\frac{1}{\sqrt{T}}\sum_{t=1}^T U_{i,t}(\widehat{\b F}_t- \b H\b F_t)\right\|\nonumber\\
       &\quad + \left[\max_{j\leq n}\|\widehat{\b\lambda}_j - \b H \b\lambda_j\| +\sqrt{r}\|\b H'\b H - I_r\|\|\b\Lambda\|_{\max}\right]\max_{i\leq n}\left\|\frac{1}{\sqrt{T}}\sum_{t=1}^T U_{i,t}\b F_t\right\|\nonumber\\
       &\quad + \max_{j\leq n}\|\widehat{\b\gamma}_j -  \b\gamma_j\|\max_{i,j\leq n}\left\|\frac{1}{\sqrt{T}}\sum_{t=1}^T U_{i,t}\b X_{jt}\right\|.
\end{align}

Recall $\delta_{i,t}:=\widehat{R}_{i,t}-R_{i,t}$. From \eqref{E:LS_decomposition0} and Holder's inequality, we conclude that
\[
    \max_{i,t}\viii{\delta_{i,t}}_{p/4}\lsim \viii{\|\widehat{\b\Sigma}_i^{-1}\|}_p \viii{\|\widehat{\b v}_i\|_\infty}_{p/2}\viii{\|\b X_{i,t}\|_\infty}_p\lsim \mathscr{R}_\alpha/\sqrt{T},
\]
where $\mathscr{R}_\alpha$ is defined in Theorem \ref{T:LS}. We now use this last result to show that
\begin{equation}\label{E:V_F}
    \viii{\|(\b V/n)(\b F_t - \b H \b F_t)\|_2}_{p/8} = O\left(\frac{\mathscr{R}_\alpha}{\sqrt{T}} +\frac{1}{\sqrt{n}}\right),
\end{equation}
which in turns uses identity \eqref{E:Bai_identity} and the fact that for any random variable $A_{st}$, by Cauchy-Schwartz inequality and the normalization $\widehat{\b F}\widehat{\b F}/T =\b I_r$, we have $\|\frac{1}{T}\sum_{s=1}^T \widehat{\b F}_s A_{st}\|\leq \sqrt{r}\left(\frac{1}{T}\sum_{s=1}^T A_{st}^2\right)^{1/2}$. Thus
\[g(A_{st}):=\viii{\left\|\frac{1}{T}\sum_{s=1}^T \widehat{\b F}_s A_{st}\right\|}_{p/8} \leq \sqrt{r} \left[\viii{\left(\frac{1}{T}\sum_{s=1}^T A_{st}^2\right)^{1/2}}_{p/8}\right].\]
\begin{enumerate}[(a)]
    \item Set $A_{st} = \E(\b U_s'\b U_t)/n$, then $g(A_{st})\lsim 1/\sqrt{T}$.
    \item Set $A_{st}=\widetilde{\zeta}_{st}:=(\widetilde{\b U}_s'\widetilde{\b U}_t - \E(\b U_s'\b U_t))/n$. By the triangle inequality, $\viii{\left(\frac{1}{T}\sum_{s=1}^T A_{st}^2\right)^{1/2}}_{p/8}\leq \max_{s\in[T]}\viii{\widetilde{\zeta}_{st}}_{p/8}$, and  $\viii{\widetilde{\zeta}_{st}}_{p/8}\leq \viii{\zeta_{st}}_{p/8} + \viii{\zeta_{st}^*}_{p/8}$. The first term is $\lsim 1/\sqrt{n}$ by Assumption \ref{Ass:Moments}(c.2). The second can be upper bounded by $\viii{\b U_s'\b\delta_t/n}_{p/8} + \viii{\b \delta_s'\b U_t/n}_{p/8} +\viii{\b \delta_s'\b\delta_t/n}_{p/8}\lsim \max_{i\in[n]}\viii{U_{i,s}}_{p/4}\viii{\delta_{i,t}}_{p/4} + \max_{i\in[n]}\viii{\delta_{i,t}}_{p/4}^2$. Thus $g(A_{st}) \lsim 1/\sqrt{n} + \mathscr{R}_\alpha/\sqrt{T} $.
    \item Set $A_{st} =\widetilde{\eta}_{st}:=\b F_s'\sum_{i=1}^n \b\lambda_i (U_{i,t}+\delta_{i,t})/n$, then apply Cauchy-Schwartz twice to obtain
    \begin{align*}
        g(A_{st})&\leq\left(\viii{\left(\frac{1}{T}\sum_{s=1}^T \|\b F_s\|^2\right)^{1/2}}_{p/4}\viii{\sum_{i=1}^n\b\lambda_{i}\frac{U_{i,t}+\delta_{i,t}}{n}}_{p/4}\right)\\
        &\lsim\viii{\sum_{i=1}^n \b\lambda_i \frac{U_{i,t}}{n}}_{p/4} + \viii{\sum_{i=1}^n \b\lambda_i\frac{\delta_{i,t}}{n}}_{p/4}.
    \end{align*}
    The first term is $\lsim 1/\sqrt{n}$  by Assumption \ref{Ass:Factor_Model}(d) and \ref{Ass:Moments}(c.1); the second is $\lsim \mathscr{R}_\alpha/\sqrt{T}$. Hence $g(A_{st})\lsim \frac{1}{\sqrt{n}} +  \mathscr{R}_\alpha/\sqrt{T}$.
    \item Set $A_{st} =\widetilde{\xi}_{st}:=\b F_t'\sum_{i=1}^n \b\lambda_i (U_{is}+\delta_{is})/n$, then apply Cauchy-Schwartz twice followed by the maximal inequality to obtain
    \begin{align*}
        g(A_{st})&\leq\viii{\|\b F_t\|^{1/2}}_{p/4}\viii{\left(\frac{1}{T}\sum_{s=1}^T\left\|\sum_{i=1}^n\b\lambda_{i}\frac{U_{is}+\delta_{is}}{n}\right\|^2\right)^{1/2}}_{p/4}\\
        &\lsim \max_{s\in[T]}\viii{\sum_{i=1}^n \b\lambda_i \frac{U_{is}}{n}}_{p/4} + \max_{s\in[T]}\viii{\sum_{i=1}^n \b\lambda_i\frac{\delta_{is}}{n}}_{p/4}.
    \end{align*}
    The first term in square brackets is $\lsim 1/\sqrt{n}$  by Assumption \ref{Ass:Factor_Model}(d) and \ref{Ass:Moments}(e); the second is $\lsim  \mathscr{R}_\alpha/\sqrt{T}$. Hence $g(A_{st}) \lsim \frac{1}{\sqrt{n}} +  \mathscr{R}_\alpha/\sqrt{T}$.
\end{enumerate}
Finally, use the identity \eqref{E:Bai_identity}, the triangle inequality twice and the bounds $(a)$-$(d)$ to obtain \eqref{E:V_F}. Also, by Holder's inequality
\begin{align*}
    \viii{\|U_{i,t}(\b V/n)(\b F_t - \b H \b F_t)\|_2}_{p/9} &\leq \viii{U_{i,t}}_{p} \viii{\|U_{i,t}(\b V/n)(\b F_t - \b H \b F_t)\|_2}_{p/8}\\
    &\lsim \frac{ \mathscr{R}_\alpha}{\sqrt{T}} +\frac{1}{\sqrt{n}}.
\end{align*}

The first term of \eqref{E:split} is $\lsim_\P n^{9/p}\left(\frac{\mathscr{R}_\alpha}{\sqrt{T}} +\frac{1}{\sqrt{n}}\right)$ due to Lemma \ref{L:10FLM}(a), the result above and the maximal inequality; the second term is $\lsim_\P(\frac{\mathscr{R}_\alpha  n^{4/p}}{T^{1/2-1/p}}  + \frac{1}{\sqrt{n}})n^{2/p}$ since, from Theorem \ref{T:LS}, we might take $\varrho_R = \frac{\mathscr{R}_\alpha  n^{4/p}}{T^{1/2-1/p}}$ in Theorem \ref{T:FM}(b). The last term is $\lsim_\P (\mathscr{R}_\alpha n^{3/p}/\sqrt{T})n^{4/p}$. Thus
\begin{equation}\label{E:rate_UU}
  \left\| \frac{1}{\sqrt{T}} \b U(\widehat{\b U} - \b U)' \right\|_{\max} \lsim_\P \mathscr{R}_\alpha\left(\frac{ n^{9/p}}{\sqrt{T}}+\frac{n^{6/p}}{T^{1/2-1/p}}\right) +\frac{1}{n^{1/2-9/p}}.
\end{equation}
The result then follows.
\end{proof}

\begin{lemma}\label{L:NW_bound} Let $\|\widehat{\b U}- \b U\|_{\max}\lsim_\P \varrho_U$ for some non-negative sequence $\varrho_U$ depending on $n$ and $T$, then, for all $ h >0$ and $\m{D}\subseteq[n]^2$,
\[
\|\widehat{\b\Upsilon}_\Sigma - \b\Upsilon_\Sigma\|_{\max}\lsim_\P h[\varrho_U(nT)^{3/p} +d^{8/p}/\sqrt{T}]
\]
in the polynomial case; and 
\[
\|\widehat{\b\Upsilon}_\Sigma - \b\Upsilon_\Sigma\|_{\max}\lsim_\P h[\varrho_U(\psi^{-1}_p(nT))^3 +\psi_{p/4}^{-1}(n^4)/\sqrt{T}] 
\]
in the exponential case, where $d:=|\m{D}|$.
\end{lemma}
\begin{proof} Let $\b i:=(i_1,i_2,i_3,i_4)\in\m{D}^2$ be a multi-index where $(i_1,i_2)\in\m{D}$ and $(i_3,i_4)\in\m{D}$. Define for $\b i$ and $|\ell|<T$:
\[\widetilde{\gamma}_{\b i}^\ell :=\frac{1}{T}\sum_{t=|\ell| +1}^T U_{i_1,t}U_{i_2,t}U_{i_3,t-|\ell|}U_{i_4,t-|\ell|};\qquad  \gamma_{\b i}^\ell := \E\widetilde{\gamma}_{\b i},\]
and $\widehat{\gamma}_{\b i}^\ell$ as $ \widetilde{\gamma}_{\b i}^\ell$ with $U$'s replaced by $\widehat{U}$'s. Also, define
\[\widetilde{\upsilon}_{\b i} := \sum_{|\ell|<T} k(\ell/h)\widetilde{\gamma}_{\b i}^\ell\qquad \upsilon_{\b i} := \sum_{|\ell|<T}\gamma_{\b i}^\ell,\]
and $\widehat{\upsilon}_{\b i}$ as $ \widetilde{\upsilon}_{\b i}$ with $U$'s replaced by $\widehat{U}$'s. Then we write
\begin{align}
  \widetilde{\upsilon}_{\b i}-\upsilon_{\b i} =   \sum_{|\ell|<T} k(\ell/h)(\widetilde{\gamma}_{\b i}^\ell - \gamma_{\b i}^\ell) + \sum_{|\ell|<T} (k(\ell/h)-1)\gamma_{\b i}^\ell.
\end{align}
By Cauchy-Schwartz inequality we have that $\viii{U_{i_1,t}U_{i_2,t}U_{i_3,t-|\ell|}U_{i_4,t-|\ell|}}_{(p+\epsilon)/4}\leq C^4$, then by Lemma \ref{L:tail_bound_poly} we obtain $\viii{\widetilde{\gamma}_{\b i}^\ell - \gamma_{\b i}^\ell}_{p/4} = \lsim \sqrt{T-|\ell|}/T = \lsim 1/\sqrt{T}$,
 the $\m{L}_{p/4}$-norm of the first term is bounded by
\[h\sum_{|\ell|<T} |h^{-1} k(\ell/h)|\viii{\widetilde{\gamma}_{\b i}^\ell - \gamma_{\b i}^\ell}_{p/4} \lsim \left( \frac{h}{\sqrt{T}}\int |k(u)|du \right)= \lsim h/\sqrt{T},\]
whereas the second term is deterministic and is shown to be $\lsim h/\sqrt{T}$ by \cite{andrews91}. Thus $\viii{\widetilde{\upsilon}_{\b i}-\upsilon_{\b i}}_{p/4} = \lsim h/\sqrt{T})$ uniformly in $ \b i\in[n]^4$. Thus, union bound followed by Markov's inequality give us, for $x>0$, 
\begin{equation}\label{E:bond1_cov}
\P(\max_{\b i}|\widetilde{\upsilon}_{\b i}-\upsilon_{\b i}|\geq x)\leq d^2\max_{\b i}\P(|\widetilde{\upsilon}_{\b i}-\upsilon_{\b i}|\geq x)\lsim\left(\frac{d^{8/p}h}{x\sqrt{T}}\right)^{p/4}.
\end{equation}

We now use the fact that for any $\b x:=(x_1,\dots,x_4)'\in\R^4$ and $\b y:=(y_1,\dots, y_4)'\in\R^4$ we have $|\prod_{i=1}^4 x_i - \prod_{i=1}^4 y_i|\lsim \sum_{i=0}^{3}\|\b x-\b y\|_\infty^{3-i}\|\b y\|_\infty^i$ combined with the fact that $\|\widehat{\b U} - \b U\|_{\max}\lsim 1$ to obtain

\begin{align*}
   \max_{\b i, \ell} |\widehat{\gamma}_{\b i}^\ell - \widetilde{\gamma}_{\b i}^\ell|&\leq \max_{\b i,t,\ell}|\widehat{U}_{i_1,t}\widehat{U}_{i_2,t}\widehat{U}_{i_3,t-|\ell|}\widehat{U}_{i_4,t-|\ell|} - U_{i_1,t}U_{i_2,t}U_{i_3,t-|\ell|}U_{i_4,t-|\ell|}|\\
   &\lsim (\|\widehat{\b U} - \b U\|_{\max}\|\b U\|_{\max}^3)\\
   &\lsim_\P \{\varrho_U[\psi^{-1}(nT)]^3\}
\end{align*}
Therefore we conclude
\begin{align}\label{E:bond2_covb}
   \max_{\b i}|\widehat{\upsilon}_{\b i} - \widetilde{\upsilon}_{\b i}|&\leq \max_{\b i, \ell} |\widehat{\gamma}_{\b i}^\ell - \widetilde{\gamma}_{\b i}^\ell|\sum_{|\ell|<T} |k(\ell/h)|\\
   &\lsim_\P \left\{ h\varrho_U[\psi^{-1}(nT)]^3\int |k(u)|du\right\} \lsim_\P h\varrho_U[\psi^{-1}(nT)]^3.\nonumber
\end{align}

The result then follows from the triangle inequality $\|\widehat{\b\Upsilon} - \b\Upsilon\|_{\max}\leq \max_{\b i}|\widehat{\upsilon}_{\b i} - \widetilde{\upsilon}_{\b i}| +\max_{\b i}|\widetilde{\upsilon}_{\b i}-\upsilon_{\b i}|$, expression \eqref{E:bond1_cov} and \eqref{E:bond2_covb}.
\end{proof}

\begin{lemma}\label{L:V_bond} Let $\|\widehat{\b U}- \b U\|_{\max}\lsim_\P \varrho_U$ and $\max_{i,j\in[n]}\|\widehat{\b\chi}_{i,j} - \b\chi_{i,j} \|_1\lsim_\P \varrho_\chi$ for some non-negative sequences $\varrho_U$ and $\varrho_\chi$ depending on $n$ and $T$,  then
\begin{enumerate}[(a)]
    \item $\max\limits_{(i,j)\in\m{D},t\in [T]} |\widehat{V}_{i,j,t}-V_{i,j,t}|\lsim_\P (1 + \widetilde{s}_1) \varrho_U + \varrho_\chi n^{1/p}$;
    \item $\max\limits_{(i,j)\in\m{D}} \left|\frac{1}{\sqrt{T}}\sum_{t=1}^T(\widehat{V}_{i,j,t}\widehat{V}_{j,i,t} - V_{i,j,t}V_{i,j,t}) \right|\lsim_\P (1+\widetilde{s}_1+\varrho_\chi)^2 \varrho_{UU} + \varrho^2_\chi n^{4/p}\sqrt{T}$,
\end{enumerate}
where $\widetilde{s}_1:=\max_{(i,j)\in\m{D}}\|\b\chi_{i,j}\|_1$ and $\varrho_{UU}$ is the rare appearing in Lemma \ref{L:mean_bound}.
\end{lemma}

\begin{proof} By the triangle inequality we have, for $i,j\in[n]$ and $t\in[T]$,
\begin{align*}
\widehat{V}_{i,j,t}-V_{i,j,t}
&= \widehat{U}_{i,t} -U_{i,t} -\widehat{\b \chi}_{i,j}' \widehat{\b U}_{-ij,t} + \b \chi_{i,j}'\b U_{-ij,t}\\
&= \widehat{U}_{i,t} -U_{i,t} -\widehat{\b \chi}_{i,j}' (\widehat{\b U}_{-ij,t} -\b U_{-ij,t})  - (\widehat{\b \chi}_{i,j} -\b \chi_{i,j})'\b U_{-ij,t}\\
&= \widehat{U}_{i,t} -U_{i,t}-\big[\b \chi_{i,j} + (\widehat{\b \chi}_{i,j} -\widehat{\b \chi}_{i,j})\big]' (\widehat{\b U}_{-ij,t} -\b U_{-ij,t})  -(\widehat{\b \chi}_{i,j} -\b \chi_{i,j})'\b U_{-ij,t},
\end{align*}

Therefore, result (a) follows by Hölder's and the triangle inequality since
\begin{align*}
    \max_{i,j,t} |\widehat{V}_{i,j,t}-V_{i,j,t}|
&\leq (\|\b \chi_{i,j}\|_1 +\| \widehat{\b \chi}_{i,j} -\b \chi_{i,j}\|_1 )\|\widehat{\b U}_{-ij,t} - \b U_{-ij,t}\|_\infty \\
&\qquad +|\widehat{U}_{i,t}- U_{i,t}|  + \|\widehat{\b \chi}_{i,j} -\b \chi_{i,j}\|_1\|\b U_{-ij,t}\|_\infty \\
    &\lsim (1+\max_{(i,j)\in\m{D}}\|\b\chi_{i,j}\|_1) \|\widehat{\b U}-\b U\|_{\max} + \max_{(i,j)\in\m{D}}\|\widehat{\b\chi}_{i,j}-\b\chi_{i,j}\|_1\|\b U\|_{\max}\\
    &\lsim_\P \big[(1 +\max_{(i,j)\in\m{D}}\|\b\chi_{i,j}\|_1) \varrho_U + \varrho_\chi n^{1/p}\big].
\end{align*}

For (b), we write for $i,j\in[n]$ and $t\in[T]$
\begin{align*}
    \widehat{V}_{i,j,t}\widehat{V}_{j,i,t} - V_{i,j,t}V_{i,j,t}&=\widehat{U}_{i,t}^2 - U_{i,t}^2 +\widehat{\b\chi}_{i,j}'(\widehat{\b U}_{-ij,t}\widehat{\b U}_{-ij,t}' - \b U_{-ij,t}\b U_{-ij,t}')\widehat{\b\chi}_{i,j}\\
    &\qquad +(\widehat{\b\chi}_{i,j} - \b\chi_{i,j})' \b U_{-ij,t}\b U_{-ij,t}'(\widehat{\b\chi}_{i,j} - \b\chi_{i,j}).
\end{align*}
Then, by Holder's and the triangle inequality, we obtain
\begin{align*}
       \max_{i,j}\left|\tfrac{1}{\sqrt{T}}\sum_{t\in[T]}(\widehat{V}_{i,j,t}\widehat{V}_{j,i,t} - V_{i,j,t}V_{i,j,t})\right| &\leq (1+\max_{i,j\in[n]}\|\widehat{\b\chi}_{i,j}\|_1)^2\|\tfrac{1}{\sqrt{T}}(\widehat{\b U}\widehat{\b U}'-\b U\b U')\|_{\max} \\ 
       &\quad +\max_{i,j\in[n]}\|\widehat{\b\chi}_{i,j} - \b\chi_{i,j}\|_1^2\|\tfrac{1}{\sqrt{T}}\b U\b U')\|_{\max}.
\end{align*}
Now, $\max_{i,j\in[n]}\|\widehat{\b\chi}_{i,j}\|_1\leq \max_{i,j\in[n]}\|\b\chi_{i,j}\|_1 +\max_{i,j\in[n]}\|\widehat{\b\chi}_{i,j} - \b\chi_{i,j}\|_1 \lsim_\P \widetilde{s}_1+\varrho_\chi$. The order in probability of $\|\tfrac{1}{\sqrt{T}}(\widehat{\b U}\widehat{\b U}'-\b U\b U')\|_{\max}$ is given by Lemma \ref{L:mean_bound}. By Cauchy-Schwartz inequality and Assumption \ref{Ass:Moments} we have $\viii{U_{i,t} U_{j,t}}_{p/2}\leq \viii{U_{i,t}}_p\viii{U_{j,t}}_p\leq C^2$, then by the triangle inequality, $\viii{\sum_{t\in[T]}U_{i,t}U_{j,t}}_{p/2}\leq \sum_{t\in[T]}\viii{U_{i,t} U_{j,t}}_{p/2}\lsim T$. Finally, by the union bound followed by Markov's inequality, $\|\tfrac{1}{\sqrt{T}}\b U\b U')\|_{\max}\lsim_\P n^{4/p}\sqrt{T}$. The result follows.
\end{proof}

\begin{lemma}\label{L:NW_bound_partial}
Let $\|\widehat{\b U}- \b U\|_{\max}\lsim_\P \varrho_U$ and $\max_{i,j\in[n]}\|\widehat{\b\chi}_{i,j} - \b\chi_{i,j} \|_1\lsim_\P \varrho_\chi$ for some non-negative sequences $\varrho_U$ and $\varrho_\chi$ depending on $n$ and $T$ then
\[
\|\widehat{\b\Upsilon}_\Pi - \b\Upsilon_\Pi\|_{\max}\lsim_\P h\left((1 + \widetilde{s}_1) \varrho_U + \varrho_\chi n^{1/p}(1+\widetilde{s}_1)^3(nT)^{3/p} + \frac{(1+\widetilde{s}_2)^4n^{8/p}}{\sqrt{T}}\right),
\]
where $\widetilde{s}_k:=\max_{(i,j)\in\m{D}}\|\b\chi_{i,j}\|_k$ for $k\in\{1,2\}$.
\end{lemma}
\begin{proof} The proof is parallel to the proof of Lemma \ref{L:NW_bound}, refer to it for details. It suffices to bound in $\max_{i,j,t}\viii{V_{i,j,t}}_{p+\epsilon}$, $\max_{i,j,t}|\widehat{V}_{i,j,t}-V_{i,j,T}|$, and $\max_{i,j,t}  |V_{i,j,t}|$, where the maximum is over $(i,j)\in\m{D}\subseteq[n]^2$ and $t\in[T]$. For the first term we have, $\max_{i,j,t}\viii{V_{i,j,t}}_{p+\epsilon}\leq (1+\max_{i,j}\|\b\chi_{i,j}\|_2)\sup_{t\in[T]}\viii{\b U_t}_{p+\epsilon}\leq (1+M_2)C$ by Assumption \ref{Ass:Moments}. Lemma \ref{L:V_bond}(a) bounds the second term. The last one is upper bounded by $(1+\max_{(i,j)\in\m{D}}\|\b\chi_{ij}\|_1)\|\b U \|_{\max} \lsim_\P \left[(1+M)(nT)^{1/p}\right]$.
\end{proof}

\subsection{Orlicz Norm Results}

In this subsection, we show some useful results for  Orlicz norms. All these results can be found in \cite{VW1996} for the case of $\gamma\geq 1$. Recall that for $\gamma\in (0,1)$ we define
\[
\psi_{e^\gamma}(x) := \co (x\mapsto \exp(x^\gamma)-1),
\] 
where $\co(f)$ denote the convex hull of $f$.

\begin{lemma}\label{L:Orlicz_basic} We have
\begin{enumerate}[(a)]
    \item $\psi_{e^\gamma}(x) = K_\gamma x\1_{0\leq x<a_\gamma} +[\exp(x^\gamma)-1]\1_{x\geq a_\gamma}$ where $K_\gamma :=(\exp a_\gamma^\gamma-1)/a_\gamma$ and $a_\gamma:=\inf\{x\in\R_+: x\geq ((1-\gamma)/\gamma)^{1/\gamma}, K_\gamma\leq \gamma\exp(x^\gamma)/x^{1-\gamma}$. Also $((1-\gamma)/\gamma)^{1/\gamma}\leq a_\gamma\leq (1/\gamma)^{1/\gamma}$.
    \item For $p\in[1,\infty)$ and $\gamma>0$, $\viii{X}_p\leq C\viii{X}_{e^\gamma}$ for some constant $C$ depending only on $p$ and $\gamma$.
    \item For $0<\gamma_1\leq \gamma_2$, $\viii{X}_{e^{\gamma_1}}\leq C\viii{X}_{e^{\gamma_2}}$ for constant $C$ depending only on $\gamma_1$ and $\gamma_2$.
    \item For random variables $X_1,\dots ,X_n$, $\viii{\max_{j\in[n]}X_j}_{\psi_{e^\gamma}}\leq C\psi^{-1}_{e^\gamma}(n)\max_{j\in[n]}\viii{X_j}_{e^{\gamma}} $ for some constant $C$ only depending on $\gamma$.
\end{enumerate}
\end{lemma}
\begin{proof}
For part (a), it is straightforward to verify that $x\mapsto \exp(x^\gamma)-1$ is convex on the interval $[((1-\gamma)/\gamma)^{1/\gamma}\lor 0,\infty)$. Also, $\psi_{e^\gamma}$ is continuous at $z_\gamma$ (hence continuous on $\R_+)$. Therefore, convexity follows as the left derivative of $\psi_{e^\gamma}$ is no larger than the right derivative at $a_\gamma$ since, by definition, $K_\gamma\leq \gamma a_\gamma^{\gamma-1} \exp(a^\gamma_\gamma)$.

For part (b), we have $x^p = x^p 1 = f(x^p) + f^*(1)$ for any $x\in\R_+$  and $f:\R_+\to \R\cup \{\pm\infty\}$, where $f^*$ denote the convex conjugate of $f$. Take $f(x)=\psi_{e^\gamma}(x^{1/p})$,  $x = X(\omega)/\viii{X}_{e^\gamma}$ (where we assume $\viii{X}_{e^\gamma}>0$, otherwise the result is trivial because $X=0$ a.s) and take expectation with respect to the law of $X$ to obtain
\[\frac{\E|X|^p}{\viii{X}_{e^\gamma}^p}\leq \E\psi_{e^\gamma}\left(\frac{X}{\viii{X}_{e^\gamma}}\right) + f^*(1)\leq  1 + f^*(1),\]
where $\psi^*(1)=\sup_{x\geq 0}\{x - \psi_{e^\gamma}(x^{1/p})\}<\infty$ and only depends on $p$ and $\gamma$, hence we might take $C= (1+\psi^*(1))^{1/p}$ to conclude.

For (d), we have that $\psi_e^\gamma$ is convex (by part (a)), nondecreasing, nonzero function vanishing at the origin. Note that $x^\gamma + y^\gamma\leq 2 (xy)^\gamma$ for $x,y\geq 1$, then $\frac{\psi_{e^\gamma}(x)\psi_{e^\gamma}(y)}{\psi_{e^\gamma}(2^{1/\gamma} xy)} \leq  \frac{\exp(x^\gamma + y^\gamma)}{\exp(2(xy)^\gamma)}\leq 1$ for $x,y\geq a_\gamma\geq 1$. The result then follows from Lemma 2.2.2 in \cite{VW1996}.

\end{proof}

\begin{lemma}\label{L:Orlicz_equivalence} If $\viii{X}_{e^\gamma}<\infty$ for $\gamma>0$ , there are constants $C_1>0$ and $C_2>0$ such that
\[ \P(|X|>x)\leq C_1\exp[-(x/C_2)^\gamma]\qquad x>0 .\]
In particular,  if $0<\viii{X}_{e^\gamma}<\infty$, we might take $C_1 = 2 + \exp \big[(a_\gamma/\viii{X}_{e^\gamma})^\gamma\big]\1_{0<\gamma<1}$ and  $C_2 = \viii{X}_{e^\gamma}$, where $a_\gamma$ is defined in Lemma \ref{L:Orlicz_basic}(a).

Conversely,  if there are constants $C_1>0$ and $C_2>0$ such that $\P(|X|>x)\leq C_1\exp[-(x/C_2)^\gamma]$ for $x>0$, then
\[\viii{X}_{e^\gamma}\leq
\begin{cases}
 C_2\big[(1+2C_1)^{1/\gamma}\lor 2K_\gamma C_1 \Gamma(1/\gamma)/\gamma\big];&0< \gamma <1\\
 C_2(1+C_1)^{1/\gamma}; & \gamma\geq 1,
\end{cases}
\]
where $\Gamma(\cdot)$ denotes the Gamma function.
\end{lemma}
\begin{proof}  If $\viii{X}_{e^\gamma}=0$ then $X=0$ a.s and the inequality holds for any choice of $C_1,C_2>0$.  For the case  when $0<\viii{X}_{e^\gamma}<\infty$ we have, by Markov inequality and the fact that $x\mapsto \exp \big[(x/\viii{X}_{e^\gamma})^\gamma\big]$ is non-decreasing
\begin{align*}
    \P(|X|\geq x) &=\P\left(\exp \big[(|X|/\viii{X}_{e^\gamma})^\gamma\big]\geq\exp \big[(x/\viii{X}_{e^\gamma})^\gamma\big]\right)\\
    &\leq \exp \big[-(x/\viii{X}_{e^\gamma})^\gamma\big]\E \exp \big[(|X|/\viii{X}_{e^\gamma})^\gamma\big].
\end{align*}
Using Lemma \ref{L:Orlicz_basic}(a), we have
\begin{align*}
\E \exp \big[(|X|/\viii{X}_{e^\gamma})^\gamma\big] &=\E \exp \big[(|X|/\viii{X}_{e^\gamma})^\gamma\big]\1_{|X|<a_\gamma}
+ \E \exp \big[(|X|/\viii{X}_{e^\gamma})^\gamma\big]\1_{|X|\geq a_\gamma}\\
&\leq   \exp \big[(a_\gamma/\viii{X}_{e^\gamma})^\gamma\big]\1_{0<\gamma<1}+ \E\psi_{e^\gamma}(|X|/\viii{X}_{e^\gamma}) + 1\\
&\leq \exp \big[(a_\gamma/\viii{X}_{e^\gamma})^\gamma\big]\1_{0<\gamma<1}+ 2.
\end{align*}
Combine the last two displays to obtain the first result.

For the converse, we have for $c>0$, by Fubini's Theorem, 
\begin{align*}
\E \exp(|X/c|^\gamma) -1 &= \int\int_0^{|x|^\gamma}c^{-\gamma}\exp (c^{-\gamma} y)  dy \P(dx)\\
&=c^{-\gamma} \int_0^\infty\P(|X|\geq x^{1/\gamma})\exp(c^{-\gamma} x) dx.
\end{align*}
Since $\P(|X|>x)\leq C_1\exp[-(x/C_2)^\gamma]$,  for $c>C_2$,  we have
\begin{align*}
\E \exp[(|X/c|)^\gamma] -1 &\leq  c^{-\gamma} C_1\int_0^\infty\exp\left[-x(C_2^{-\gamma} - c^{-\gamma})\right]dx\leq \frac{c^{-\gamma} C_1}{C_2^{-\gamma}-c^{-\gamma}} = \frac{C_1}{(c/C_2)^{\gamma}-1}.
\end{align*}
Also
\begin{align*}
\E|X| = \int_0^\infty\P(|X|>x)dx\leq  C_1\int_0^\infty\exp[-(x/C_2)^\gamma]dx = \tfrac{C_1C_2}{\gamma}\Gamma(1/\gamma).
\end{align*}

Therefore using the last two displays and Lemma \ref{L:Orlicz_basic}(a), we have, for $c>C_2$,
\begin{align*}
\E\psi_{e^\gamma}(|X|/c) &\leq \frac{K_\gamma}{c}\E|X|\1_{0<\gamma<1}+ \E \exp[(|X|/x)^\gamma] -1\\
&\leq \frac{K_\gamma C_1C_2}{c\gamma}\Gamma(1/\gamma)\1_{0<\gamma<1} + \frac{C_1}{(c/C_2)^{\gamma}-1}.
\end{align*}
For $\gamma\geq 1$ the right hand side is less or equal than $1$ for $c\geq C_2(1+C_1)^{1/\gamma}$ hence $\viii{X}_{e^\gamma}\leq C_2(1+C_1)^{1/\gamma}$.  For $0<\gamma<1$,  the right hand side is less or equal than $1$ for $c\geq C_2(1+2C_1)^{1/\gamma}\lor 2K_\gamma C_1 C_2 \Gamma(1/\gamma)/\gamma$ then $\viii{X}_{e^\gamma}\leq  C_2(1+2C_1)^{1/\gamma}\lor 2K_\gamma C_1 C_2 \Gamma(1/\gamma)/\gamma$.
\end{proof}

\begin{lemma}\label{L:CS_exp} If $\viii{X}_{e^\gamma}<\infty$ and $\viii{Y}_{e^\gamma}<\infty$ for $\gamma>0$ then $\viii{XY}_{e^{\gamma/2}}<\infty$ for $\left(\viii{X}_{e^\gamma}^2\lor \viii{Y}_{e^\gamma}^2\right)$. In particular, $\viii{XY}_{e^{\gamma/2}}\leq C\left(\viii{X}_{e^\gamma}^4\lor \viii{Y}_{e^\gamma}^2\right)$ where $C=5^{2/\gamma}$ for $\gamma\geq 1$; and $C=(1+2C_1)^{2/\gamma}\lor 2K_{\gamma/2} C_0 \Gamma(2/\gamma)/\gamma$ with $C_0 = 2(2 + \exp \big[(a_\gamma/(\viii{X}_{e^\gamma}\land\viii{Y}_{e^\gamma} ))^\gamma\big] )$ for $\gamma\in(0,1)$ provided that $\viii{Y}_{e^\gamma}\land \viii{Y}_{e^\gamma}>0$.
\end{lemma}
\begin{proof}
If $\viii{X}_{e^\gamma}=0$ or $\viii{Y}_{e^\gamma}=0$ then $XY=0$ a.s and the inequality holds trivially.  Assume that  $0<\viii{X}_{e^\gamma}<\infty$ and  $0<\viii{Y}_{e^\gamma}<\infty$. From  Lemma \eqref{L:Orlicz_equivalence} we have for $x>0$
\begin{align*}
\P(|X|>x)\leq C_X\exp[-(x/\viii{X}_{e^\gamma})^\gamma]\\
\P(|Y|>x)\leq C_Y\exp[-(x/\viii{Y}_{e^\gamma})^\gamma],
\end{align*}
for positive constants $C_X$ and $C_X$. Then, by the union bound,
\begin{align*}
\P(|XY|\geq x)&\leq \P(|X|\geq \sqrt{x})+ \P(|Y|\geq \sqrt{x})
\\
&\leq C_X \exp(-x^{\gamma/2}/\viii{X}_{e^\gamma}^\gamma)+ C_Y \exp\big[-x^{\gamma/2}/\viii{Y}_{e^\gamma}^\gamma)\\
&\leq 2(C_X\lor C_Y)\exp\left[-\left(\frac{x}{\viii{X}_{e^\gamma}^2\lor \viii{Y}_{e^\gamma}^2}\right)^{\gamma/2}\right].
\end{align*}
Apply once again  Lemma \eqref{L:Orlicz_equivalence} in the other direction to conclude.
\end{proof}

\end{supplement}

\end{document}